\newif\ifdraft \drafttrue
\newif\iffull \fulltrue

\makeatletter \@input{tex.flags} \makeatother

\documentclass[11pt]{article}

\usepackage{graphicx}

\usepackage{tikz}  
\usetikzlibrary{arrows}
\pgfdeclarelayer{bg}    
\pgfsetlayers{bg,main} 
\usepackage{subcaption}

\usepackage{algorithm}
\usepackage[noend]{algorithmic}
\usepackage{url}
\usepackage{hyperref}

\usepackage{authblk}
\newcommand*\samethanks[1][\value{footnote}]{\footnotemark[#1]}

\usepackage{fullpage}

\usepackage{amsmath, amssymb, amsthm}

\usepackage{mathtools}
\usepackage{verbatim}
\usepackage{color}
\usepackage{xcolor}

\usepackage{multicol}

\definecolor{DarkGreen}{rgb}{0.1,0.5,0.1}
\definecolor{DarkRed}{rgb}{0.5,0.1,0.1}
\definecolor{DarkBlue}{rgb}{0.1,0.1,0.5}

\newcount\Comments  
\newcommand{\kibitz}[2]{\ifnum\Comments=1{\color{#1}{#2}}\fi}

\usepackage{color}
\definecolor{english}{rgb}{0.0, 0.5, 0.0}
\newcount\CommentsA  
\newcommand{\kibitzA}[2]{\ifnum\CommentsA=1{\color{#1}{#2}}\fi}

\usepackage{xspace}
\usepackage{cleveref}
\usepackage{thmtools, thm-restate}
\usepackage[numbers,sort&compress]{natbib}

\usepackage[mathscr]{euscript}
 \let\mathscr\relax
\usepackage[scr]{rsfso}

\newcommand\cB{\mathcal{B}}

\newcommand\cP{\mathcal{P}}

\newcommand\cR{\mathcal{R}}

\newcommand\cN{\mathcal{N}}

\newcommand\cV{\mathcal{V}}
\newcommand\cZ{\mathcal{Z}}

\newcommand{\vect}[1]{\boldsymbol{\mathbf{#1}}}

\def\epsilon{\varepsilon}
\DeclareMathOperator{\OPT}{OPT}
\DeclareMathOperator{\GAS}{GAS}

\DeclareMathOperator*{\argmin}{\mathrm{argmin}}
\DeclareMathOperator*{\argmax}{\mathrm{argmax}}

\newtheorem*{theorem*}{Theorem}
\newtheorem*{observation*}{Observation}
\declaretheorem[
  name=Theorem,
  refname={theorem, theorems},
  Refname={Theorem, Theorems}]{theorem}
\declaretheorem[
  name=Lemma,
  refname={lemma, lemmas},
  Refname={Lemma, Lemmas}]{lemma}

\declaretheorem[
  name=Definition,
  refname={definition, definitions},
  Refname={Definition, Definitions}]{definition}

\newtheorem{proposition}[theorem]{Proposition}

\title{Differential Liquidity Provision in Uniswap v3 and Implications for Contract Design}

\author[1]{Zhou Fan\thanks{These authors contributed equally to this work.}\textsuperscript{,}\thanks{ \{zfan, fjmarmol, he\_sun, xintongw\}@g.harvard.edu }\textsuperscript{,}}
\author[1,2]{Francisco Marmolejo-Coss\'{i}o\samethanks[1]\textsuperscript{,}\samethanks[2]\textsuperscript{,}}
\author[1]{Ben Altschuler\thanks{baltschuler@college.harvard.edu}\textsuperscript{,}}
\author[1]{He Sun\samethanks[2]\textsuperscript{,}}
\author[1]{\\ Xintong Wang\samethanks[2]\textsuperscript{,}}
\author[1]{David C. Parkes\thanks{parkes@eecs.harvard.edu; also  DeepMind}\textsuperscript{,}}
\affil[1]{Harvard University}
\affil[2]{IOHK}
\date{}

\begin{document}
\maketitle              
\begin{abstract}

Decentralized exchanges (DEXs) provide a means for users to trade pairs of assets on-chain without the need for a trusted third party to effectuate a trade. Amongst these, constant function market maker DEXs such as Uniswap handle the most volume of trades between ERC-20 tokens. With the introduction of Uniswap v3, liquidity providers can differentially allocate liquidity to  trades that occur within specific price intervals. In this paper, we formalize the profit and loss that liquidity providers can earn when providing specific liquidity allocations to a  v3 contract. We give a convex stochastic optimization problem for computing optimal liquidity allocation for a liquidity provider who holds a belief on how prices will evolve over time and use this to study the design question regarding how v3 contracts should partition the price space for permissible liquidity allocations. Our results show that making a greater diversity of price-space partitions available to a contract designer can simultaneously benefit both liquidity providers and traders. 
\end{abstract}


\section{Introduction}


A key application in  decentralized finance (DeFi) is that of {\em decentralized exchanges} (DEXs).
DEXs offer smart contracts that allow users to trade tokens without the need of a trusted-third party to effectuate the trade. A  benefit of such an implementation is that it avoids the hacking risks that can be suffered by centralized, off-chain exchanges. Amongst DEXs, there are two prevailing algorithmic paradigms for executing a trading contract: {\em order book DEXs} and {\em automated market maker} (AMM) DEXs. Order book DEXs  maintain a list of buy and sell orders from users at distinct prices for a given pair of assets to be traded, with these orders received, matched and executed. An AMM DEX, on the other hand, will always quote a buy and sell prices for any trade, where these prices depend on the contract's assets and the rules of the AMM.

AMMs are the most common form of DEXs, amongst which {\em Uniswap} contracts handle a substantial proportion of trading volume. Uniswap contracts serve as {\em constant function market makers} (CFMMs), which are a popular kind of AMM design. In CFMMs, the contract computes the price of buying and selling between two assets by preserving a functional invariant of its existing liquidity reserves. 
To briefly describe the operation of a CFMM, let $x$ and $y$ denote the liquidity reserves the contract has of each of two assets, say token $A$ and  $B$ respectively. The trading invariant can be expressed as $F(x,y) = C$, for a given function $F$ and constant $C$. A trader who wishes to sell $\Delta x>0$ of token $A$ must send $\Delta x$ units of $A$ tokens to the contract, and the amount of token $B$ they receive is the value, $\Delta y>0$, such that the functional invariant is maintained, i.e. $F(x + \Delta x, y - \Delta y) = C$.  The quantity $\Delta y/\Delta x$ represents the average, per-unit price of  token $A$ for the  trade (in terms of token $B$). As $\Delta x \rightarrow 0$, this ratio gives the instantaneous price of token $A$ in terms of token $B$ for a contract with with a bundle of assets given by $(x,y)$.

 {\em Liquidity providers} (LPs)  provide assets to the contract and enable these trades. An LP lends the contract a bundle  of  $A$ and $B$ tokens, which is traded against as the relative price of token $A$ (or equivalently token $B$) changes. Liquidity provision is rewarded by means of {\em transaction fees} on trades. In May 2021, {\em Uniswap v3} introduced  a new family of AMMs where LPs can differentially allocate liquidity to a v3 contract  \cite{adams2021uniswap}. v3 contracts allow users to allocate liquidity to be used for trades in a specific price interval. The fees associated with a trade are  shared proportionally amongst the LPs who provide liquidity on intervals that contain the associated price change. 
 With this change, LPs can use the same capital to obtain more aggressive liquidity allocations around tight price intervals, and thus potentially earn more fees  (at the risk of losing out on fees all together if prices leave a particular price interval). Another important consideration for LPs is that
 trades that occur on the contract and price changes change  the composition of capital that an LP has a claim to. 
 If prices returns to where they were when liquidity was first provided, then the LP can withdraw its liquidity in the same token quantities as initially lent. Otherwise, the  bundle of tokens to which an LP has claim has a lower value when evaluated at the new price. This phenomenon is referred to as the {\it impermanent loss} of a liquidity allocation, and is a crucial consideration for liquidity provision.

\subsection{Our Contributions}

Given the rapid increase in DEX usage,\footnote{At the  time of writing, the daily trade volume in Uniswap v2 and v3 contracts is approximately 96 Million USD and 1.14 Billion USD respectively. In addition, the total liquidity locked by users to facilitate trades in v2 and v3 contracts is 1.51 Billion USD and 4.57 Billion USD respectively. v2: \href{https://v2.info.uniswap.org/home}{\url{https://v2.info.uniswap.org/home}}, v3: \href{https://v3.info.uniswap.org/home}{\url{https://v3.info.uniswap.org/home}}.} and the associated questions around contract design, it is important to understand the decision problem facing an LP. This paper builds off existing work to provide a new  theoretical and empirical understanding of LP behavior, and in turn to provide concrete design recommendations for Uniswap v3 contracts.

We begin by providing an overview of Uniswap dynamics for v2 and v3 contracts in Section \ref{sec:uniswap-overview}. In Section \ref{sec:LP-PnL}, we provide an expression for LP profit and loss over a sequence of price changes in the contract. In Section \ref{sec:riskneutral}, we formulate a convex stochastic optimization problem for the problem facing a  risk-averse LP who seeks to optimize its profit and loss over a finite time horizon, given  beliefs as to how prices will evolve (we also give an linear programming formulation for a risk neutral provider). 
The optimization problem fundamentally relies on the simplified assumption that prices follow a stochastic process that is  independent of liquidity provision. More specifically, we assume a price-based, stochastic flow of non-arbitrage trades that can shift the spot price of a Uniswap contract away from an assumed market price, as well as arbitrage trades that keep the spot price of the Uniswap contract close to the market price (how close depends on the underlying transaction fees of the contract). Remarkably, this assumption of an exogeneous price dynamic makes the problem facing a single LP that of a (single-agent) decision-theoretic problem rather than a game-theoretic problem.  
We leave to future work a more sophisticated model of trade flow with resulting price sequences that also depend on liquidity within a Uniswap contract.

In Section \ref{sec:computational-results}  we apply this model of LP behavior to the  design question of how v3 contracts should partition the price space in affording different liquidity allocations by LPs. 
There is a trade off to be struck: finer partitions provide more flexibility for LPs and thus a better ability to optimize return on investment, but  increases the gas cost to traders as as a result of higher computational overhead  in determining trade dynamics over more complicated liquidity allocations.
Through experiments that are calibrated to real price data \footnote{As further specified in Appendix \ref{append:MLE}, we use per-minute price data between Ethereum (ETH) and Bitcoin (BTC) as well as between Ethereum (ETH) and Tether (USDT) for the month of February 2020, obtained from Binance, to calibrate beliefs LPs may have regarding external prices in our experiments},  
we provide empirical evidence that a greater diversity of  price partitions available to a contract designer can benefit both LPs and traders. 
We study a wide family of log-linear price partitions that generalize the current partition of price space used by v3 contracts, and show that our results  are robust to a wide array of assumptions regarding token price volatility, drift, as well as gas prices. 
Furthermore, we also study how different degrees of risk-aversion impact optimal liquidity allocations for LPs. We find that as LPs become more risk-averse, they spread their liquidity across larger price ranges, as this helps reduce the variance in their profit and loss over extended price sequences. Moreover, we also see that increased volatility in price sequences also causes risk-averse LPs to spread their liquidity allocations further over a larger price range.

\subsection{Related Work}

This paper extends a growing body of literature around the incentives of liquidity provision  from on-chain implementations of CFMMs. Most closely related is  Neuder et al.~\cite{NRMP-strategic-LP}, who study strategic liquidity provision in Uniswap v3. As in our work, they assume that LPs holds beliefs that contract prices will evolve according to a Markov chain, and they provide a method for LPs to maximize fees earned in the steady state of the chain. One major difference in the present work is that we  model  impermanent loss in addition to fees earned, which is a first-order consideration for LPs. A second major difference is that we study the question of the design of v3 contracts in regard to the partitioning of the price space for potential liquidity allocations, a problem which to our knowledge has not been considered before in the CFMM literature. 

With regards to Uniswap v2, Angeris et al.~\cite{angeris2019analysis}  provide theoretical evidence showing that under reasonable conditions v2 exchanges closely track reference markets. Their work involves modeling potential arbitrage opportunities for traders, even as they are faced with trading fees in CFMMs. In addition, their work  provides expressions for potential LP earnings under simple price changes in v2 contracts. This work is extended in Angeris and Chitra~\cite{angeris2020improved}, who provide similar guarantees for a more general class of CFMMs, as well as in Angeris et al.~\cite{angeris2021constant}, where similar results and techniques are extended to CFMMs supporting multi-asset trades. In Angeris et al.~\cite{angeris2020does}, the authors also study the implications of the curvature in reserve curves for traders and LPs, providing concrete tradeoffs for when high and low curvature regimes favor each of these two classes of agents. All these results focus on v2 contracts however, and not on the richer space of liquidity allocations presented by v3, nor on the design questions related to v3 contracts.

A related branch of work also studies  how LPs can replicate the payoffs of financial derivatives. Evans~\cite{evans2020liquidity} focuses on geometric mean market makers (CFMMs with functional invariants which preserve a weighted geometric mean of assets held in the contract), and more general results are shown in a series of papers by Angeris et al.~\cite{angeris2021replicating} and~\cite{angeris2021replicatingoracles} for a larger class of financial derivatives. Finally, Capponi and Jia \cite{capponi2021adoption} study the adoption of Uniswap v2  by considering a sequential game between LPs and traders, and  use a similar stochastic model for price changes as that adopted in the present paper.

The liquidity provision problem in Uniswap also shares similarities with  market making  in  limit order  markets. Avellaneda \cite{Avellaneda2007}, for example, sets up the market making problem of posting bid and ask quotes in a maximized exponential utility framework and solves it in a two step approach 
where the market-maker first computes an indifference valuation of the traded stock that incorporates inventory risk, then decides where to place the bid and ask quotes based on a probabilistic model of orders being executed at different prices given the current mid-price. 
%
Our study of LP revenue as a function of  the fidelity in setting liquidity allocations  is related, also, to the tick size design  problem in traditional limit order markets, which has attracted considerable  attention from both an economics and regulatory perspective. European exchanges, for example, compete directly on the minimum pricing increment in the limit order book to capture market shares of quoted and executed volumes~\cite{foley2021}. In 2016, the US Securities and Exchange Commission (SEC) launched the {\em Tick Size Pilot Program} to assess the impact of increasing tick sizes on the market quality for illiquid stocks~\cite{Chung2019}.

\section{Overview of Uniswap}
\label{sec:uniswap-overview}
\setcounter{equation}{0}

A core functionality of Uniswap contracts is to provide a family of CFMM DEXs for trading between ERC-20 token pairs. In this section, we provide a brief overview of the mechanics of Uniswap v2 and Uniswap v3 to lay the groundwork  for subsequent sections. 

\subsection{Uniswap v2: Providing Liquidity  at all Prices}
\label{sec:v2-description}

A Uniswap v2 contract  facilitates trades between a pair of ERC-20 tokens, say Token $A$ and Token $B$. LPs provide liquidity in the form of bundles of $A$ and $B$ tokens to the contract. Let $x>0$ and  $y>0$ denote the number of $A$ tokens and $B$ tokens, respectively, provided by LPs to the contract. We refer to $(x,y)$ as the {\it v2 contract state}. 

\subsubsection{v2 Reserve Curve}
As trades are made, Uniswap v2 maintains a functional invariant on the liquidity held in the contract. The number of assets in the contract must satisfy $F(x,y) = C$, for function $F$ and constant $C$. 
We refer to the collection of all states $(x,y) \in (\mathbb{R}^+)^2$ that satisfy $F$ as the {\it reserve curve}. In particular, if a trader sells $\Delta x > 0$ units of  $A$, they receive $\Delta y > 0$ units of  $B$, such that $F(x + \Delta x, y - \Delta y) = C$ is maintained.  For Uniswap v2, $F(x,y) = x  y$, and the reserve curve is as visualized in Figure~\ref{fig:reals}. 

Let $L=\sqrt{xy}$ denote a quantity that we refer to as the number of {\em liquidity units} in a contract. This provides a convenient, single-valued measure of how much liquidity is held in the contract. With this, the reserve curve is the set of states $(x',y')$ that satisfy $x'y' = L^2$. In this sense, liquidity $L$ controls the set of allowable trades, including  the way trades are priced.
\begin{definition}[Uniswap v2 Reserve Curve]
\label{def:v2-reserve-curve}
For $L>0$ units of liquidity held in the contract, we let $\mathcal{R}^{(2)}(L)$ denote the  {\em v2 reserve curve at liquidity $L$} between token $A$ and token $B$, with

\begin{align}
\mathcal{R}^{(2)}(L) = \{ (x,y) \in \mathbb{R}^+ \times \mathbb{R}^+  \mid  xy = L^2\}.
\end{align}

\end{definition}

\begin{figure}
    \centering
    \includegraphics[width=0.4\linewidth]{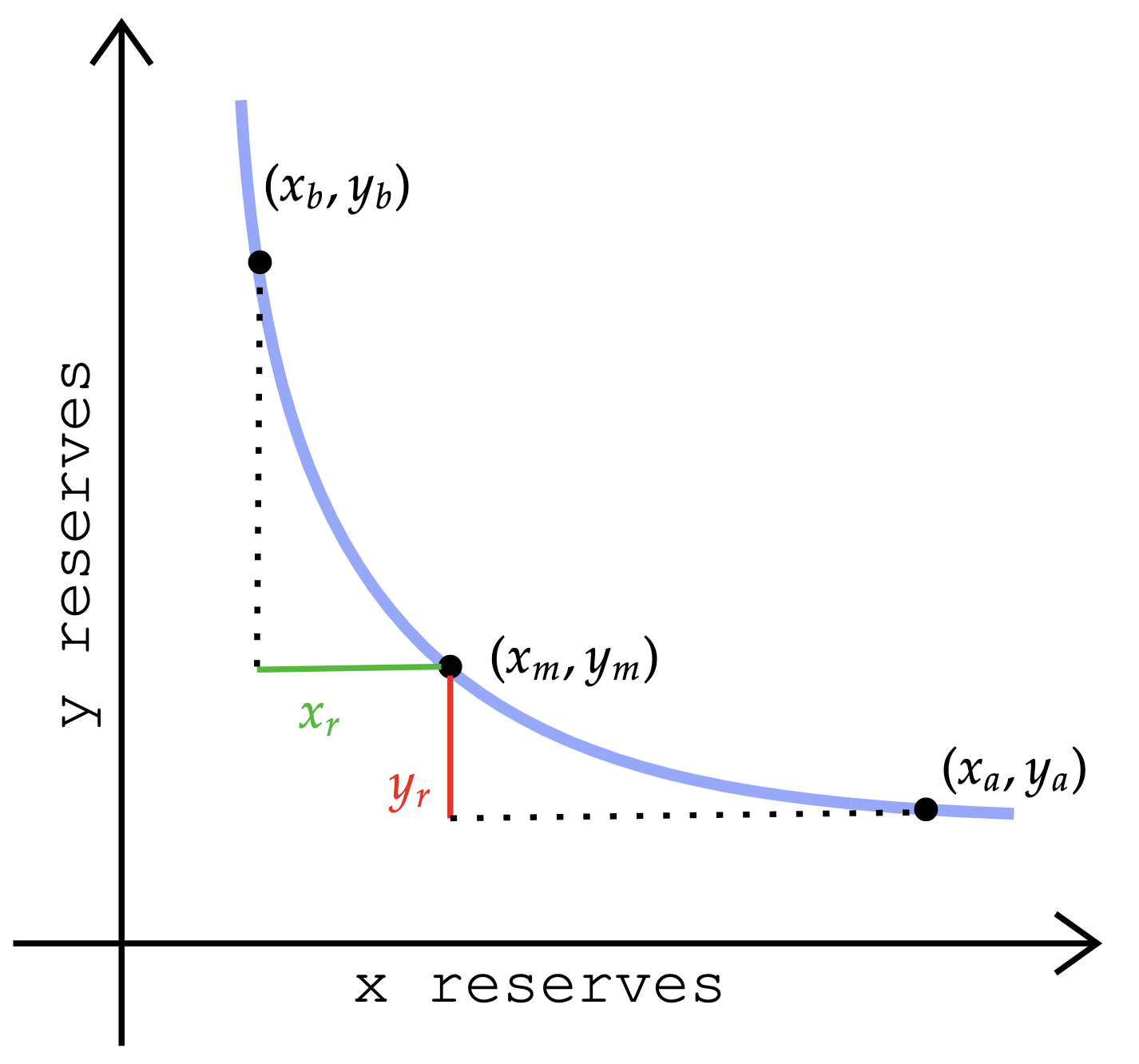}
    \caption{The $x,y$ coordinates on the Uniswap v2 $xy=L^2$ reserve curve, with points illustrated for prices $P_a < P_m < P_b$. Buying $x_r$ units of token $A$ requires sending $(y_b - y_m)$ units of token $B$ 
      to the contract and increases the price from $P_m$ to $P_b$. Buying $y_r$ units of token $B$
    requires sending $x_a - x_m$ units of token $A$ to the contract   and  decreases the price from $P_m$ to $p_a$. 
    \label{fig:reals}}
\end{figure}

\subsubsection{Contract Price}
The infinitesimal price $P$ of token $A$ in terms of token $B$ is the {\it contract price}, and depends on the contract state $(x,y)$.
%
%
For an amount of liquidity $L$, the amount of token $B$ held in the contract depends functionally on the amount of token $A$, with $y=f_L(x)$ for $f_L(x)=L^2/x$. With this,   the instantaneous price $P$ of token $A$ given liquidity $L$ is  $-\mathrm{d}f_L(x)/\mathrm{d}x$ and we have,
\begin{align}
P = -\frac{\mathrm{d}f_L(x)}{\mathrm{d}x} = \frac{L^2}{x^2} = \frac{y}{x}. 
\end{align}

The constant-product function can also be re-written as $(x/L)(y/L)= 1$, from which we see that when $L$ is higher, the contract price, $P = (y/L)/(x/L)$ slips less in response to trades; i.e., changing the price from $P$ to $P'$ requires a larger change in the contract state, in terms of $\Delta x$ or $\Delta y$, when the liquidity, $L$, is larger.
Based on the following correspondence, each point on reserve curve $\mathcal{R}^{(2)}(L)$ can be  identified by  $(x,y)$ assets of token $A$ and token $B$, or equivalently $(L,P)$  and the amount of liquidity  and the contract price.
\begin{proposition}\label{prop:pricecoords-v2}
If a Uniswap v2 contract has $L$ units of liquidity and a contract price $P$, it must be the case that the contract state is given by $(x,y)$ such that:
\begin{align}
    x = \frac{L}{\sqrt{P}}, \qquad y = L \sqrt{P}.
\end{align}
\end{proposition}
\begin{proof}
Immediate, by verifying  $(x,y) \in \mathcal{R}^{(2)}(L)$ and   $P = y/x$.
\end{proof}

A consequence of  Proposition~\ref{prop:pricecoords-v2}  is that a contract's state can  be represented by the bundle of $A$ and $B$ tokens the contract holds as liquidity, $(x,y)$, or equivalently by how much liquidity is held in the contract, and at what contract price, i.e. $(L,P)$. We call the former the {\it token-bundle} contract state and the latter the {\it liquidity-price} contract state.
For a given liquidity-price state, $(L,P)$, we  refer to the corresponding token-bundle contract state as the {\em  v2 bundle value of $L$ units of liquidity at price $P$} and denote this as $\mathcal{V}^{(2)}(L,P) = (L/\sqrt{P}, L \sqrt{P})$.
From Proposition \ref{prop:pricecoords-v2}, the  bundle value of $L$ units of liquidity is linear in $L$,  and we also write $\mathcal{V}^{(2)}(P) = \mathcal{V}^{(2)}(1,P)$ for the bundle value of 1 unit of liquidity, so that $\mathcal{V}^{(2)}(L,P) = L \cdot \mathcal{V}^{(2)}(P)$.  

\subsubsection{Providing and Removing Liquidity}

LPs can add liquidity to the contract or remove liquidity  they own under the invariant that the price $P$  remains unchanged. For example, an LP who wants to add $L'$ units of liquidity to the contract with current liquidity-price state of $(L,P)$ and token-bundle state of $(x,y)$ must deposit $\mathcal{V}^{(2)}(L',P) = (x',y')$, which is a bundle of $A$ tokens and $B$ tokens. The effect is to change the token-bundle state to $(x+x',y+y')$ and the liquidity-price state to $(L + L',P)$. 
Similarly, an LP with claim to $L'$ units of liquidity  may remove a token bundle consisting of $\mathcal{V}^{(2)}(L',P) = (x',y')$, from the contract. The resulting token-bundle state is $(x - x',y-y')$ and the  liquidity-price state is $(L-L',P)$. This is formalized in the following proposition, whose proof is deferred to Appendix \ref{appendix:sec-2-proofs}.
\begin{restatable}{proposition}{Firstprop}
\label{prop-add-liquidity}
Suppose that $(x,y) \in \mathcal{R}^{(2)}(L)$ with $P =y/x$, and $(x',y') \in \mathcal{R}^{(2)}(L')$ such that $P' = y'/x'= P$. Then  $(x+x',y+y') \in \mathcal{R}^{(2)}(L+L')$, and in addition, $(y'+y)/(x+x') = P$. If $L' < L$, then $(x-x',y-y') \in \cR^{(2)}(L-L')$ and $(y-y')/(x-x') = P$.
\end{restatable}

\subsubsection{Trading Fees}

 When a trade occurs in Uniswap v2, a portion goes to the LPs as fees. Suppose a trader sends $\Delta x$ units of token $A$  to purchase units of token $B$ (the case where the roles of $A$ and $B$ are reversed is identical).  In this case, $\gamma \Delta x$ units of  $A$  are skimmed as trade fees and allocated to LPs  in proportion  to how much liquidity they have contributed,  for  {\em fee rate},  $\gamma \in (0,1)$.   
More specifically, let us assume that there 
are $d \in \mathbb{N}$ LPs, such that the $j$-th LP has provided $L_j$ units of liquidity and the total liquidity in the contract is given by $L = \sum_{j=1}^d L_j$. For the given trade, the $j$-th LP  receives $( \frac{L_j}{L} )\gamma \Delta x$ units of token $A$. The remaining $(1-\gamma) \Delta x$ units of $A$ sent from the trader are used to move along the reserve curve, shifting  the contract's token-bundle state to $(x',y')$, where $x' = x + (1-\gamma) \Delta x$, and $y'$ is such that $(x',y') \in \cR^{(2)}(L)$, and the trader receives $y' - y$ units of token $B$ in return for this trade. In Appendix \ref{appendix:example-v2-dynamics} we provide an in-depth example of v2 trade dynamics.

\subsection{Uniswap v3: Concentrated Liquidity Provision}
\label{sec:v3-description}

In v2 contracts, LPs provide assets to the contract to facilitate trades at any contract price and an LP's contributions to the contract is measured in units of v2 liquidity. In v3 contracts, LPs are given the option to allocate liquidity to be used only for trades in a finite price interval $[a,b]$. As we will see, we use an analogous, single-valued measure of an LP's contribution to prices in this range, which we refer to as {\em units of $[a,b]$-liquidity}. 

At a high level, providing $L$ units of $[a,b]$-liquidity has two key consequences regarding transaction fees earned by an LP: 

(i) this liquidity only earns fees when the contract price is in $[a,b]$, where fees are split proportionally amongst all LPs who have allocated liquidity at intervals including the contract price, and 

(ii) the value (in terms of either units of token $A$ or token $B$) of the bundle of tokens equivalent to $L$ units of $[a,b]$-liquidity is smaller as the interval $[a,b]$ becomes smaller.

Combining these two points, if an LP has a certain initial capital in terms of tokens $A$ and  $B$, they can potentially obtain more liquidity units over smaller intervals, and hence increase the fees they potentially accrue. However, this comes at the risk of not earning fees  when prices exit the liquidity's given interval. 
As with v2 contracts, when an LP provides $L$ units of $[a,b]$-liquidity, they send a bundle of $A$ and $B$ tokens equivalent to the bundle value of $L$ units of $[a,b]$-liquidity at the given contract price (albeit with different expressions for bundle value). As the token price changes, the LP may face losses in the redeemable token bundle value of $L$ units of $[a,b]$-liquidity relative to the value they would have accrued simply holding their initial liquidity before obtaining the liquidity from the contract. This  phenomenon is called {\it impermanent loss} (this phrasing recognizing that this loss disappears if the prices return to their original values). In Section~\ref{sec:LP-PnL}, we provide more details on these considerations, where we note that as the relevant interval decreases for $[a,b]$-liquidity the potential impermanent loss to an LP increases.

\subsubsection{\texorpdfstring{$[a,b]$}{[a,b]}-Liquidity}
In Uniswap v3, an LP provides liquidity for a specific price interval, denoted $[a,b]$, and we quantify a user's contribution in terms of {\it $[a,b]$-liquidity units}. We introduce the following notation,
\begin{align}
    \Delta^x_{P,P'} = \frac{1}{\sqrt{P'}} - \frac{1}{\sqrt{P}}, \qquad \Delta^y_{P,P'} = \sqrt{P'} - \sqrt{P},
\end{align}
where $\Delta^x_{P,P'}$ and $\Delta^y_{P,P'}$ represent the change in $x$ and $y$ values, respectively, along the  v2 reserve curve for unit liquidity, $\mathcal{R}^{(2)}(1)$, as the contract price changes from $P$ to $P'$.
We provide a map from $[a,b]$-liquidity to a token position, for Uniswap v3, that is analogous to the v2  function that gives the bundle value.
\begin{definition}[${[a,b]}$-Liquidity value]
\label{def:a-b-value}
Suppose that the contract price is  $P \in (0,\infty)$. We let $\mathcal{V}^{(3)}(L,a,b,P) \in \mathbb{R}^+ \times \mathbb{R}^+$ denote the {\em bundle value} of $L$ units of $[a,b]$-liquidity, which is  a bundle of $A$ and $B$ tokens respectively, and the number of tokens that is equivalent to liquidity $L$ on this interval at price $P$, with 
\begin{equation}
\mathcal{V}^{(3)}(L,a,b,P) = \left\{
        \begin{array}{ll}
            (L \Delta^x_{b,a}, 0)& \quad \mbox{if $P < a$} \\
            (0,L \Delta^y_{a,b}) & \quad \mbox{if $P > b$} \\
            (L\Delta^x_{b,P}, L\Delta^y_{a,P}) & \quad \mbox{if $P\in[a,b]$} 
        \end{array}
    \right.
\end{equation}
\end{definition}

As with the v2 value function,   this expression is linear in $L$, hence we also use the shorthand for $\mathcal{V}^{(3)}(a,b,P) = \mathcal{V}^{(3)}(1,a,b,P)$ so that $\mathcal{V}^{(3)}(L,a,b,P) = L \cdot \mathcal{V}^{(3)}(a,b,P)$. In particular,  an LP who wants to add $L'$ units of $[a,b]$-liquidity to a v3 contract, when the contract price is $P$,  must send the token bundle $\mathcal{V}^{(3)}(L',a,b,P)$ to the contract. An LP who wants to remove $L'$ units of $[a,b]$-liquidity will receive the token bundle $\mathcal{V}^{(3)}(L',a,b,P)$.

As mentioned above, as the interval $[a,b]$ becomes smaller, the bundle value of $L'$ units of $[a,b]$-liquidity falls; i.e., $L'$ units can be added with fewer  $A$ and $B$ tokens.
At the same time, fees for $[a,b]$-liquidity are limited to trades  at prices $P \in [a,b]$. 
In the converse, as the interval $[a,b]$ approaches the entire price interval, we recover the v2 liquidity value. The proof of the following is deferred to Appendix \ref{appendix:sec-2-proofs}.
\begin{restatable}{proposition}{Secondprop}
In the limit of an $[a,b]$-interval that approaches the entire price interval,
the value of Uniswap v3 liquidity is that of Uniswap v2 liquidity, i.e., 
\begin{align}
&\lim_{a \rightarrow 0,  b \rightarrow \infty} [\mathcal{V}^{(3)}(a,b,P)] = \mathcal{V}^{(2)}(P). 
\end{align}
\end{restatable}

\subsubsection{v3 Reserve Curves}
Suppose that a v3 contract has $L$ units of $[a,b]$-liquidity, in addition to other liquidity over the rest of the price space. In v2, trades move the contract state along the reserve curve, $\cR^{(2)}(L)$, and this keeps the total number of liquidity units in the contract  constant. The same holds for a v3 contract, albeit for trades for which the contract price remains in the interval. To explain this, we first define the reserve curve that corresponds to  $L$ units of $[a,b]$-liquidity. 
For this, we
define $\phi^A_{a,b}(z) = z + \frac{L}{\sqrt{b}}$ and $\phi^B_{a,b}(z) = z + L \sqrt{a}$, and  write $\phi_{a,b}(w,z) = (\phi^A_{a,b}(w),\phi^B_{a,b}(z))$.

\begin{definition}[Uniswap v3 reserve curve]
\label{def:v3-reserve-curve}
Given $L$ units of $[a,b]$-liquidity in Uniswap v3, we denote the {\em reserve curve} at $[a,b]$-liquidity $L$ by $\mathcal{R}^{(3)}(L,a,b)$, with 
\begin{align}
\mathcal{R}^{(3)}(L,a,b) = \{ (x,y) \in \mathbb{R}^+ \times \mathbb{R}^+ \mid \phi^A_{a,b}(x) \cdot \phi^B_{a,b}(y) = L^2 \}.
\end{align}
\end{definition}

For a given v3 reserve curve, $\mathcal{R}^{(3)}(L,a,b)$, we call $\mathcal{R}^{(2)}(L)$ the {\it virtual reserve curve} of the assets, appealing here to the v2 reserve curve. As the following proposition shows, the v3 reserve curve is  a positive affine transformation of the portion of $\mathcal{R}^{(2)}(L)$ that corresponds to token-bundle states with a contract price in $[a,b]$.
Furthermore,  prices are preserved for trades in this interval  since the transformation is positive affine.
An example of a v3 reserve curve and its virtual reserve curve is visualized in Figure \ref{fig:shift}.

\begin{restatable}{proposition}{Thirdprop}
\label{lemma:v3-virtual-affine-map}
Suppose that $(x,y) \in \mathcal{R}^{(3)}(L,a,b)$. It follows that  $\phi_{a,b}(x,y) \in \mathcal{R}^{(2)}(L)$, and if we let $P = \phi^B_{a,b}(y)/\phi^A_{a,b}(x)$, then $P \in [a,b]$ and $(x,y) = \mathcal{V}^{(3)}(L,a,b,P)$.  
\end{restatable}

The proof of Proposition \ref{lemma:v3-virtual-affine-map} can be found in Appendix \ref{appendix:sec-2-proofs}. It shows that the affine map $\phi_{a,b}$ naturally maps $\phi_{a,b}$ to $\cR^{(2)}(L)$. In this regard, we say a token bundle $(x,y)\in \cR^{(3)}(L,a,b)$ has a contract price given by $P = \frac{\phi^B_{a,b}(y)}{\phi^A_{a,b}(x)}$. 

\begin{figure}
    \centering
    \includegraphics[width=0.4\linewidth]{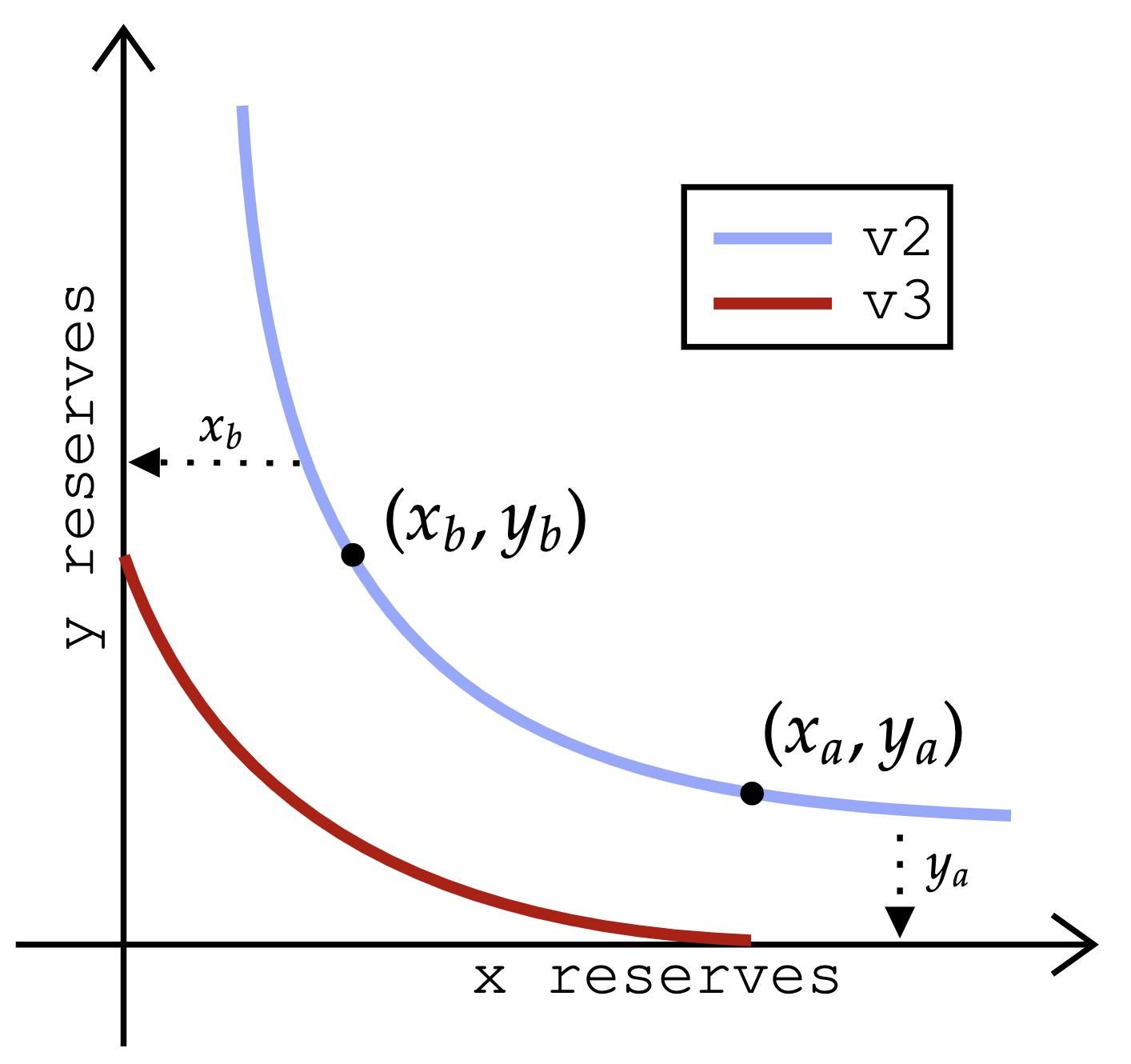}
    \caption{The affine shift of the Uniswap v2 reserve curve (blue line) to  the  v3 reserve curve (red line). The v3 reserve curve is only defined when the  price is in the  interval $[a,b]$, where $a < b$ are the price bounds specified by the LP.  
    \label{fig:shift}}
\end{figure}

\subsubsection{Trading and Fees}

Having described how LPs can provide and remove liquidity over different price intervals by sending a v3 contract the required token bundles, we now describe the trade dynamics for a v3 contract. Let us assume that $d$ LPs have provided liquidity to the contract, where we assume that the $j$-th LP has provided precisely $L_j$ units of $[a_j,b_j]$-liquidity \footnote{This is without loss of generality. To see why, suppose that an LP has provided $L_1$ units of liquidity to the interval $[a_1,b_1]$ and $L_2$ units of liquidity to $[a_2,b_2]$ with $L_1 > L_2$ and $b_1 = a_2$. This allocation is identical to the LP providing $L_1$ units of liquidity to $[a_1,b_2]$ and $L_2-L_1$ units of liquidity to $[a_2,b_2]$, where we can treat each allocation as a different LP. This argument can be extended to arbitrary liquidity allocations.}.
In addition,  suppose  the v3 contract has a current price given by $P \in (0,\infty)$. In what follows, we will show the dynamics of a trade that moves the contract price to $P' < P$ (the case where the price increases is symmetric). Before continuing, we make two key assumptions on the positions provided by LPs:
\begin{itemize}
    \item [A1] If the $j$-th and $r$-th LPs have allocations over intervals $[a_j,b_j]$ and $[a_r,b_r]$ respectively, then either the intervals are the same, the intervals are disjoint, or their intersection is at their boundaries.
    \item [A2] There exists an LP with an allocation $[a_i,b_i]$ such that $[P',P] \subseteq [a_i,b_i]$.
\end{itemize}



These assumptions are without loss of generality. For A1, this hinges on   Proposition~\ref{prop:v3subintervals}, with proof in Appendix \ref{appendix:sec-2-proofs}, which tells us that if an LP has an allocation worth $L$ units of $[a,b]$-liquidity, the token-bundle value of such an allocation can be decomposed into that of $L$ units of $[a,c]$ and $[c,b]$ liquidity for any choice of $c \in (a,b)$. Furthermore, this decomposition holds at any price. 
\begin{restatable}{proposition}{Fourthprop}
\label{prop:v3subintervals}
Consider an arbitrary closed interval $[a,b]$ and a value $c \in (a,b)$. Let $P$ be an arbitrary contract price, then 
\begin{align}
\mathcal{V}^{(3)}(L,a,c,P) + \mathcal{V}^{(3)}(L,c,b,P) = \mathcal{V}^{(3)}(L,a,b,P).
\end{align}
\end{restatable}

To apply this result to our  scenario, suppose that the $j$-th and $r$-th LP do not satisfy the first assumption. Let the intersection over their positions be $[a',b'] = [a_j,b_j] \cap [a_r,b_r]$.
For each LP, we can decompose their liquidity to be over $[a',b']$ and at most two other price segments (this follows from Proposition~\ref{prop:v3subintervals}). The resulting positions, when treated as different LPs, satisfy our first assumption.  

With regards to A2, when a v3 contract is created, an LP allocation is created on the entire price space, this corresponding to a v2 allocation for some amount of liquidity determined by the entity creating the contract. This means that there must exist LPs with positions over intervals $[a_i,b_i]$ and $[a_j,b_j]$ such that $P \in [a_i,b_i]$ and $P' \in [a_j,b_j]$. If these intervals are not equal, then it must be the case that $P' < a_i < P$, in which case, we can decompose the price movement as occurring from $P$ to $P'' = a_i$ and subsequently to $P'$. Ultimately $[P,P''] \subseteq [a_i,b_i]$, and we obtain the desired assumption above. 


\paragraph{Trading with Active Liquidity.}

Returning to the example, there is a contract price $P$ and a trade that moves the contract price to $P' < P$ by sending $A$ tokens to the contract. If the allocation corresponding to the interval $[a_j,b_j]$ of the $j$-th LP contains $P$, then we say the LP is {\it active} and that their $L_j$ units of $[a_j,b_j]$-liquidity is {\em  active liquidity}. Assumptions A1 and A2  imply the existence of  values, $a^*$ and $b^*$, such that if the $j$-th LP is active, then $a_j = a^*$ and $b_j = b^*$. We call $[a^*,b^*]$ the {\it active interval} at price $P$. 

Without loss of generality, assume the first $s$ of the $d$ LPs are active, and say that $L = \sum_{j=1}^s L_i$ is the {\it total active liquidity}  at  price $P$ and $(x,y) = \cV^{(3)}(L,a^*,b^*,P) \in \cR^{(3)}(L,a^*,b^*)$ denote the {\em active bundle} at   price $P$. Traders will send assets to the v3 contract, which move the active bundle along the v3 reserve curve, which is given by $\cR^{(3)}(L,a^*,b^*)$ and we refer to as the {\it active v3 reserve curve}. 
Finally, we let $x^* = \frac{L}{\sqrt{a^*}}$ and $y^* = L\sqrt{b^*}$, where  values $x^*$ and $y^*$ denote the maximum amount of $A$ tokens and $B$ tokens, respectively, that can be achieved 
by bundles on the active reserve curve. 

We recall that  $\gamma \in (0,1)$ is the trade fee rate of the contract and begin by considering a trader who sends $\Delta x \leq \frac{1}{1-\gamma} (x^* - x)$ units of token $A$ to the contract. This  trade amount ensures that the active bundle can move to another bundle on the active reserve curve. First,  fees  of $\gamma \Delta x$ units of token $A$ are skimmed for LPs and this is shared among the active LPs, proportionally,
such that if the $j$-th LP is active, they will receive $\frac{L_j}{L}\gamma \Delta x$ units of token $A$.
%
The remaining $(1-\gamma)\Delta x \leq x^* - x$ of token $A$  moves the active bundle to a bundle with $x' =  x + (1-\gamma)\Delta x \leq x^*$ units of $A$ tokens. That $x' \leq x^*$, ensures that there is a corresponding $y'$ such that the bundle $(x',y')$ lies on the active reserve curve, with contract price of $P' =  \frac{\phi_{a^*,b^*}(y')}{\phi_{a^*,b^*}(x')} \in [a^*,b^*]$. The trader receives $y - y'$ units of $B$ tokens. 

If the trader sends $\Delta x > \frac{1}{1-\gamma} (x^* - x)$ units of token $A$,  the contract first trades $\frac{1}{1-\gamma} (x^* - x)$ such tokens, which moves the contract bundle to the boundary of the active reserve curve. To trade the remaining tokens the contract must exit the current active interval, and determine the new active interval along with the set of active LPs and active liquidity. Whatever gas fees are needed for the computation of the aforementioned active quantities within the contract are  charged to the trader \footnote{Gas fees are also charged to traders for the computational overhead in updating v2 contracts, but we model only the (significant) surplus gas fees paid by traders participating in v3 contracts for the computation of active intervals}. 
Ultimately, the remaining $\frac{1}{1-\gamma} (x^* - x) - \Delta x$ units of $A$ tokens  are then traded iteratively as per this  process. In Appendix \ref{appendix:example-v3-dynamics}, we provide an in-depth example of v3 trading dynamics. 

\paragraph{Price Buckets.}
A Uniswap v3 contract partitions the space of all prices, $(0,\infty)$, into intervals, which are called {\em price buckets}. The effect is that an LP can only provide  $[a,b]$-liquidity  for values of $a$ and $b$ that correspond to endpoints of buckets (and an interval over which liquidity is provided corresponds to a set of contiguous buckets).
We assume that there are $n+m+1$ buckets, with indices from the set $\{-m,\dots,0,\dots,n\}$. The $i$-th such bucket is denoted by $B_i$, and typically represents an interval $[a_i,b_i]$. We distinguish $B_{-m}$ and $B_n$ however as representing the intervals $(a_{-m},b_{-m}] = (0,b_{-m}]$ and $[a_n,b_n) = [a_n,\infty)$ respectively.   
Thus we have $\cup_{i} B_i = (0,\infty)$, and $b_i = a_{i+1}$ for all $i<n$. By convention, the {\em unit price} $P=1$, where  tokens $A$ and  $B$ are at parity within the contract, lies in the $0$-th bucket (and $a_0 < 1 < b_0$). In addition, we let $\vect{\mu} = \{B_{-m},\dots,B_0,\dots,B_n\}$ denote the set of buckets in the v3 contract. 

\section{LP Profit and Loss over Price Sequences}
\label{sec:LP-PnL}
\setcounter{equation}{0}

Going forward, we assume that v3 contracts are secondary to larger markets between token pairs, these existing on centralized exchanges. As a result, there is a {\em market price} that is determined by on a centralized exchange and tha also influences the contract price through arbitrage trading.  
In this regard, it is important to highlight that the contract price from the previous section need not coincide with the market price of token $A$ on external exchanges. As we will see, keeping track of both quantities is important to the profit and loss of an LP.  

To facilitate exposition, we will  denote the {\em market price} by $P_m$ and the {\em contract price} by $P_c$ when it is necessary to distinguish them. Furthermore, we  let $P = (P_c,P_m)$ denote a {\it contract-market price pair}.
In what follows, we consider the perspective of an LP who borrows capital to create a liquidity allocation in a v3 contract. 
%
We let $\vect{\ell} = (\ell_{-m}, \dots, \ell_0, \dots, \ell_n)$ denote the {\it liquidity allocation}  
of the LP over each of the $m+n+1$ buckets in the contract, where $\ell_i$ denotes the units of $[a_i,b_i]$-liquidity held by the provider.

%
Let $\mathcal{B}: (\mathbb{R}^+)^2 \times \mathcal{R}^+ \rightarrow \mathbb{R}^+$ be the function that returns the {\em market worth in terms of token $B$} of a bundle of $A$ and $B$, under the assumption that token $A$ has a market price of $P_m$, that is
\begin{align}
\mathcal{B}((x,y),P_m) = P_m \cdot x + y.
\end{align}

To obtain $L'$ units of $[a_i,b_i]$-liquidity, i.e. corresponding to a particular bucket, when the contract price is $P_c$ requires an LP to send a bundle of $A$ and $B$ tokens given by $\mathcal{V}^{(3)}(L',a_i,b_i,P_c) = L' \cdot \mathcal{V}^{(3)}(a_i,b_i,P_c)$. At market price $P_m$, this bundle has an equivalent token $B$ worth given by $\mathcal{B}(\mathcal{V}^{(3)}(L',a_i,b_i,P_c),P_m)$ units of token $B$.

\subsection{Transaction Fees over a Price Sequence} \label{subsec:fee_over_price_seq}

LPs earn fees when trades occur (and thus the contract price changes). We  express the fees accrued by an LP with liquidity allocation $\vect{\ell}$ over a finite sequence of contract-market price pairs $\vect{P} = (P_0,...,P_T)$. For this, let $P_t = (P_{c,t},P_{m,t})$ denote the $t$-th contract-market price pair.
%
Without loss of generality, we assume that each individual contract price movement from $P_t$ to $P_{t+1}$ occurs within a single bucket (we can always split the price movement accordingly). We focus on a single contract price movement from $P_{c,t} \rightarrow P_{c,t+1}$, where $P_{c,t}, P_{c,t+1} \in [a_i, b_i]$. We let $L'_{t\rightarrow t+1}=\ell_i$ denote the liquidity the LP has in this bucket, and let $L_{t \rightarrow t+1}$ denote the total liquidity of all LPs have in this bucket.

We first consider an upward contract price movement, $P_{c,t} \rightarrow P_{c,t+1}$,   where $P_{c,t+1} > P_{c,t}$. Here, $\Delta y = L_{t \rightarrow t+1} \Delta^y_{P_{c,t},P_{c,t+1}}$ units of token $B$ are used for trading with the contract. Since a $\gamma$ proportion of all funds sent to the contract are skimmed for LP fees, this  means  $(1/(1-\gamma))L_{t \rightarrow t+1} \Delta^y_{P_{c,t},P_{c,t+1}}$ units of token $B$ are sent to the contract to move the price. Of this, a proportional $ L'_{t\rightarrow t+1}/{L_{t \rightarrow t+1}}$ quantity of $\Delta y$ is traded using the $L'_{t \rightarrow t+1}$ units of liquidity of the LP (liquidity that is active in the interval $[P_{c,t},P_{c,t+1}]$), and of this  $\gamma$ fraction is skimmed as fees. Combining, the LP earns
\begin{align}
F^B_{P_t \rightarrow P_{t+1}}(\vect{\ell}) = L'_{t\rightarrow t+1} \left(\frac{\gamma}{1-\gamma} \Delta^y_{P_{c,t},P_{c,t+1}} \right)
\end{align}
units of token $B$ for this contract price movement and $F^A_{P_t \rightarrow P_{t+1}}(\vect{\ell})=0$ units of $A$ tokens.
 
The analysis for a downward contract price movement from $P_{c,t}$ to $P_{c,t+1}$  with $P_{c,t+1} < P_{c,t}$, is almost identical. In this case, 
some number of units of token $A$ are sent to the contract by a trader, to move the contract price.
In particular, $\Delta x = (1/(1-\gamma)) L_{t \rightarrow t+1} \Delta^x_{P_{c,t},P_{c,t+1}}$ units of token $A$, of which the LP also receives a $\gamma L'_{t \rightarrow t+1}/L_{t \rightarrow t+1}$ portion. Putting everything together, the LP earns
\begin{align}
F^A_{P_t \rightarrow P_{t+1}}(\vect{\ell}) =  L'_{t \rightarrow t+1}  \left( \frac{\gamma}{1-\gamma} \Delta^x_{P_{c,t},P_{c,t+1}} \right)
\end{align}
units of token $A$ for this price movement and $F^B_{P_t \rightarrow P_{t+1}}(\vect{\ell})=0$ units of $B$ tokens. 
Finally, if the contract price does not move, then no transaction fees are accrued. In other words, if $P_{c,t} = P_{c,t+1}$, then $F^A_{P_t\rightarrow P_{t+1}}(\vect{\ell}) = F^B_{P_t \rightarrow P_{t+1}}(\vect{\ell}) = 0$.

Note that the transaction fee earned by a LP is  determined by $P_t, P_{t+1}$, and $L'_{t \rightarrow t+1}$, and does not depend on the total liquidity $L_{t \rightarrow t+1}$ that is active for this price movement. Similarly, in Uniswap v2, the transaction fee that an LP earns from a single price movement is totally decided given $P_t, P_{t+1}$, and $L'$, where $L'$ is the LP's liquidity over the entire price interval $[0, \infty)$, regardless of the total liquidity $L$ over the entire price interval. This means that 1 unit of liquidity over $[0, \infty)$ in v2 and 1 unit of liquidity over a price interval $[a, b]$ in v3 would gain the same amount of transaction fees if the same price movement $P_t$ to $P_{t+1}$ happens in v2 and v3 pools (given $P_{c,t}, P_{c,t+1} \in [a, b]$). 

We  now express the total transaction fees earned under liquidity allocation $\vect{\ell}$ and price sequence $\vect{P} = (P_0, \dots P_T)$. In what follows, we let $\mathbb{I}(A)$  denote the indicator function for event $A$.  
\begin{definition}[Trading fees]
Suppose that a provider has liquidity allocation $\vect{\ell}$ over the contract-market price pair sequence given by $\vect{P} = (P_0, \dots P_T)$. Let $F^A(\vect{\ell},\vect{P})$ and $F^B(\vect{\ell},\vect{P})$ denote the accrued amounts of token $A$ and token $B$  to the provider respectively as {\em trade fees}, expressed as follows: 
\begin{align}
    F^B(\vect{\ell},\vect{P}) = \sum_{t=0}^{T-1} F^B_{P_t \rightarrow P_{t+1}}(\vect{\ell}) \cdot \mathbb{I}(P_{c,t+1} > P_{c,t}), \quad F^A(\vect{\ell},\vect{P}) = \sum_{t=0}^{T-1} F^A_{P_t \rightarrow P_{t+1}}(\vect{\ell}) \cdot \mathbb{I}(P_{c,t+1} < P_{c,t}).
\end{align}

In addition,  let $F(\vect{\ell},\vect{P})$ denote the accrued trading fees in terms of token $B$ value at the final market price, which is given by:
\begin{align}
    F(\vect{\ell},\vect{P}) = \cB((F^A(\vect{\ell},\vect{P}),F^B(\vect{\ell},\vect{P})),P_{T,m}) = P_{T,m} \cdot  F^A(\vect{\ell},\vect{P}) + F^B(\vect{\ell},\vect{P}).
\end{align}

\end{definition}

Notice that $F(\vect{\ell},\vect{P})$ is linear in $\vect{\ell}$ for any contract-market price sequence $\vect{P}$.

\subsection{Impermanent Loss}

%
Suppose that an LP borrows the initial capital, which is a bundle of tokens A and  B, to purchase $\ell_i$ units of $[a_i,b_i]$-liquidity in a v3 contract, and needs to repay this  bundle in the future. We assume the initial contract-market price pair is given by $P = (P_c,P_m)$, in which case the capital borrowed is precisely the bundle $\mathcal{V}^{(3)}(\ell_i,a_i,b_i,P_c) = \ell_i \cdot \mathcal{V}^{(3)}(a_i,b_i,P_c)$. Suppose the contract-market price pair changes to $P' = (P_c',P_m') \neq P$ . At this price pair, the capital that was borrowed  has a token $B$ worth, given by $v^{(3)}_h(\ell_i,B_i,P,P') = \mathcal{B}(\mathcal{V}^{(3)}(\ell_i,a_i,b_i,P_c),P_m')$. We call this the token $B$ {\it holding value} for the v3 asset.
On the other hand, at price pair $P'$, the $\ell_i$ units of $[a_i,b_i]$-liquidity have a token $B$ worth of $v^{(3)}_p(\ell_i,B_i,P')=\mathcal{B} (\mathcal{V}^{(3)}(\ell_i, a_i,b_i,P_c'),P_m').$
We call this the token $B$ {\it purchase value} for the v3 asset. 

The discrepancy in the purchase and holding value is the impermanent loss, and represents a potential loss suffered by the LP, as they have to repay the equivalent token $B$ value of the initial borrowed capital. 
%
For v2 contracts, we can obtain similar expressions by letting $v^{(2)}_h(\ell,P,P') = \mathcal{B}(\mathcal{V}^{(2)}(\ell,P_c),P_m')$ be the token $B$ holding value, and  $v^{(2)}_p(\ell,P')=\mathcal{B} (\mathcal{V}^{(2)}(\ell,P_c'),P_m')$ be the token $B$ purchase value of $\ell > 0$ units of v2 liquidity. 
%
\begin{definition}[Impermanent Loss]
\label{def:add-imp-loss}
 Suppose that for a given bucket, $B_i = [a_i,b_i]$, an LP has obtained $\ell_i$ units of $[a_i,b_i]$-liquidity at initial price pair $P = (P_c,P_m)$. As the contract price shifts from $P$ to $P' = (P_c',P_m')\neq P$,  the LP suffers a v3 {\em impermanent loss} of $\mathit{IL}^{(3)}(\ell_i,B_i,P,P')$:
\begin{align}
\mathit{IL}^{(3)}(\ell_i,B_i,P,P') = v^{(3)}_h(\ell_i,B_i,P,P') - v^{(3)}_p(\ell_i,B_i,P').
\end{align}

If instead, an LP has obtained $\ell > 0$ units of v2 liquidity at initial price pair $P$, then as the contract price shifts from $P$ to $P'$,  the LP suffers a v2 {\em impermanent loss} of $\mathit{IL}^{(2)}(\ell,P,P')$:
\begin{align}
\mathit{IL}^{(2)}(\ell,P,P') = v^{(2)}_h(\ell,P,P') - v^{(2)}_p(\ell,P').
\end{align}


For a given liquidity allocation $\vect{\ell}$ obtained at initial price pair $P = (P_c,P_m)$,  let $\mathit{IL}^{(3)}(\vect{\ell},P,P')$ denote the {\em overall impermanent loss} that an LP suffers from the contract-market price movement from $P$ to $P'$:
\begin{align}
\mathit{IL}^{(3)}(\vect{\ell},P,P') = \sum_{i=-m}^n \mathit{IL}^{(3)}(\ell_i,B_i,P,P').
\end{align}
%
%
Since we focus on v3 contracts, we also use the shorthand $\mathit{IL}$ to refer to v3 impermanent loss.
\end{definition}

In practice, when the contract price deviates enough from the market price, this presents an arbitrage opportunity. For this reason, contract prices are expected to  track market prices reasonably closely, with the degree of closeness  depending  on the fee rate, $\gamma$ of the Uniswap contract, as well as gas fees incurred by trades. 
Proposition~\ref{lemma:imp-loss-linear-nonneg} tells us that when the contract price perfectly tracks the market price, then a trader's impermanent loss is always non-negative. 
\begin{restatable}{proposition}{Fifthprop}
\label{lemma:imp-loss-linear-nonneg}
For any choice of initial contract-market prices, $P = (P_c,P_m)$, and end prices $P' = (P_c',P_m')$, the v3 impermanent loss, $IL^{(3)}(\vect{\ell},P,P')$, is linear in $\vect{\ell}$. In addition, if $P_c' = P_m'$, then $IL^{(3)}(\vect{\ell},P,P')$ is non-negative for any $\vect{\ell}$. Similarly, for any choice of $P$ and $P'$, the v2 impermanent loss $IL^{(2)}(\ell,P,P')$ is linear in $\ell$, and when $P_c' = P_m'$, it is non-negative for any choice of $\ell > 0$. 
\end{restatable}

\subsection{Profit and Loss}
\label{sec:prof-and-loss}

We now  describe the overall profit and loss of an LP with liquidity allocation $\vect{\ell}$ as contract-market prices follow the sequence  $\vect{P} = (P_0,...,P_T)$, and where the LP borrows the capital required to create their allocation $\vect{\ell}$. The LP accrues fees over the price sequence, and at the end of the sequence, they remove their allocation $\vect{\ell}$ from the contract, thereby receiving an overall bundle of $A$ and $B$ tokens that is a function of the end contract price, $P_{c,T}$. Finally, the LP must repay the capital used to initially create the position, hence the overall profit and loss of the LP's allocation simply consists of their accrued fees minus their impermanent loss, where quantities are measured in terms of token $B$ using the final market price as reference of exchange between $A$ and $B$ tokens.
%
\begin{definition}[Profit and Loss]
\label{def:PnL-sequence}
We denote $\mathit{PnL}(\vect{\ell},\vect{P})$ as the overall {\em profit and loss} of an LP with liquidity allocation $\vect{\ell}$ over  contract-market price  sequence  $\vect{P} = (P_0, \dots, P_T)$, i.e., 
\begin{align}
\mathit{PnL}(\vect{\ell},\vect{P}) = F(\vect{\ell},\vect{P}) - \mathit{IL}(\vect{\ell},P_0,P_T).
\end{align}
\end{definition}

Since $F$ and $\mathit{IL}$ are both linear in $\vect{\ell}$, so is $\mathit{PnL}$ for any price sequence $\vect{P}$.  

\section{Optimal Liquidity Provision}
\label{sec:riskneutral}
\setcounter{equation}{0}


In this section, we assume that an LP has a probabilistic model for  contract-market prices sequences, and  formulate the problem of finding an optimal liquidity provision  over a finite time horizon as an optimization problem. For a risk neutral LP this is a linear program and for a risk averse LP this is a convex optimization problem. 
We use the notation $\cP$ to denote a distribution over contract-market price sequences and call an instance of $\cP$ the {\it belief profile of an LP}, prior to opening a v3 liquidity allocation. 

\subsection{Liquidity-independent LP Beliefs}

In general,  liquidity provision in Uniswap v3 is a  game-theoretic problem, where the  agents include the LPs as well as traders. An important manifestation of the strategic interdependence of agents in this game is the distribution $\cP$ on contract-market prices. As a concrete example, consider multiple LPs providing liquidity over buckets. If the aggregate liquidity allocation places large amounts of liquidity around the current contract-market prices, then we can expect, at least in the short term, that $\cP$ will follow contract-market prices closer to the current price, as higher liquidity results in less price slippage. At the same time, less slippage makes the pool more attractive to traders, and can potentially increase the volume of trade, and hence further change $\cP$.

In our model, we make the simplifying assumption that $\cP$ is independent of the  liquidity provided by LPs, and this is what we mean by a {\it Liquidity-independent} belief profile. Given this, the question of how to optimally allocate liquidity is  simplified and  becomes decision-theoretic, such that an LP's profit and loss only depends on the liquidity allocation they provide to the contract and not on the investment by other LPs. Though the model fails to capture the coupling between investment and price dynamics, which further entails game-theoretic aspects of liquidity provision, we view  this as an important first step in understanding key design considerations in Uniswap v3 contracts. 

That  a Liquidity-independent LP belief profile  reduces the question of how to optimally compute LP positions to a decision-theoretic problem arises  from our analysis in Section~\ref{sec:LP-PnL}. Since LPs earn fees proportional to the liquidity they provide over contract prices traded, at face value it may seem that their profit and loss depends on the liquidity allocations of other LPs (as this affects the proportion of liquidity an LP owns at given prices). However, 
for a liquidity-independent price movement, the rules of the CFMM mean that the volume traded is proportional to the liquidity supporting that price movement. As a result, as an LP provides more liquidity, not only are they increasing their proportional share of  fees, but also the total volume traded for a given price sequence. The combination of these two dependencies balance in such a way that earned fees are linear in the liquidity allocation of an LP  (the same is true for the impermanent loss).

\subsection{Maximizing Expected Profit and Loss: Risk Neutral LP}
\label{sec:maximizing-PnL}

We assume that an LP borrows the capital required to create a liquidity allocation  for a fixed time horizon. Once the time horizon ends, the LP removes their liquidity from the contract, and uses this capital alongside accrued fees to pay the amount owed for creating the position. 
For a risk-neutral LP, the relevant quantity to optimize is the {\em expected profit and loss}  over  a distribution of  price sequences corresponding to  belief profile $\cP$,
\begin{align}
\mathit{PnL}_{\cP}(\vect{\ell}) = \mathbb{E}_{\vect{P} \sim \cP}(\mathit{PnL}(\vect{\ell},\vect{P})).
\end{align}

Since  $\mathit{PnL}(\vect{\ell},\vect{P})$ is linear for any given choice of $\vect{P}$, it follows that $\mathit{PnL}_{\cP}(\vect{\ell})$ is also linear in $\vect{\ell}$. We consider an LP that has an initial budget of $D$ 
units of token $B$ and wishes to create an optimal liquidity allocation. 
For each bucket $B_i \in \vect{\mu}$, we  let $w_i = \mathcal{B}(\cV^{(3)}(1,a_i,b_i,P_{c,0}),P_{m,0})$
denote the token $B$ worth of 1 unit of liquidity at the intial parity contract-market price, $P_0 = (P_{c,0},P_{m,0}) = (1,1)$.  
Summing over buckets, the budget constraint is $\sum_i \ell_i w_i \leq D$. The  optimization problem is:
\begin{align}
\label{eq:optimization-PnL-risk-neutral}
\max_{\vect{\ell}} \quad & \mathit{PnL}_{\cP}(\vect{\ell}) \nonumber\\
\textrm{s.t.} \quad & \sum_i \ell_i w_i \leq D \\
  &\ell_i \geq 0, \quad \mbox{for all $i$} \nonumber   
\end{align}

The linearity of the objective function allows us to write  $\mathit{PnL}_{\cP}(\vect{\ell}) = \sum_{i=-m}^n \alpha_i \ell_i$, where $\alpha_i$ is the expected PnL of a single unit of liquidity in bucket $B_i \in \vect{\mu}$. 
Based on this, the optimal bucket has index $i^* = \argmax_{i} [\alpha_i/{w_i}]$, and the optimal PnL is given by $D \alpha_{i^*}/{w_{i^*}}$. This is  linear in $D$, and so we can assume without loss of generality that $D = 1$. The {\it optimal normalized PnL},  $\OPT(\cP,\vect{\mu})$, is 
\begin{align}
\OPT(\cP,\vect{\mu}) = \max_{i \in\{-m, \dots, n\}} \left[\frac{\alpha_i}{w_i}\right].
\end{align}

It is the Liquidity-independent nature of  $\cP$ that  leads to linearity in $\vect{\ell}$ in the objective $PnL_\cP(\vect{\ell})$, and in turn to this linear optimization problem.

\subsection{Maximizing Expected Profit and Loss: Risk-Averse LP}
\label{sec:riskaverse}
\setcounter{equation}{0}

Although an LP  maximizes expected PnL to allocate  liquidity to the single bucket with the highest normalized PnL, such an allocation is inherently  risky: if the price ever goes outside of this  bucket over the LPs time horizon then  not only does the LP miss out on fees, but the LP also suffers an 
impermanent loss.
Indeed, empirical evidence shows that LPs tend to use liquidity allocations that make use of multiple buckets,  spreading their liquidity  over a large range of prices. 

For these reasons, we introduce the notion of risk-aversion, and use the exponential  utility function (\textit{constant absolute risk aversion})~\cite{arrow1965aspects, pratt1978risk}.  
%
\begin{definition}
\label{def:exponential_utility}
Let $u_a(x)$, for parameter $a\in \mathbb{R}$, denote the {\em exponential utility function}, with 

\begin{align}
    u_a(x) &= 
    \begin{cases}
    \left(1-e^{-ax}\right) / a & \mbox{if $a\neq 0$, and} \\
    x & \mbox{otherwise.}
    \end{cases}
\end{align}

If $a < 0$ the LP is risk-seeking,  if $a > 0$ the provider is risk-averse, and if $a=0$ the LP is risk-neutral. 
\end{definition}


We can now express the utility that an LP obtains for a given price sequence as follows.
\begin{definition}[Risk-averse PnL]
Suppose that an LP has a liquidity allocation  $\vect{\ell}$ over a bucket instance $\vect{\mu}$. For a given contract-market price sequence, $\vect{P} = (P_0,\dots,P_T)$, we denote the {\em risk-averse profit and loss} of the LP by $\mathit{PnL}^a(\vect{\ell},\vect{P})$, for risk-profile $a \geq 0$, given by:
\begin{align}
\mathit{PnL}^a(\vect{\ell},\vect{P}) = u_a(\mathit{PnL}(\vect{\ell},\vect{P})).
\end{align}

If $\vect{P}$ is  generated according to an LP belief profile $\cP$, then we let $\mathit{PnL}^a_{\cP}$ denote the expected risk-averse profit and loss of the LP:
\begin{align}
\mathit{PnL}^a_{\cP}(\vect{\ell}) = \mathbb{E}_{\vect{P}\sim \cP}(\mathit{PnL}^a(\vect{\ell},\vect{P}))
\end{align}

\end{definition}

\begin{lemma}
\label{lemma:convex-risk-averse}
$\mathit{PnL}^a_{\cP}(\vect{\ell})$ is concave in $\vect{\ell}$ for any choice of $a$ and $\cP$.
\end{lemma}
\begin{proof}
Function $u_a(x)$ is smooth and concave for any choice of $a$, as $e^{-ax}$ within the expression is convex. We have previously seen that $\mathit{PnL}(\vect{\ell},\vect{P})$ is linear in $\ell$ for any choice of $\vect{P}$, hence $\mathit{PnL}^a(\vect{\ell},\vect{P})$ is  concave in $\vect{\ell}$ for any choice of $\vect{P}$. Finally, concavity is preserved over  expectation.  
\end{proof}

We can now define the optimization problem for a risk-averse LP seeking to optimally create a liquidity allocation, subject to budget-constraints:
\begin{align}
\label{eq:optimization-PnL-risk-averse}
\min_{\vect{\ell}} \quad & -\mathit{PnL}^a_{\cP}(\vect{\ell}) \nonumber\\
\textrm{s.t.} \quad & \sum_i \ell_i w_i \leq D\\
  &\ell_i \geq 0, \quad \mbox{for all $i$} \nonumber    
\end{align}

As before, it is WLOG to focus on the scenario where $D = 1$, and in such a setting, we let $\OPT_a(\cP,\vect{\mu})$ denote the  value of the optimal  solution, with $\vect{\ell}^*$ to denote the corresponding liquidity allocation. As a consequence of Lemma~\ref{lemma:convex-risk-averse}, the optimization problem is convex. In the next section we exploit convexity to provide efficient methods for solving  $\OPT_a(\cP,\vect{\mu})$.

\subsubsection{Computational Methods}
\label{sec:risk-averse-comp-methods}

The main difficulty in solving the optimization problem~\eqref{eq:optimization-PnL-risk-averse} is that the objective function is an expectation with respect to $\cP$. In this regard, we approximate $\mathit{PnL}^a_{\cP}(\vect{\ell})$ as the average $\mathit{PnL}$ over a sample of price paths from  $\cP$.
We take $N$ i.i.d. sample price paths, $\vect{P}_1, \dots, \vect{P}_N \sim \cP$, and define the following objective:
\begin{align}
\mathit{PnL}^a (\vect{\ell} \mid \vect{P}_1,\dots,\vect{P}_N)
=
\frac{1}{N}\sum_{i=1}^N \mathit{PnL}^a (\vect{\ell},\vect{P}_i).
\end{align}

Taking expectations, we get:
\begin{align}
\mathbb{E}_{\vect{P}_1,\dots,\vect{P}_N \sim \cP} \left[ \mathit{PnL}^a (\vect{\ell} \mid \vect{P}_1,\dots,\vect{P}_N)\right] = 
\mathit{PnL}^a_{\cP}(\vect{\ell}).
\end{align}

We can define a corresponding optimization problem for a risk-averse LP seeking to approximately optimally create a liquidity allocation, subject to budget-constraints:
\begin{align}
\label{eq:optimization-PnL-risk-averse-approx}
\min_{\vect{\ell}} \quad & -\mathit{PnL}^a (\vect{\ell} \mid \vect{P}_1,\dots,\vect{P}_N) \nonumber\\
\textrm{s.t.} \quad & \sum_i \ell_i w_i \leq D\\
  &\ell_i \geq 0, \quad \mbox{for all $i$} \nonumber    
\end{align}

This is a convex optimization problem, and we can evaluate gradients for the objective function via standard methods. As before, we focus on the case where $D=1$, and in this case we use {\em projected gradient descent} to solve~\eqref{eq:optimization-PnL-risk-averse-approx}.

\subsection{Optimal Risk-averse PnL as a function of Bucket Characteristics}
\label{subsec:opt-pnl-buckets-characteristics}

$\OPT_a(\cP,\vect{\mu})$ is intricately tied to the characteristics of the buckets available in a v3 contract.
In practice, the Uniswap v3 contract uses buckets $\vect{\mu}$ with endpoints that correspond to multiplicative increases and decreases of the parity price, $P = 1$. The contract maintains a fixed set of price ticks $\{t(-q),\dots t(0),\dots t(q)\}$, where $t(i) = 1.0001^i$ and $q = 2^{23}$. Each contract also has a positive integer variable  {\tt tickspacking}, which we denote by $\Delta$, and dictates which tick values can be used as bucket endpoints. A tick can only be a bucket endpoint if $i \equiv 0 \mod \Delta$. We can express the  bucket structure by letting the $i$-th bucket, $B_i$, represent an interval $[a_i,b_i]$ such that $a_i = 1.0001^{\Delta i}$ and $b_i = 1.0001^{\Delta(i+1)}$. 
\begin{proposition}
\label{prop:small-bucket-more-pnl}
Suppose that $\vect{\mu}(\Delta)$ is the bucket list that results from setting {\tt tickspacing} to $\Delta \geq 1$. Furthermore, let $\Delta' = q \Delta$ for any integer $q > 2$. For any choice of LP belief profile $\cP$, and risk-aversion parameter $a$, we have:
\begin{align}
\OPT_a(\cP,\vect{\mu}(\Delta')) \leq  \OPT_a(\cP,\vect{\mu}(\Delta)).
\end{align}
\end{proposition}
\begin{proof}
These bucket designs are nested, with $\vect{\mu}(\Delta')$ a coarsened version of $\vect{\mu}(\Delta)$, which means that any liquidity allocation over $\vect{\mu}(\Delta')$ can also be represented by a allocation over $\vect{\mu}(\Delta)$. Suppose an allocation over $\vect{\mu}(\Delta')$ consists of allocating $\ell_i$ units of liquidity to $B_i$. An LP who creates an allocation with $\ell_i$ units of liquidity in each of the $q$ buckets in $\vect{\mu}(\Delta)$ that correspond to this bucket has an overall bundle value of the allocation 
that is equal to $\ell_i$ units of liquidity in $B_i$ (in a similar vein to Proposition~\ref{prop:v3subintervals}). Furthermore, as per Section~\ref{subsec:fee_over_price_seq}, this  results in the same fees for the LP. 
\end{proof}

\subsection{Trader Gas Fees}
\label{sec:trader-gas-fees}

Finer bucket lists allow LPs to create more complicated liquidity allocations that potentially make use of more active endpoints, and as per Proposition~\ref{prop:small-bucket-more-pnl}, this  leads to weakly improving expected utility for LPs in v3 contracts. However, this kind of refinement also results in increased gas fees for traders, since a  trade that pushes the contract price outside of the current active interval of the v3 contract requires that a new active interval is computed, alongside the set of active LPs and total active liquidity. 
This computational overhead is passed on to traders through gas fees.

For each bucket, $B_i \in \vect{\mu} = \{B_{-m}, \dots, B_n\}$ such that $i \neq -m$, we define $c_i(\cP,\vect{\mu})$, which is the expected number of crossings of its left endpoint, $a_i$, over the course of the  stochastic process associated with belief profile $\cP$.  What matters for gas fees  is the number of  active bucket endpoints that are crossed: for there to be computational overhead, there needs to be an LP with an allocation that uses an endpoint. We assume that each bucket  endpoint is active, in this sense, in what we call  the {\em bucket coverage assumption}. This is justified empirically:
%
  we find  that for  buckets close to the contract price there  is almost always  at least  one LP  with an allocation ending in that bucket's endpoint (see Appendix \ref{appendix:empirical-liquidity-dist}.)
\begin{definition}[Trader Gas Cost]
Let us consider a liquidity provision instance dictated by LP belief profile $\cP$ and buckets given by $\vect{\mu}$. We let $\GAS(\cP,\vect{\mu})$ denote the {\em expected gas fees} incurred by all traders over the time horizon $T$. Under the bucket coverage assumption, this is  given by:
\begin{align}
\GAS(\cP,\vect{\mu}) = \sum_{i=-m+1}^n c_i(\cP,\vect{\mu}).
\end{align}
\end{definition}

Since we are interested in the relative gas cost between distinct bucketing schemes, 
we normalize the gas fees per crossing of an active endpoint to be 1. 
In practice this quantity depends on the price of ETH.
%
As in the computation of $\OPT_a(\cP,\vect{\mu})$, we can take a sample-based approach to compute $\GAS(\cP,\vect{\mu})$, whereby we sample $N$ sequences of contract-market price pairs   from $\cP$ and compute the average number of bucket endpoint crossings as an approximation to $\GAS(\cP,\vect{\mu})$.

\subsection{The Uniswap v3 Contract Design Problem}

We are interested in how choices for bucket design $\vect{\mu}$ impact both $\OPT_a$ and $\GAS$, for fixed LP belief profiles $\cP$ and  a risk-aversion parameter $a \geq 0$. Both objectives are important:   higher PnL may attract liquidity providers to reduce slippage in trades, while increased gas costs can dissuade traders from participating, thereby reducing fees. We  focus on parametrized families of bucket sets.
\begin{definition}[Exponential Bucket Scheme]
Suppose that $\theta > 1$ is a real number and $\Delta,m,n \geq 1$ are integers. In an {\em $(\theta,\Delta,m,n)$-exponential bucket scheme}, buckets in $\vect{\mu}(\theta,\Delta,m,n)$ are parametrized as follows:
\begin{itemize}
    \item $B_{-m} = (0,b_{-m}]$, with $b_{-m} = \theta^{\Delta (-m+1)} $, 
    \item $B_n = [a_n,\infty)$, with $a_n = \theta^{\Delta n}$, \mbox{and}
    \item for $-m < i < n$, $B_i = [a_i,b_i] = [\theta^{\Delta i}, \theta^{\Delta (i+1)}]$.
\end{itemize}
\end{definition}

For convenience, we also let $\vect{\mu}(\theta,\Delta)$ denote an exponential bucket scheme for  $m,n\rightarrow \infty$, such that $B_i$ is a finite bucket for all $i \in \mathbb{Z}$. Exponential bucket schemes are a natural extension of the  bucketing scheme used by v3 contracts. In particular, a v3 contract with {\tt tickspacing} = $\Delta$ is  equivalent to a $(1.0001,\Delta)$-exponential bucketing scheme.
We can  express the Pareto frontier of such bucketing schemes, considering the dual 
objectives $\OPT$ and $\GAS$.  
%
\begin{definition}[The Uniswap v3 OPT-GAS Pareto Frontier]
Suppose that we fix an LP optimization profile given by $\cP$, risk-aversion parameter  $a \geq 0$, and a family of bucket sets parametrized by $\vect{\zeta}$ and denoted by $\vect{\mu}(\vect{\zeta})$. Furthermore, let $\vect{\zeta}' \neq \vect{\zeta}$ be such that:
\begin{align}
    \GAS(\cP,\vect{\mu}(\vect{\zeta}')) \leq \GAS(\cP,\vect{\mu}(\vect{\zeta})), \qquad \mbox{and}\quad \OPT_a(\cP,\vect{\mu}(\vect{\zeta}')) \geq \OPT_a(\cP,\vect{\mu}(\vect{\zeta})), 
\end{align}
where one of the inequalities is strict. We say that bucketing scheme $\vect{\mu}(\vect{\zeta}')$ {\it Pareto dominates} $\vect{\mu}(\vect{\zeta})$ for  LP belief profile $\cP$ and risk-aversion parameter $a \geq 0$. We denote this relationship over the parameter space of $\vect{\mu}$ by $\vect{\zeta'} \succ_{\cP} \vect{\zeta}$. In addition, we let $Pareto(\cP,\vect{\mu})$ denote the set of all parameters, $\vect{\zeta}$ which are not Pareto dominated. 
\end{definition}

\section{Simulation Results}
\label{sec:computational-results}

In this section, we take a computational approach to investigate the Uniswap v3 contract design problem\footnote{Our code of the computational method: \href{https://github.com/Evensgn/uniswap-modeling}{\url{https://github.com/Evensgn/uniswap-modeling}}}. For our results, we utilize a specific family of LP belief profiles that can be seen as an extension of the model from Capponi and Jia~\cite{capponi2021adoption}, repeated 
here over a longer time horizon. 

\subsection{Modeling Liquidity-independent Contract-market Price Changes}
\label{LP-belief-model-assumptions}

\begin{figure}[t]
    \centering
         \includegraphics[width=0.9\textwidth]{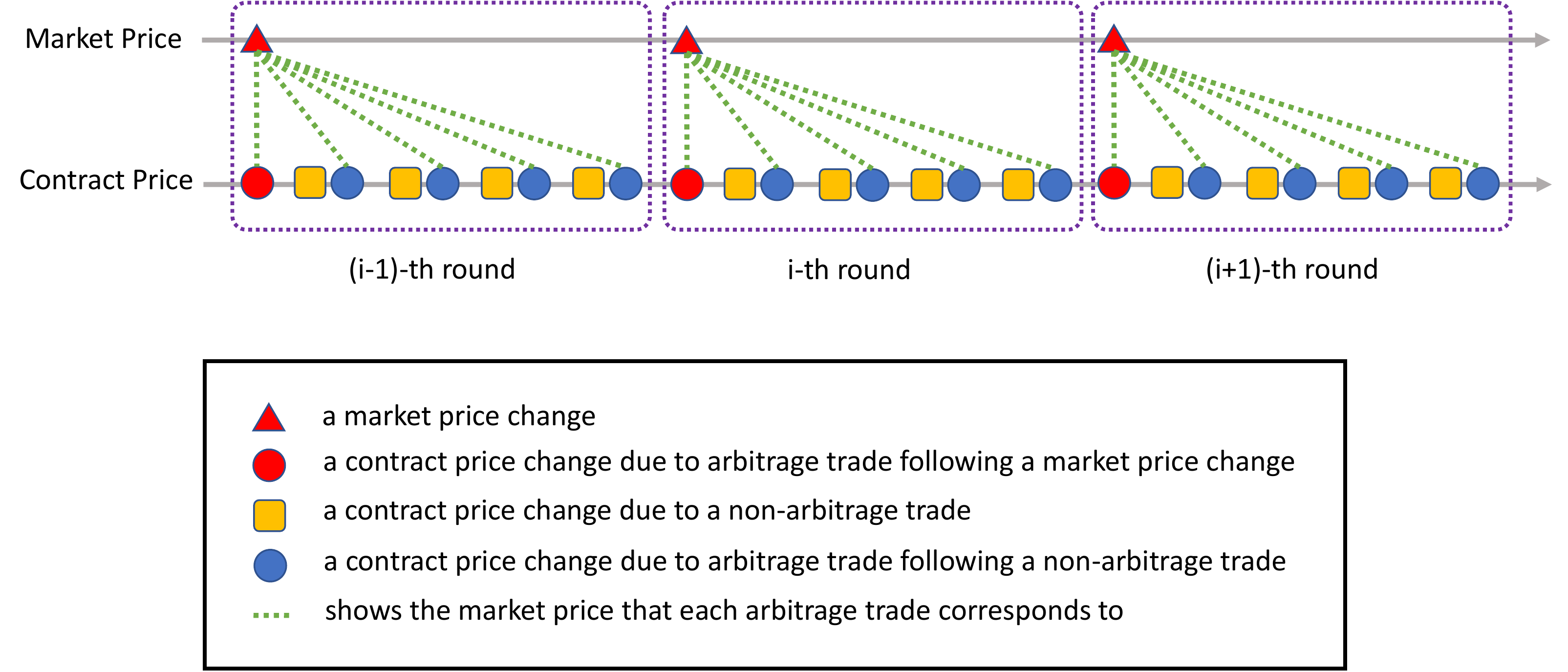}
    \caption{Market price changes triggering contract price changes, illustrated for three rounds, where each round contains four non-arbitrage and four arbitrage trades following a market price change.
    \label{fig:price-model}}
\end{figure}


Our family of Liquidity-independent belief profiles  generates contract-market prices via a repeated, two-stage process over $T$ rounds. We model a stochastic market price, with this inducing a stochastic contract price. In each round, in a  first stage we sample the new market price according to an  exogenous stochastic process, and in a second stage we model non-arbitrage traders who either increase or decrease contract price in a stochastic manner. At any point in either stage, if there is a large enough discrepancy between the contract and market price (as a function of the transaction fee rate, $\gamma$), then  arbitrageurs trade with the contract to bring the contract price close to market price.\footnote{Even though  arbitrageurs need to pay a gas fee when making an arbitrage trade, we ignore this because the gas fee does not grow with the trade size and becomes negligible for a large  trade.} Figure~\ref{fig:price-model} illustrates the way that we model changes in contract price in response to  market price changes, non-arbitrage trades, and arbitrage trades.

\subsubsection{Market Prices}
\label{sec:market-prices}

Market prices take values in set $\cZ(r,s,\omega) = \{Z_{-r} , \dots , Z_0 , \dots Z_s\}$, which is parametrized by $r,s \in \mathbb{N}$ and $\omega  > 1$, such that $Z_i = \omega^i$. As mentioned above, the beginning of a round marks an exogenous stochastic change in market price and we let $z_t \in \cZ(r,s,\omega)$ denote the market price in the $t$-th round. We model the sequence of market prices $\vect{z} = (z_0,\dots, z_T)$ with $z_0 = 1$ as approximating a {\em geometric random walk} (GRW), as is typical for modeling time series data in financial markets~\cite{LeRoy1992}. 

For this,  we assume that conditional on the price being $z_t = Z_i$ at time step $t$, the price will move to $z_{t+1} = Z_j$ where $j-i$ is distributed according to a binomial distribution, suitably truncated under the constraint that $Z_j \in \cZ$. In particular, we assume an integral {\it bandwidth} parameter of $W \geq 1$. 
Let $Y' \sim Binom(2W,p)$, such that $Y = Y'-W$ encodes the maximal change in index from the current price $Z_i$. This means that the price transitions as follows:
\begin{equation}
z_{t+1} = \left\{
        \begin{array}{ll}
            Z_{-r}& \quad \mbox{if $i+Y < -r$, or} \\
            Z_{s}& \quad \mbox{if $i+Y > s$, or} \\
            Z_{i+Y}& \quad \mbox{otherwise.} 
        \end{array}
    \right.
\end{equation}

By imposing a constant price ratio of $\omega$ in $\cZ$, and having price indices transition according to a binomial distribution, the stochastic price process approximates a {\it geometric binomial random walk} (GBRW). An exception is when the price is close to the boundaries of $\cZ$, where it deviates from GBRW due to truncation. We also require that the stochastic process is approximately on-average stable ($\mathbb{E}[\frac{z_{t+1}}{z_t}] = 1$), which we achieve by choosing a suitable value of $p \in [0,1]$ to govern the draw $Y' \sim Binom(2W,p)$, and subsequently, $Y = Y'-W$. 

\begin{proposition}[Informal]
\label{prop:binomial-martingle}
Suppose that $m = \log (\omega)$. If we let $p = \frac{m+2-\sqrt{m^2 + 4}}{2m}$, then for large bandwidth values, it follows that $\mathbb{E}[\frac{z_{t+1}}{z_t}] \approx 1$. Consequently, the price ratio stochastic process is approximately on-average stable. 
\end{proposition}

\begin{proof}[Proof (informal)]
First of all, we make use of the fact that if $X_n \sim Binom(n,p)$, then as $n \rightarrow \infty$, $X_n \overset{d}{\to} X$, where $X \sim \cN(np,np(1-p))$. We recall that $Y + W \sim Binom(2W, p)$, hence it follows that we can approximate $Y$ as being distributed according to $\cN(2Wp-W,2Wp(1-p))$. Let $Q = \omega^Y = \frac{z_{t+1}}{z_t}$, which in turn can be approximated as distributed according to a log normal distribution with parameters $(\mu,\sigma^2) = (2Wp-W,2Wp(1-p))$. Under the approximation, it follows that $\mathbb{E}[Q] = e^{m\mu + \frac{m^2\sigma^2}{2}}$. The given choice of $p$ ensures that this is   equal to 1. 
\end{proof}


\paragraph{Empirically Informing the GBRW.}

We require 
five parameters, $(r,s,\omega,W,T)$, to define how the market price evolves. The parameters $(r,s,\omega)$  give rise to the price space, $\cZ(r,s,\omega)$. The parameters $\omega > 1$ and $W \geq 1$  define the approximate GBRW that dictates random transitions over $\cZ(r,s,\omega)$ and $T \geq 1$ controls the number of times market price can change. Proposition~\ref{prop:binomial-martingle} implies that the design choice mostly depends on the multiplicative ratio $\omega$, and is not dependent on the total number of price tick $r+s+1$. This provides a way to obtain a price sequence that approximates an on-average stable geometric random walk, where $\omega$ and $W$ govern the overall volatility, as larger values of $\omega$ imply that price index changes result in larger multiplicative price changes, and larger bandwidth values, $W$, imply that the random walk can make larger jumps
in a given time step. 

 
The parameters $\omega$ and $W$ directly impact the volatility of the  GBRW. We fit these two parameters using historical price data via a maximum likelihood estimate (see Appendix \ref{append:MLE}). 
We use prices between Ethereum (ETH) and Bitcoin (BTC) for the low volatility regime, as the prices of these two tokens are highly correlated. We use prices between Ethereum (ETH) and USDT for the high volatility regime. For ETH/BTC prices, we estimate values $\omega = 1.0005$ and $W=3$, and for ETH/USDT, we estimate values $\omega = 1.0005$ and $W = 7$. Given this, we adopt $\omega = 1.0005$ in all  LP belief profiles and consider $W \in  \{3,5,7\}$ to represent low/medium/high market price volatility regimes.



We set $r = s = 150$, so that the GBRW takes prices in the interval $[\omega^{-r},\omega^s] \approx [0.9278, 1.0778]$, and consider time horizon  $T = 100$. Given the fact that our empirical results are informed from per-minute price data  (see Appendix \ref{append:MLE}), this time-scale corresponds to a roughly two hour window. This choice of a smaller price range and time horizon is for two reasons.
First,  we consider optimal LP strategies with respect to introducing liquidity at time $t=0$ and removing it at time $t=T$. In practice we expect such simple strategies to be more prevalent at smaller time scales. 
Moreover, using a GBRW with a coarse discretization of price space over large time horizons runs the risk of inadequately modeling smaller order price oscillations, which in turn impacts LP profit and loss (in terms of fees mostly); working with a smaller time horizon and price space mitigates this risk.


\subsubsection{Contract Price Updates}

As mentioned before, we model contract-price updates via rounds that undergo a two-stage process. 
The first stage in a round samples the market price, as per the approximate GBRW process from before.
The second stage  modifies the contract price according to both non-arbitrage and arbitrage trading. These modifications require 3 relevant parameters: $(k,\lambda,\gamma)$. The first parameter $k \in \mathbb{N}$ specifies the number of non-arbitrage trades that occur in a given round. The second parameter $\lambda \in (0,1)$ specifies the multiplicative price impact a single non-arbitrage trade has in a given round. The final parameter $\gamma \in [0,1]$ is the trading fee present in the v3 contract at hand. 

In more detail, the $t$-th round consists of a sequence of $2(k+1)$ contract-market price pairs  given by $P^t = (P^t_{0} ,\dots P^t_{2k+1})$. The $j$-th  pair is given by $P^t_{j} = (P^t_{c,j}, P^t_{m,j})$, and for all pairs it holds that $P^t_{c,j} = z_t$; i.e., the market price remains unchanged within this second stage of the round. 

For all odd $j$, the transition from $P^t_{j}$ to $P^t_{j+1}$ consists of a single {\em non-arbitrage trade}. With probability 1/2, the trader buys token $A$ from the contract such that contract price increases to $P^t_{c,j+1} = (1 - \lambda)^{-1}P^t_{c,j}$, and with probability 1/2, the trader sells token $A$ to the contract such that contract price decreases to $P^t_{c,j+1} = (1 - \lambda)P^t_{c,j}$. 
This form of price-based non-arbitrage trade is similar to that used in Capponi and Jia~\cite{capponi2021adoption}.

For all even $j < 2(k+1)$, the transition from $P^t_{j}$ to $P^t_{j+1}$ consists of a (potentially empty) {\em arbitrage trade}. Given the fee rate, $\gamma$, and any market price $P_m \in (0,\infty)$, we let $I_\gamma(P_m) = [(1-\gamma)P_m, (1-\gamma)^{-1} P_m]$ be the {\it no-arbitrage interval} around market price $P_m$. 
As shown in Angeris et al.~\cite{angeris2019analysis}, if a contract price $P_c$ is such that $P_c \in I_\gamma(P_m)$, then even if $P_c \neq P_m$, arbitrage is not profitable due to transaction fees. More specifically, no-arbitrage conditions for trading with a Uniswap contract with contract price $P_c \in (0,\infty)$ precisely amount to having $P_c \in I_\gamma(P_m)$. Consequently, if at $P^t_j$, the contract price is outside the no-arbitrage interval of the market price, we will assume that arbitrageurs trade in such a way that the contract price reaches the closest point in the no-arbitrage interval. For a given $x \in \mathbb{R}$ and interval $I \subseteq \mathbb{R}$ we let $\pi(x,I) = \argmin_{y \in I}|x - y|$ be the projection of $x$ onto $I$; i.e., the closest point in $I$ to $x$. Using this notation, we  say an arbitrage trader updates the contract price from $P^t_{c,j}$ to $P^t_{c,j+1} = \pi(P^t_{c,j}, I_\gamma(z_t))$. 

\paragraph{Relevant Parameter Regimes.}
Going forward we fix the empirically-informed values of $r,s = 150$, $\omega = 1.0005$ and $T = 100$ for precisely the reasons specified at the end of \ref{sec:market-prices}: 1) static LP strategies are more realistic at smaller time scales 2) finer discretizations of price space (and correspondingly smaller price ranges) mitigate the risk of model error arising from smaller order price oscillations. Consequently, we focus on modulating four parameters, $(W,k,\lambda,\gamma)$ to change the way in which market prices are updated and also how non-arbitrage trades and arbitrage trades induce contract price updates. Parameters $k$ and $\lambda$  modulate the number of non-arbitrage trades in a round, and the per-trade price impact of said trades, respectively. Parameters $W$ and $\gamma$  modulate how market prices evolve and when arbitrage trading kicks in. 

In Sections~\ref{sec:effect-delta-opt-gas}, \ref{sec:Pareto-results}, and \ref{sec:riskaverse-results} we fix $(W,k,\lambda,\gamma) = (5,10,0.00025,0.01)$. Here, the choice of $W = 5$ represents intermediate market price volatility and $k$ and $\lambda$ are chosen such that two consecutive non-arbitrage trades, if uninterrupted by arbitrage trades, result in a multiplicative change of $1.0005$ which equals the value of the parameter $\omega$ from the GBRW. Finally, the choice of $\gamma = 0.01$ is the highest fee tier available in Uniswap contracts, which can have fee tiers $\gamma \in \{0.0005, 0.003, 0.01\}$. 

In Section~\ref{sec:modulating-extra-params}, we vary $k,$ $\lambda$ and $\gamma$ to see how $\OPT_a$ and $\GAS$ depend on different scales of non-arbitrage and arbitrage trade while maintaining intermediate market price volatility fixed. In Section \ref{sec:price-vol-empirical}, we explore the impact overall contract-market price volatility has on our results. We do so by exploring two further parameter settings for $(W,k,\lambda,\gamma)$ given by $(3,5,0.002,0.01)$ and $(7,15,0.003,0.01)$. The former setting corresponds to a lower contract-market price volatility and the latter to a higher contract-market price volatility, compared with the baseline of $(5,10,0.00025,0.01)$.

\begin{figure}[t]
    \centering
         \includegraphics[width=0.48\textwidth]{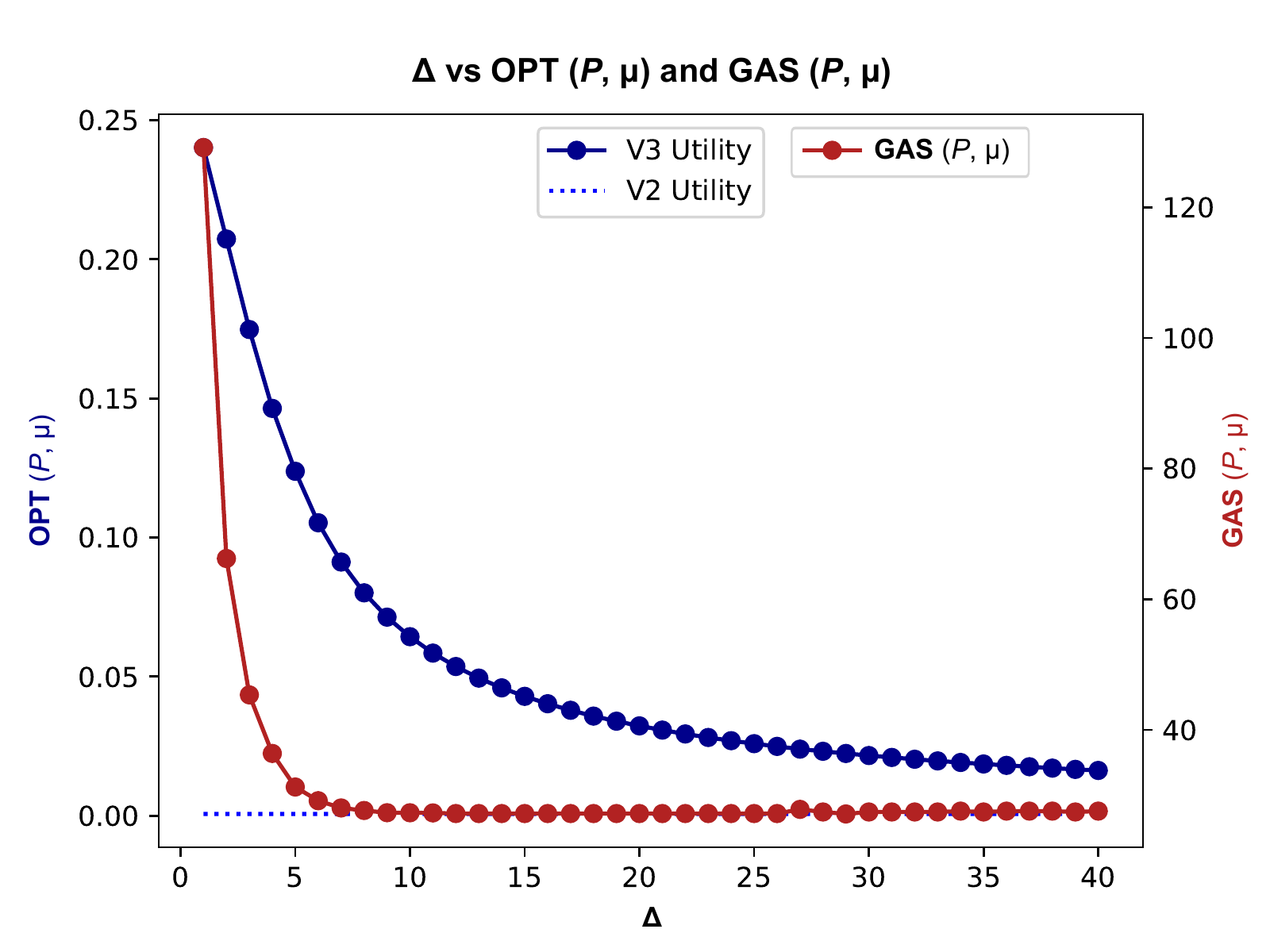}
          \includegraphics[width=0.48\textwidth]{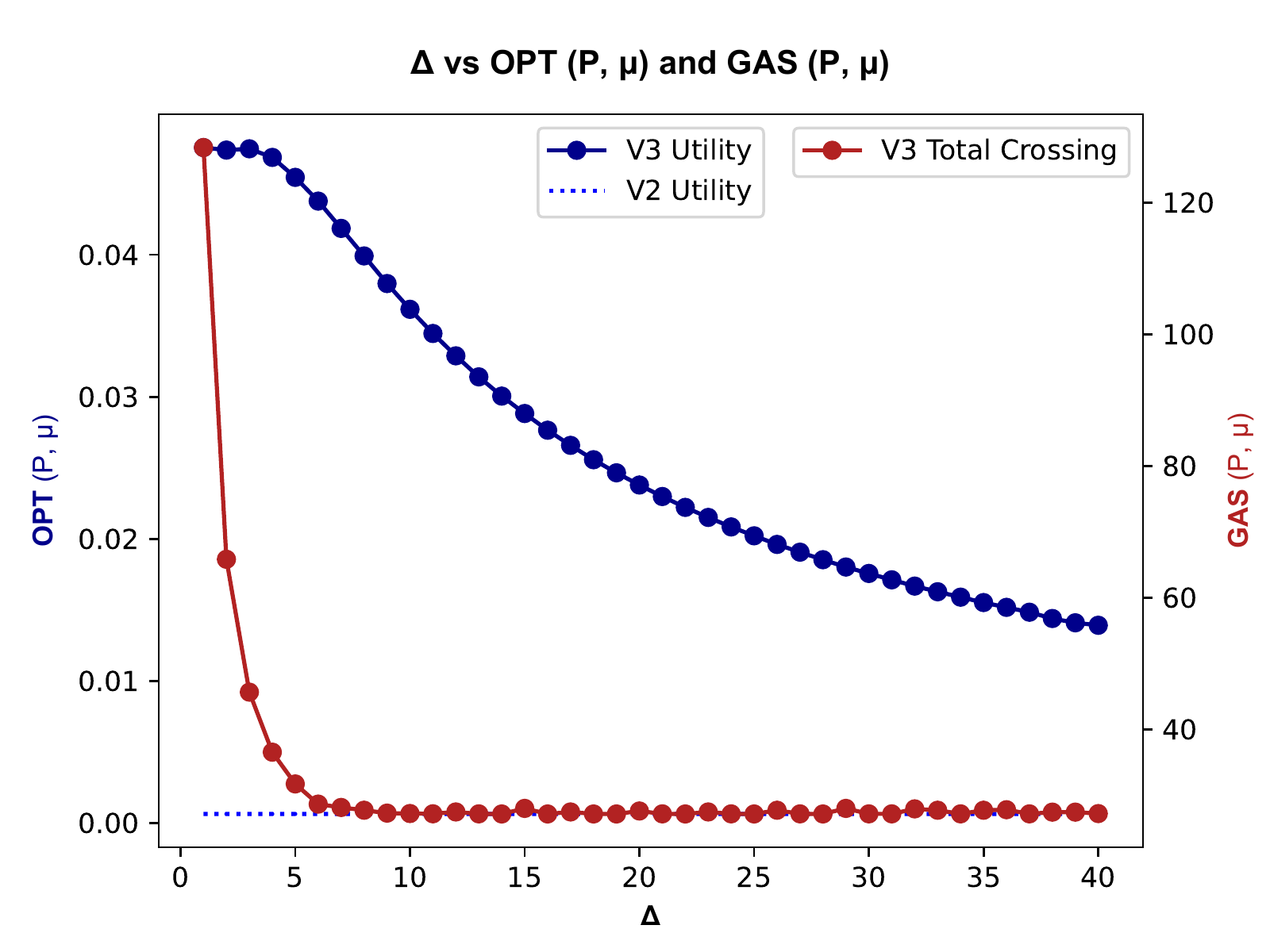}
    \caption{LP profit an loss and trader Gas fee, for $(W,k,\lambda,\gamma) = (5,10,0.00025,0.01)$, and $a=0$ risk-neutral (left) and $a=20$ risk-averse (right), and  a $(\theta,\Delta)$-exponential bucketing scheme with  multiplicative factor $\theta=1.002$ and  bucket spacing  $\Delta\in\{1,\dots,40\}$. The plots also show the corresponding v2 PnL.
    \label{fig:delta-opt-gas}}
\end{figure}

\subsection{The Effect of Bucket Size on PnL and Gas Cost}
\label{sec:effect-delta-opt-gas}

    As bucket sizes decrease, an LP's optimal PnL increases, in line with our result from Proposition \ref{prop:small-bucket-more-pnl}, and the expected gas cost to traders  increases as finer partitions result in more crossings of active price ticks (see Figure \ref{fig:delta-opt-gas}, where multiplicative factor $\theta$ is fixed to $1.002$ and we vary bucket spacing $\Delta$.
    The plot also includes the optimal PnL of an LP under a v2 contract (the v2 Gas cost is zero in our model, since no crossings of bucket endpoints ever occur). 

\subsection{The Uniswap v3 OPT-GAS Pareto Frontier}
\label{sec:Pareto-results}

\begin{figure}[t!]
    \centering
         \includegraphics[width=0.45\textwidth]{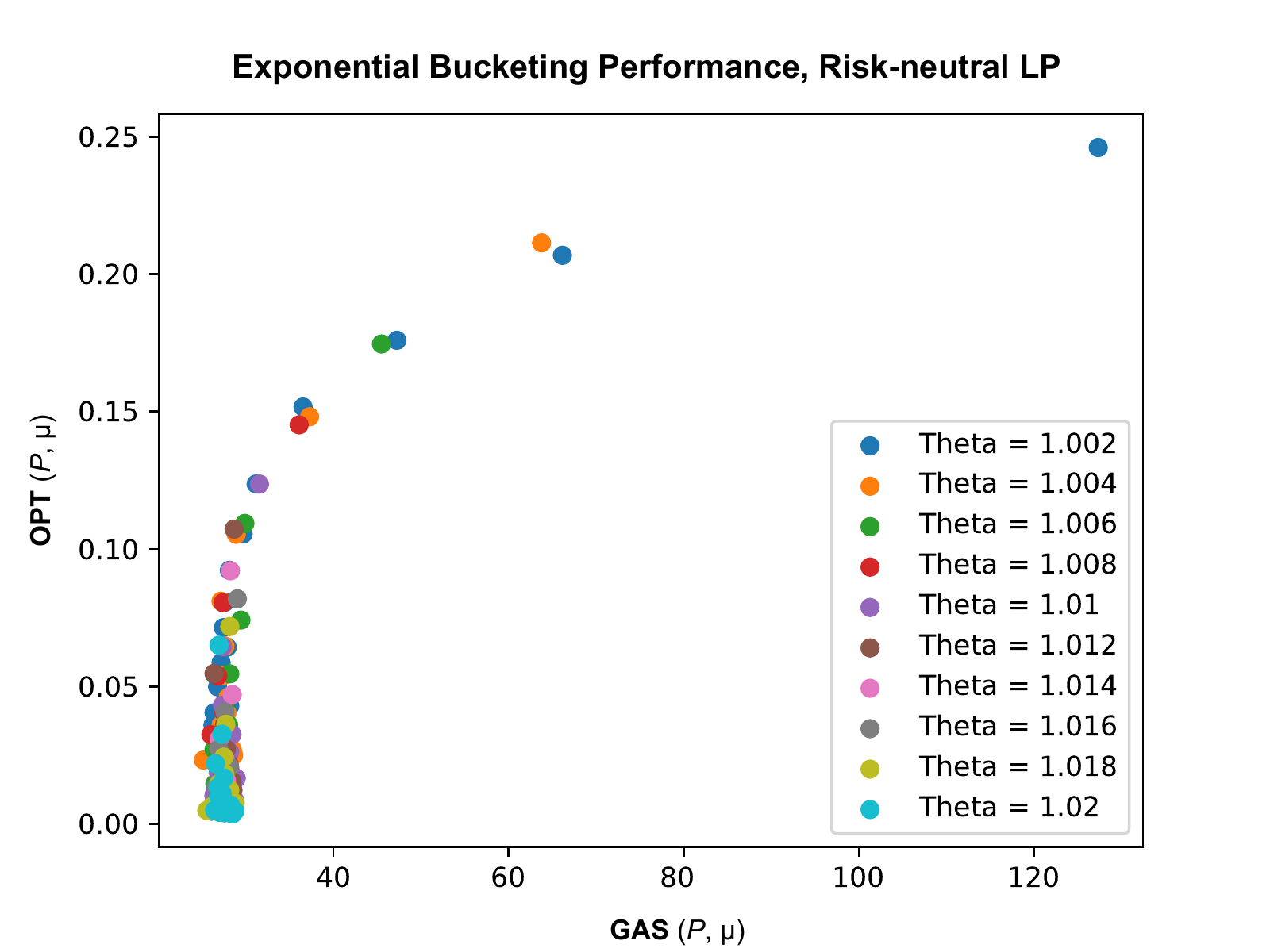}
         \includegraphics[width=0.45\textwidth]{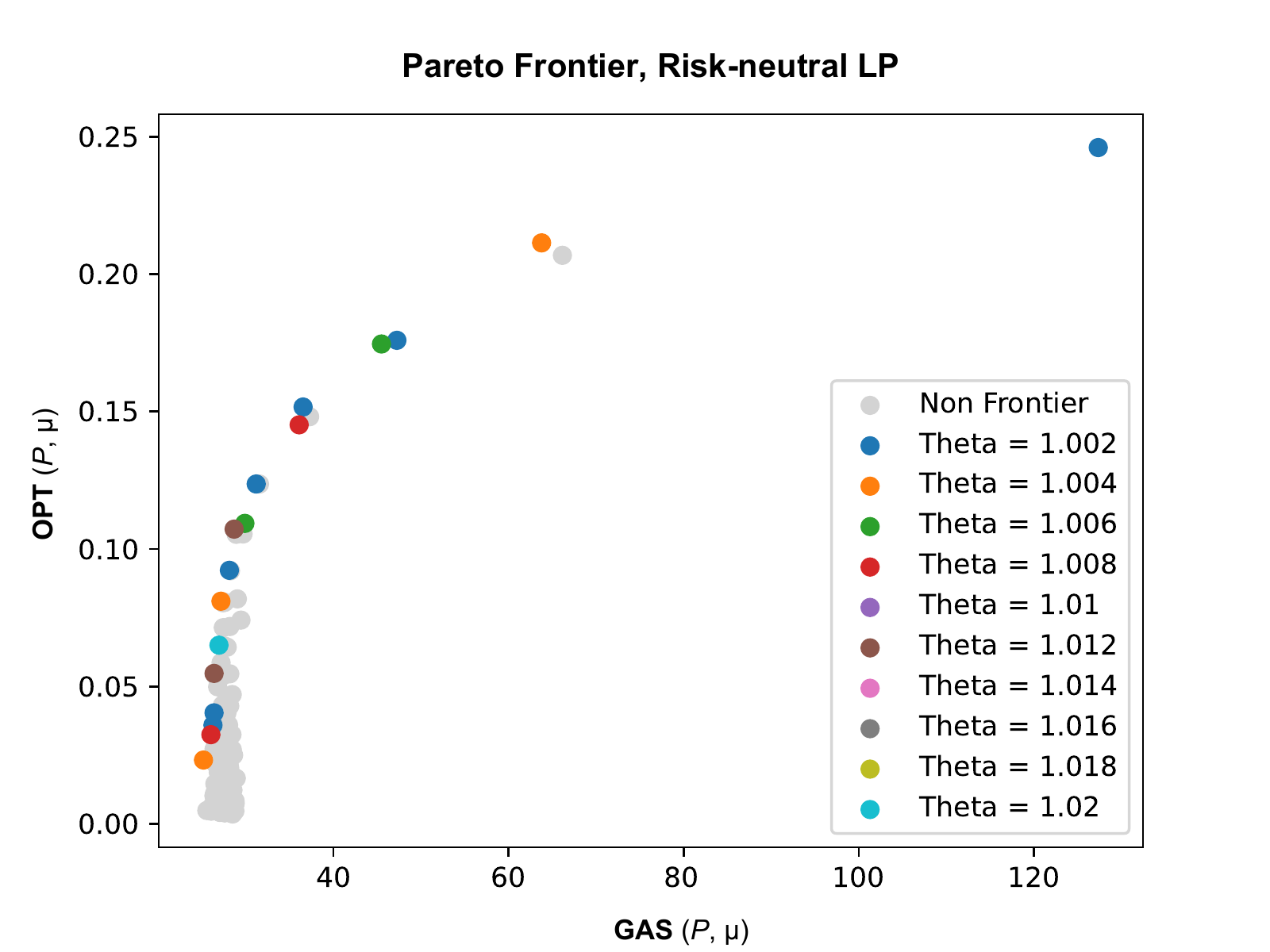}
         \includegraphics[width=0.45\textwidth]{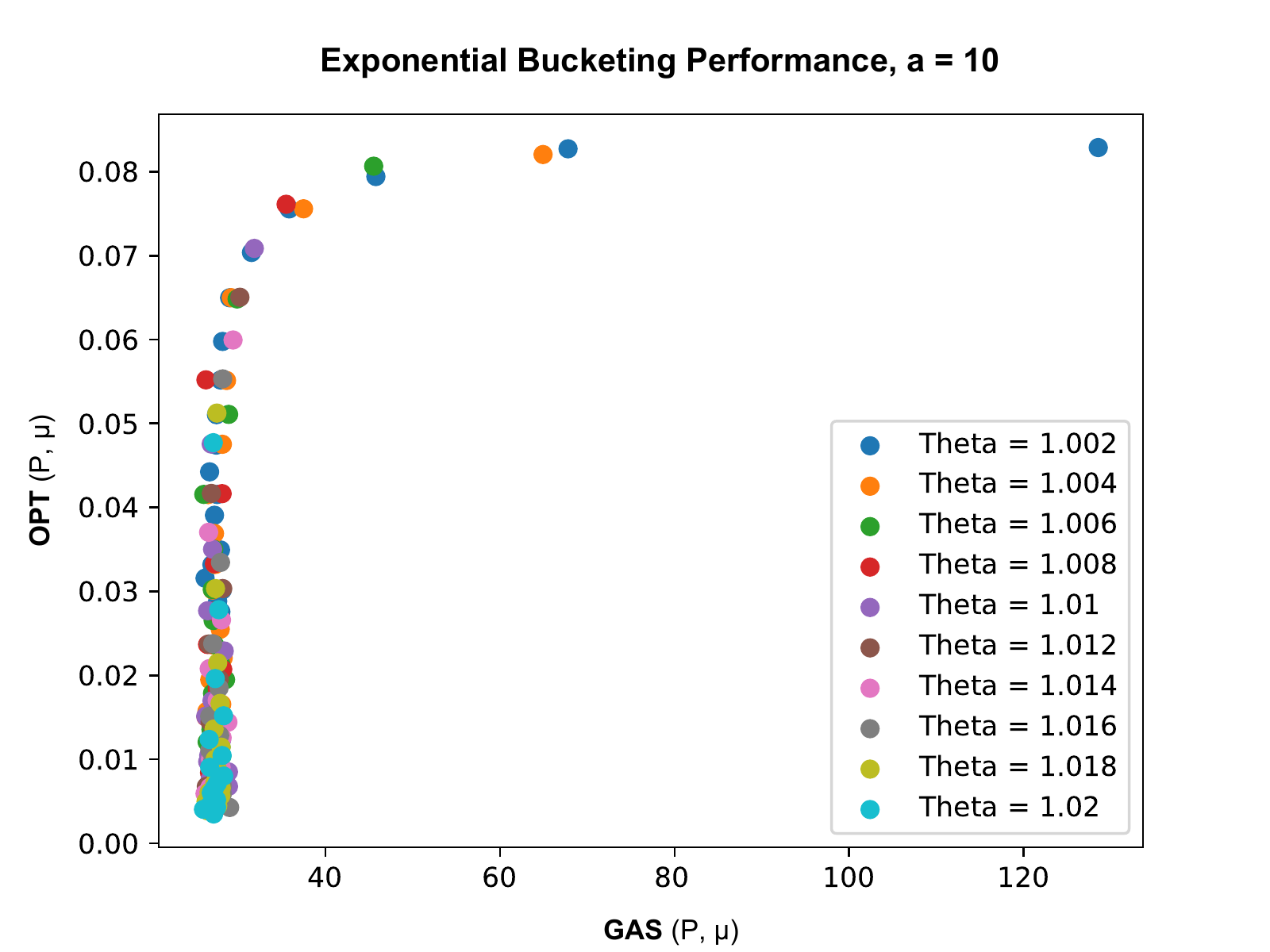}
         \includegraphics[width=0.45\textwidth]{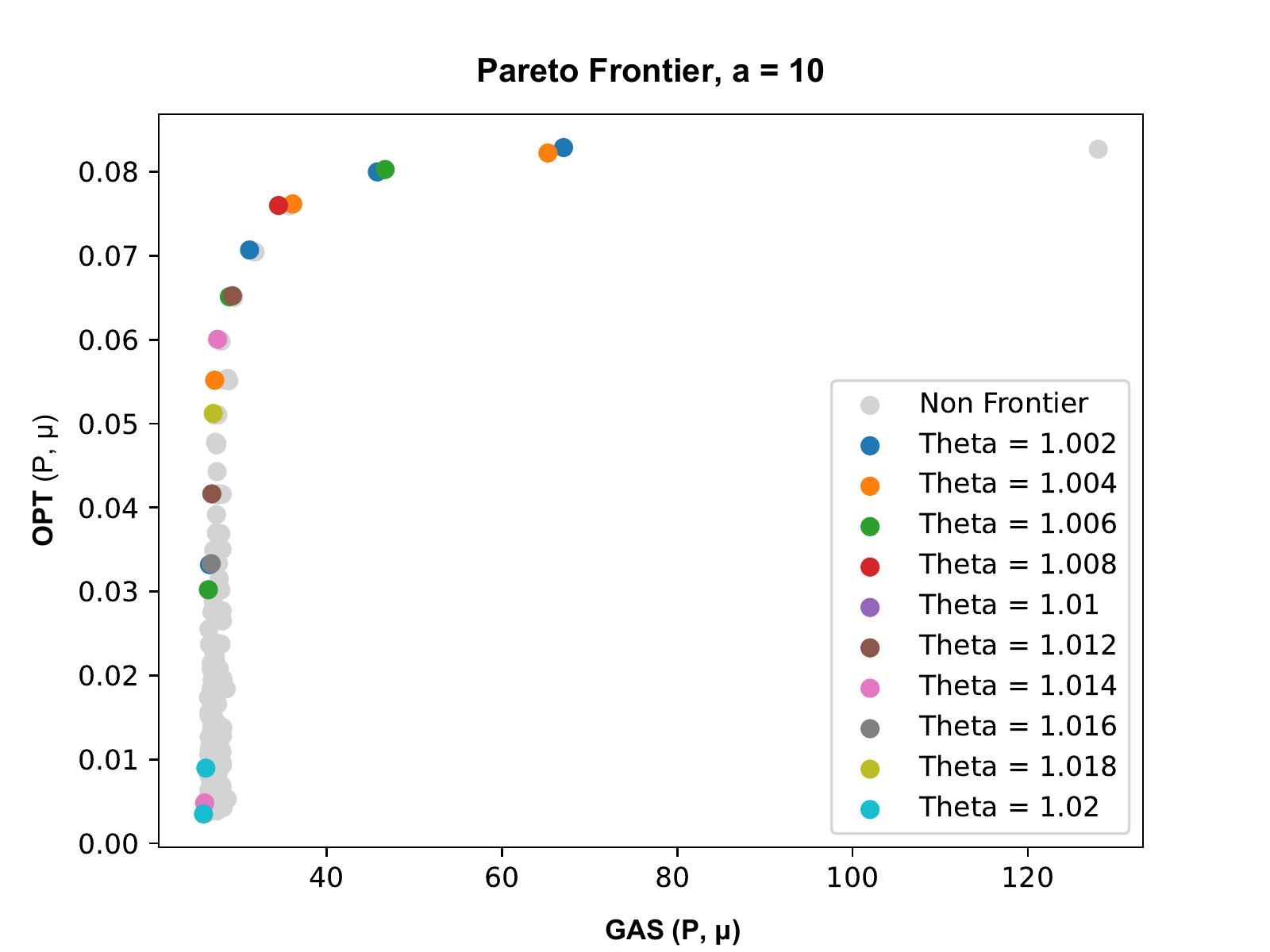}
         \includegraphics[width=0.45\textwidth]{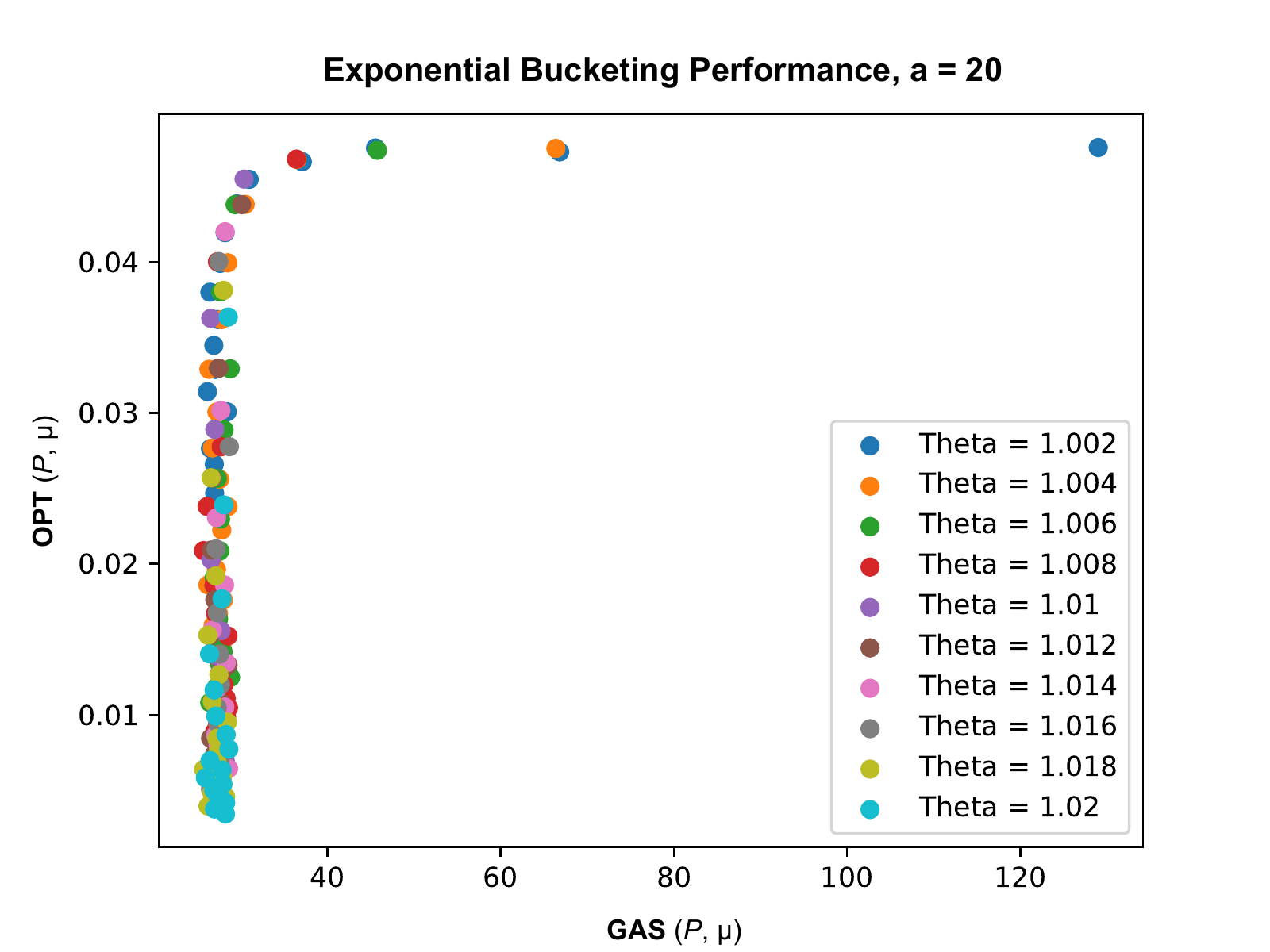}
         \includegraphics[width=0.45\textwidth]{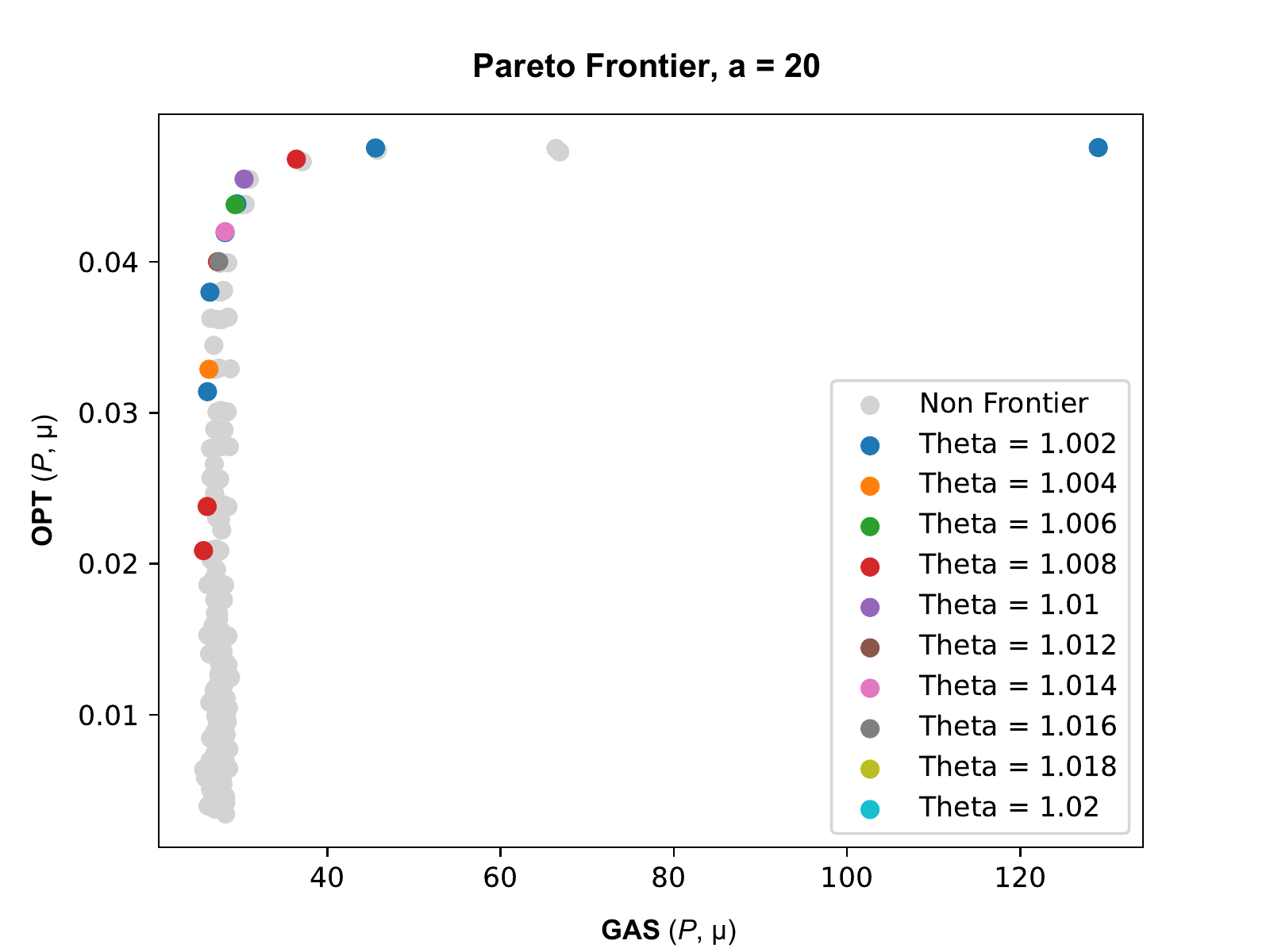}
    \caption{LP PnL vs Trader gas fees for $(W,k,\lambda,\gamma) = (5,10,0.00025,0.01)$,  LPs with risk parameter $a \in \{0,10,20\}$ (top, middle, bottom), and different  $(\theta,\Delta)$-exponential bucketing schemes. Each point corresponds to a specific  $(\theta,\Delta)$, varying bucket spacing $\Delta \in \{1,\dots,20\}$ and with color indicating multiplicative factor ($\theta$). The right column highlights only the bucketing schemes on the Pareto Frontier. 
    \label{fig:Pareto}}
\end{figure}

In Figure~\ref{fig:Pareto}, we plot the performance of different  bucketing schemes for different LP risk profiles. Points with the same color have the same multiplicative factor ($\theta$)   but  vary in bucket spacing $\Delta\in \{1,\dots,20\}$. The right-hand plots   highlight which of these points are on the  Pareto frontier. 
%
There is a multiplicity of different exponential bucketing schemes  on the Pareto frontier, as illustrated through  multiple colors and thus $\theta$ values and multiple $\Delta$ values for a given color, as there are multiple points of the same color on the frontier.
From this, a concrete design improvement for Uniswap v3 contracts lies in providing additional bucketing schemes for a contract designer, as suited to different PnL-Gas fee tradeoffs and depending on the preferences of agents (the LPs and traders).
The {\em status quo} design in Uniswap v3 is equivalent to 
a $(\theta,\Delta)$-exponential bucketing schemes with $\theta = 1.0001$ and $\Delta$ corresponding to {\tt tickspacing} in the contract.  
%
%
The top right point of each Pareto frontier plot (in light blue) of Figure~\ref{fig:Pareto} corresponds to a $(1.002, 1)$-exponential bucketing scheme. This $\theta$ value is closest to the v3 status quo of $\theta = 1.0001$ and results in relatively high PnL for LPs but  higher gas fees for  traders.\footnote{We only tested larger $\theta$ values due to the granularity constraint of our GBRW model.} 

\subsection{The Effect of Risk-Aversion}
\label{sec:riskaverse-results}
\setcounter{equation}{0}


The expected PnL and  standard deviation of PnL are decreasing functions, for an optimal LP behavior, are each decreasing  in the risk-aversion parameter $a \geq 0$ (see  Figure~\ref{fig:utility-std-vs-a}). This is to be expected for the constant absolute risk-aversion function. 
Moreover, as the level of risk-aversion increases, optimal liquidity allocations remain centered around the initial unit price, but become more spread out, as visualized in Figure~\ref{fig:prop_liquidity_allW}.  
This is to be expected, as spreading out an allocation  mitigates the risk that an LP will lose out on fees when prices exit a liquidity allocation. 
\begin{figure}[h]
    \centering
         \includegraphics[width=0.48\textwidth]{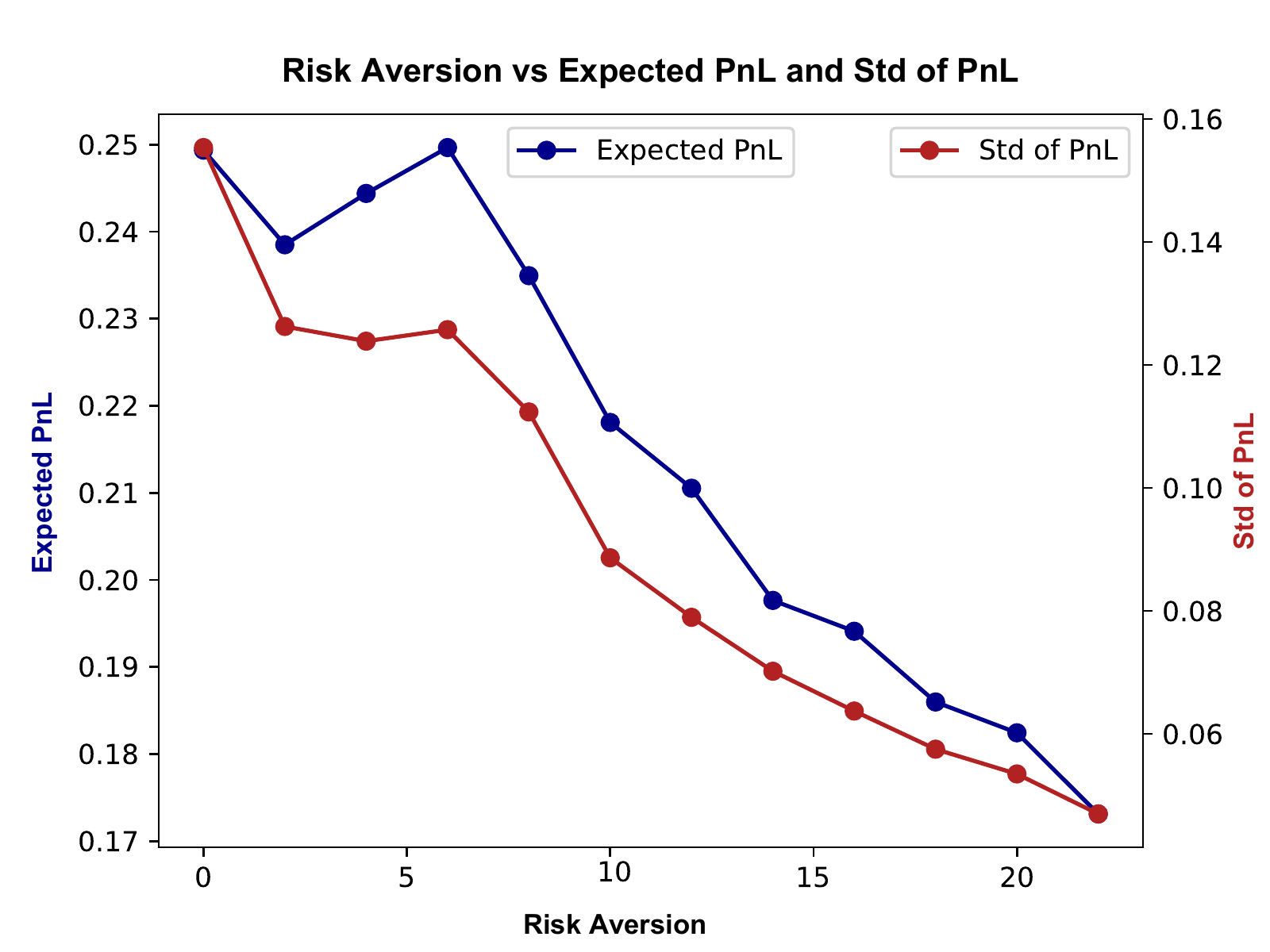}
    \caption{Expected PnL and standard deviation of PnL vs.~risk parameter $a$  for $(W,k,\lambda,\gamma) = (5,10,0.00025,0.01)$,  multiplicative factor $\theta = 1.002$, bucket spacing $\Delta = 1$.
    \label{fig:utility-std-vs-a}}
\end{figure}
\begin{figure}[h]
    \centering
         \includegraphics[width=0.7\textwidth]{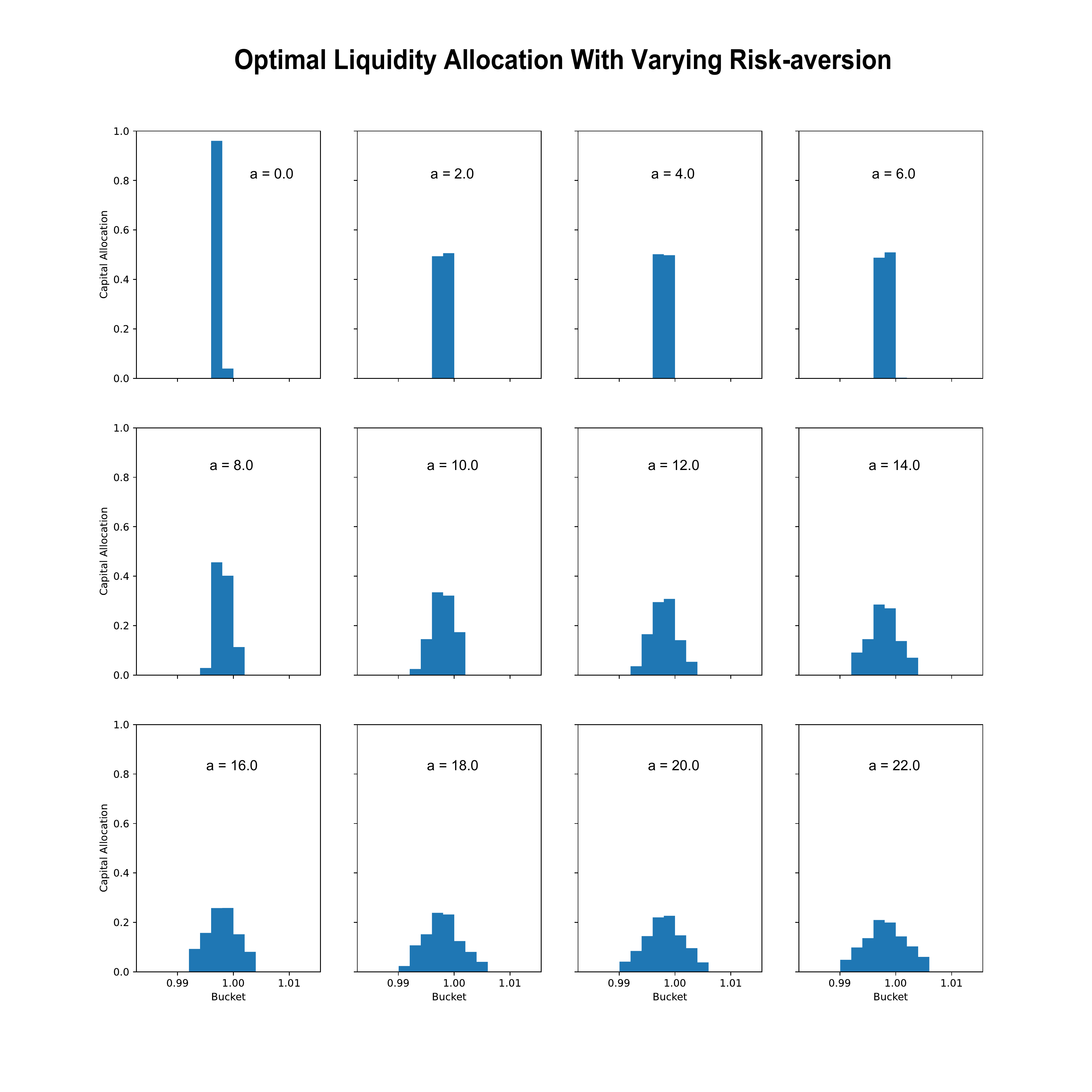}
    \caption{Optimal liquidity allocations for $(W,k,\lambda,\gamma) = (5,10,0.00025,0.01)$ and varying risk parameter, $a$.  Each bar represents the proportion of an LP's initial capital  allocated to a  bucket in the optimal liquidity allocation.
    \label{fig:prop_liquidity_allW}}
\end{figure}

Referring back to Figure~\ref{fig:Pareto}, we can also see that with an increased level of LP risk-aversion, the same set of exponential bucketing schemes gives rise to a steeper Pareto curve that is  composed of less points, indicating that different exponential bucketing schemes with similar gas costs give rise to a  wider spread of expected utility in the risk-averse setting. In addition, the risk-averse Pareto frontier  is still composed of multiple $(\theta,\Delta)$ values, which continues to indicate that a richer partition of the price space can  simultaneously benefit LPs and traders over a wide spread of risk-aversion values.

\subsection{Modulating Non-Arbitrage Trade and Fee Rates}
\label{sec:modulating-extra-params}

Parameters $k$ and $\lambda$ control the quantity of non-arbitrage trades and their impact on the contract price, respectively. The PnL and Gas costs are  monotonically increasing in each of parameters $k$ and $\lambda$, as  we would expect (see Figure~\ref{fig:utility-std-vs-k-lambda} which give results for both risk-neutral and risk-averse LPs). 
%
For PnL,  LPs only make a profit for non-arbitrage trades~\cite{capponi2021adoption}, hence this relationship is expected.
As for gas costs, it is clear that more contract price movements will 
trigger the crossing  of more active ticks. 
\begin{figure}[h]
    \centering
         \includegraphics[width=0.48\textwidth]{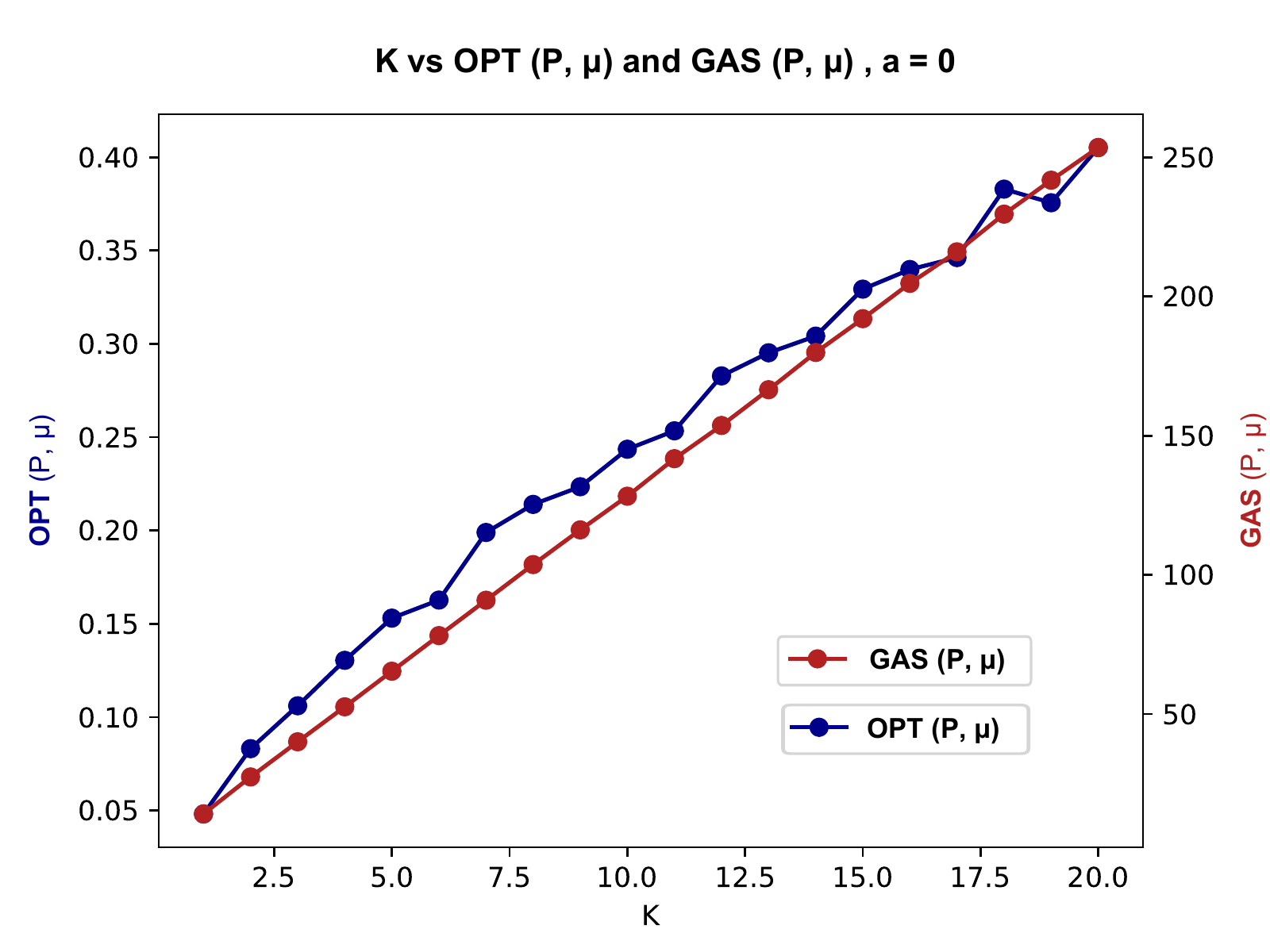}
         \includegraphics[width=0.48\textwidth]{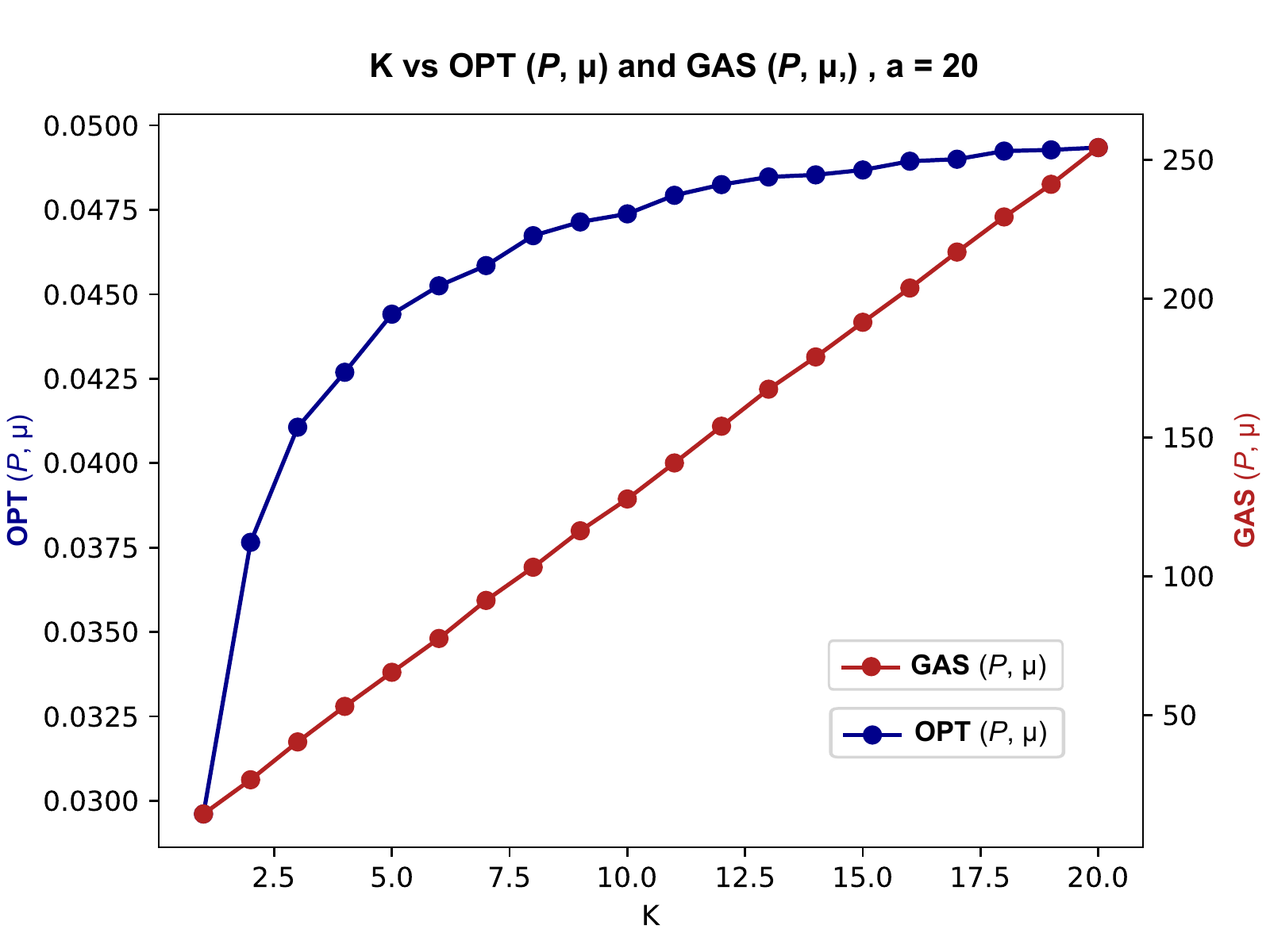}
          \includegraphics[width=0.48\textwidth]{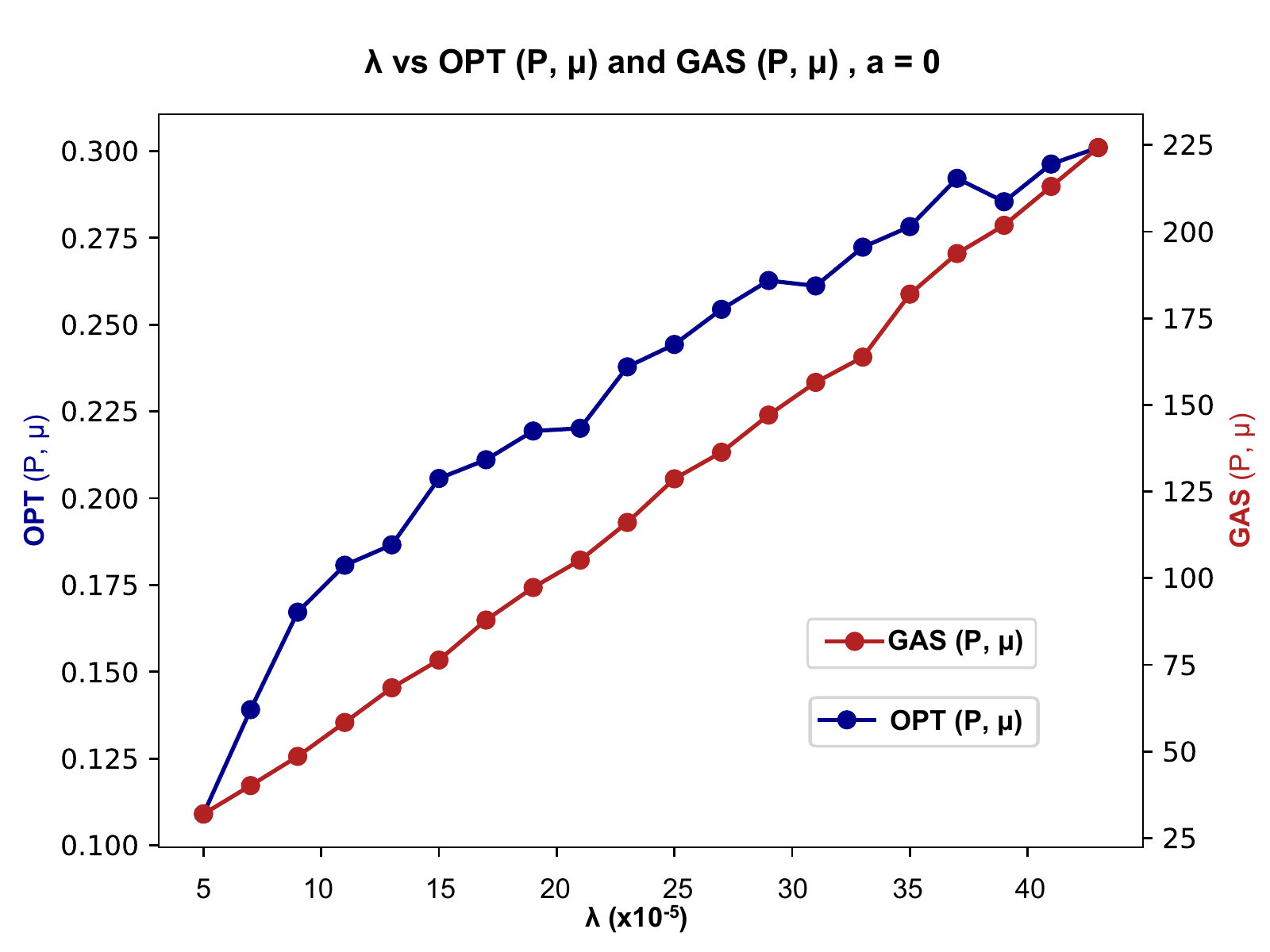}
         \includegraphics[width=0.48\textwidth]{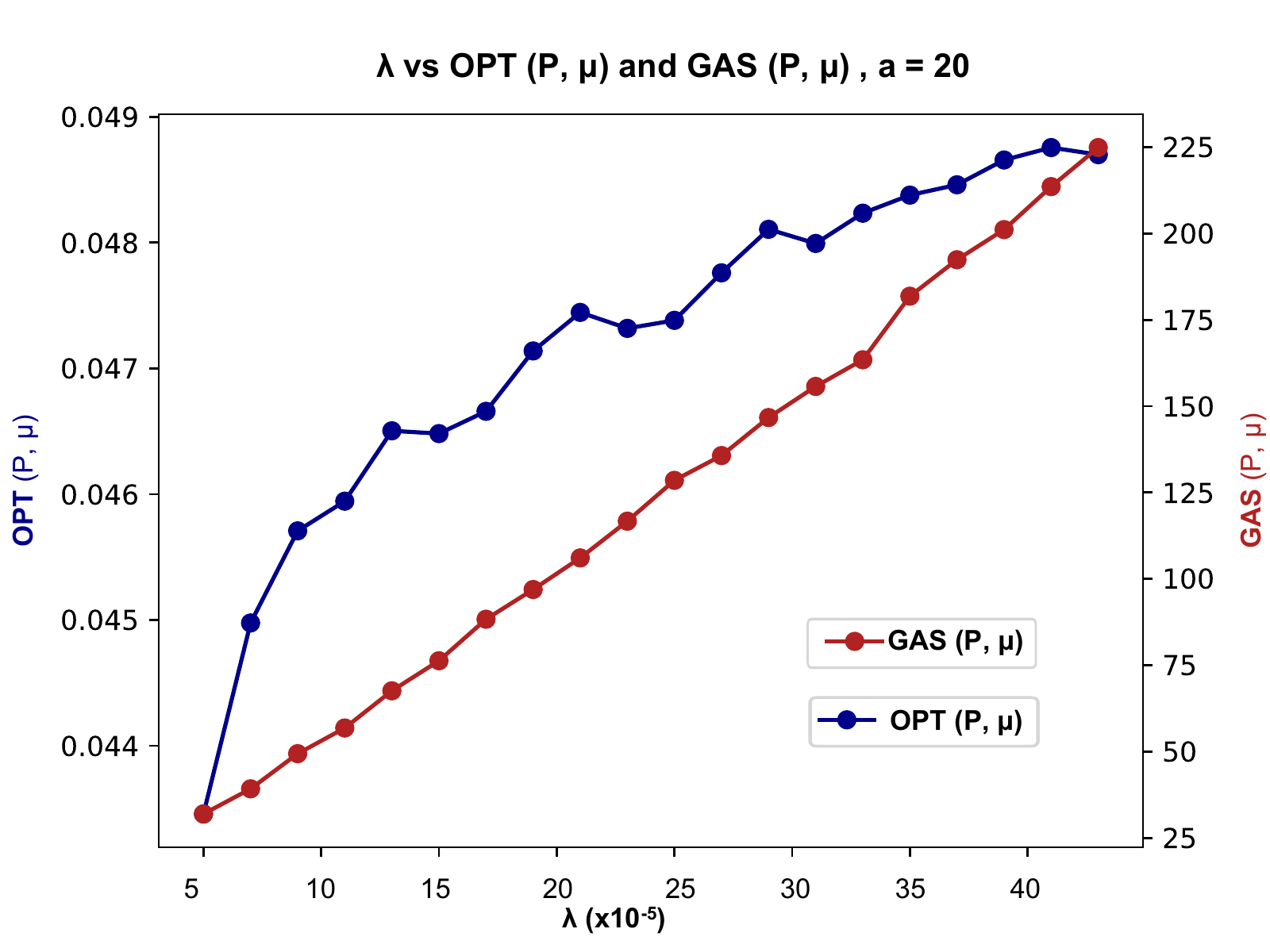}
    \caption{The top row plots expected PnL and Gas costs for an LP as a function of $k$ (the number of per-round non-arbitrage trades) and the bottom row plots expected PnL and Gas costs for an LP as a function of $\lambda$ (the scale of non-arbitrage trade). Left plots correspond to $a = 0$ (risk neutral LP) and right plots to $a = 20$ (risk-averse LP). For all plots, we let market prices follow a GBRW with bandwidth $W = 5$ and consider v3 contracts with fee rate $\gamma = 0.01$.
    \label{fig:utility-std-vs-k-lambda}}
\end{figure}

The transaction fee rate, $\gamma$, affects the LP PnL and the  Gas costs in multiple ways. On one hand, changing $\gamma$ directly affects LP profits, as price movements  provide more fees. On the other hand, $\gamma$ affects the no-arbitrage interval around a given market price, and thus the stochastic process governing the contract prices. In Figure~\ref{fig:fees-pareto},  we see that 
with a risk-neutral LP each of the different fee rates is included on the Pareto frontier.
This is quite different for  risk-averse LPs, with the highest fee rate, $\gamma = 0.2$ dominating the rest of the bucketing designs for $a\in \{10,20\}$.
To interpret this,  to earn the same expected utility as $\gamma = 0.2$ with  lower fee rates, the contract must make the sacrifice of having larger gas costs.  

On the basis of this analysis,  higher fee rates seem desirable for contract design, at least for the case of risk averse LPs. However, we caution that a higher fee rate also makes a trading pool less desirable for traders, and that this is a phenomenon that is outside of our current model.  For this,  we would need to model competition between pools, and the  interplay between LPs and pool desirability to traders, in order to obtain a more precise statement. 
%
\begin{figure}[t]
    \centering
         \includegraphics[width=0.48\textwidth]{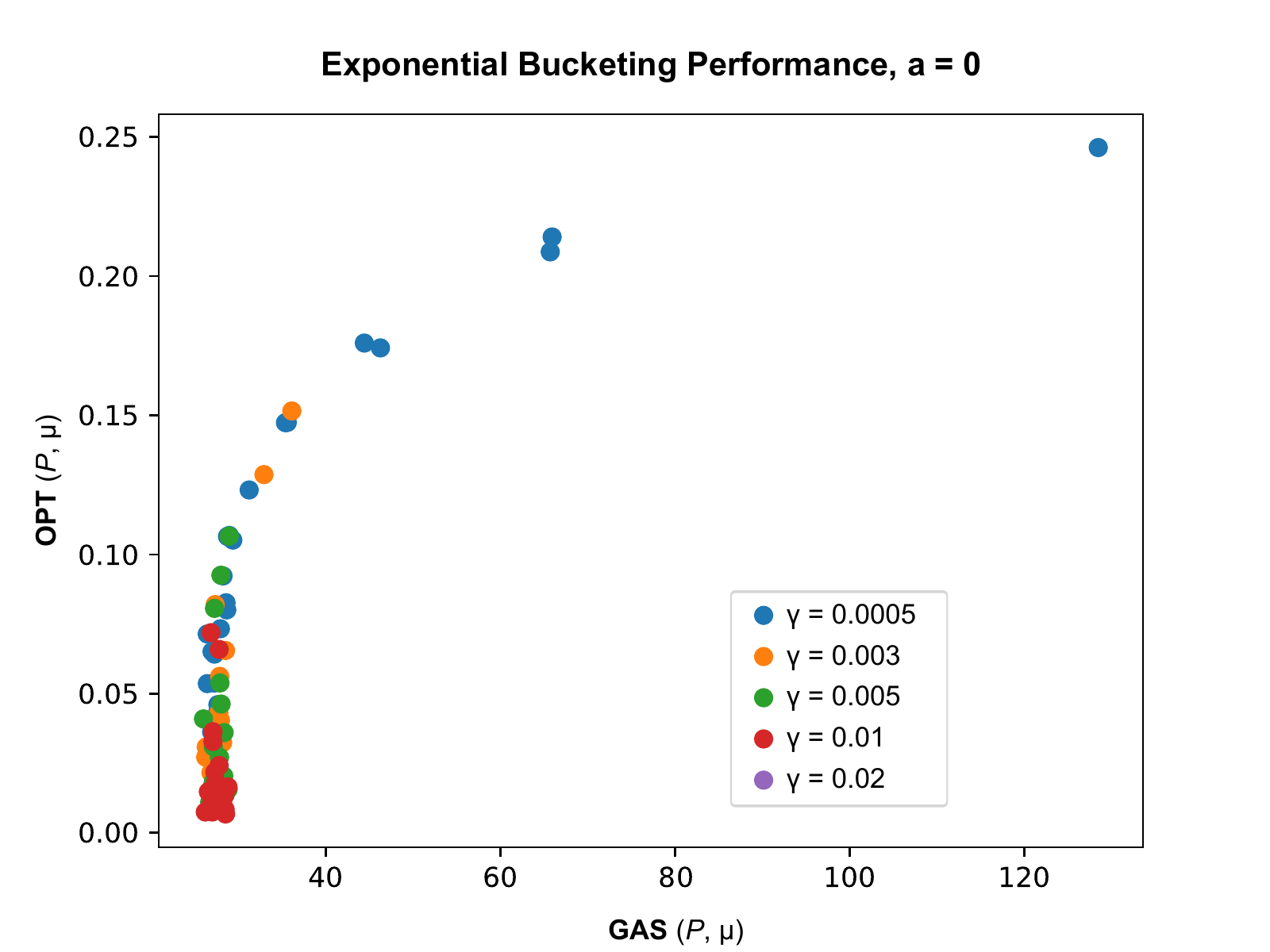}
         \includegraphics[width=0.48\textwidth]{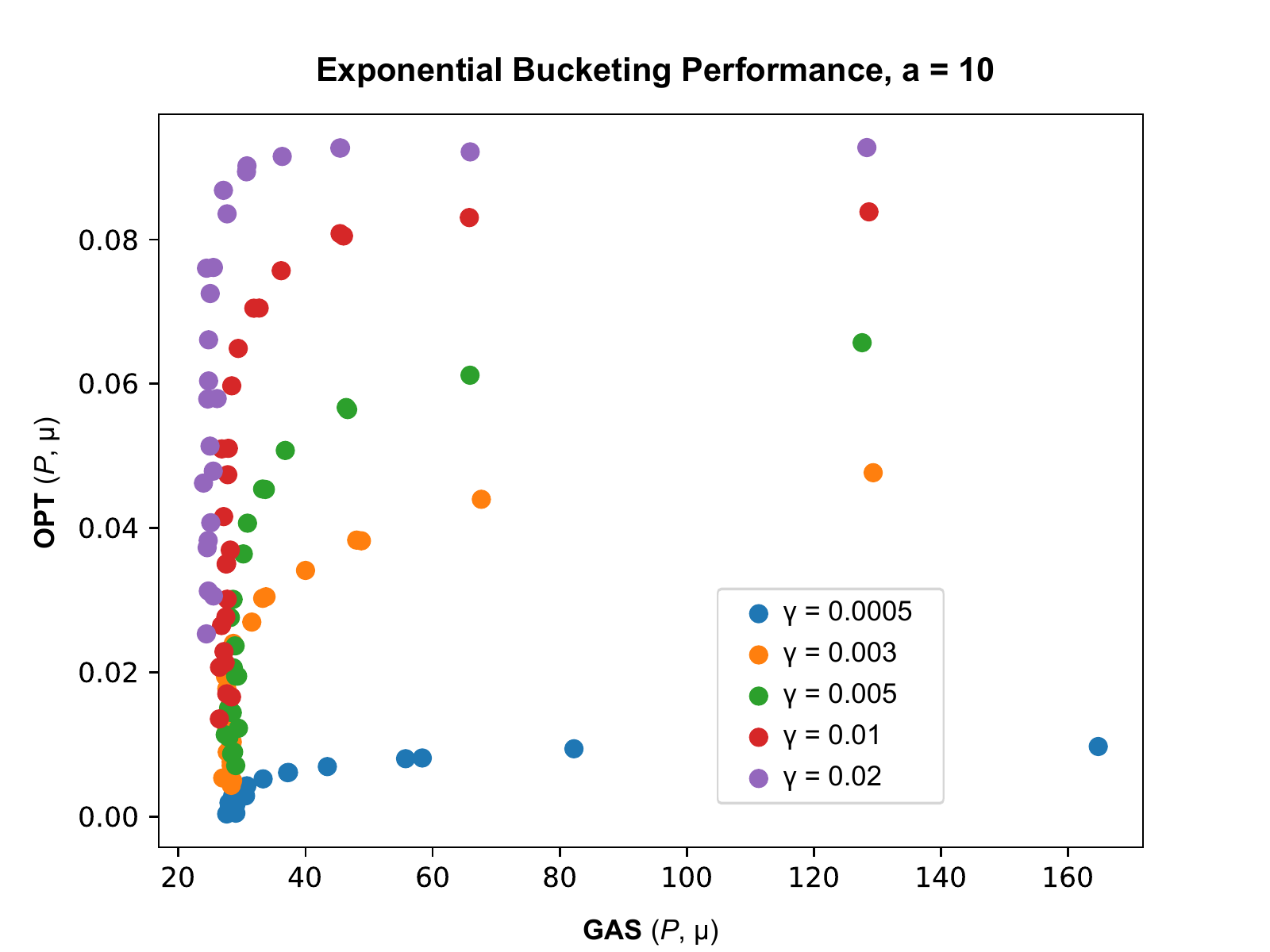}
         \includegraphics[width=0.48\textwidth]{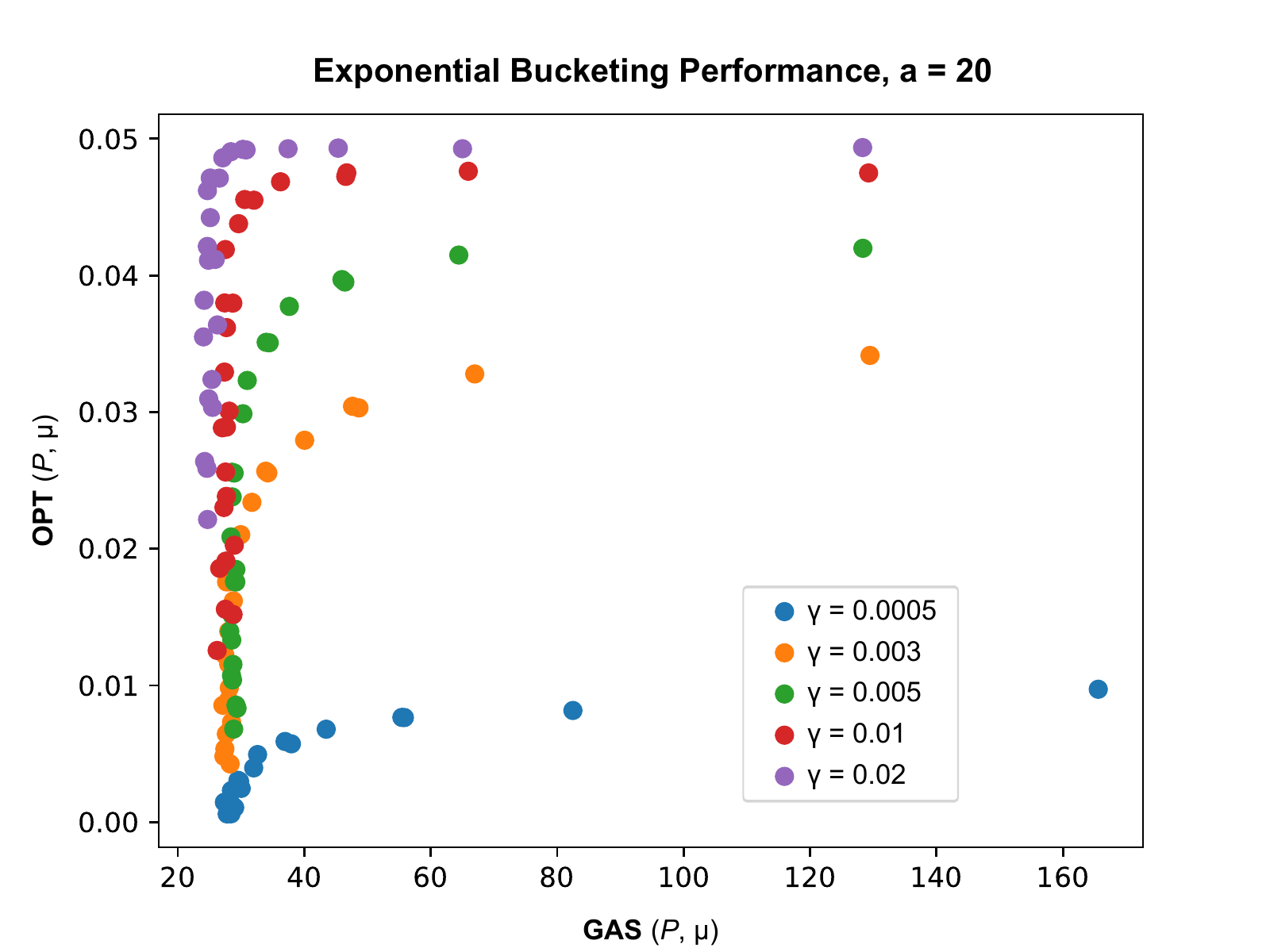}
    \caption{LP PnL vs.~Trader gas fees for $W = 5$, $k = 10$,  $\lambda = 0.01$,  LPs with risk parameter $a\in\{0,10,20\}$,  and for different values of transaction fee rate, $\gamma$. Each point corresponds to a specific $(\theta,\Delta)$, varying multiplicative factor $\theta \in \{1.002 + i\cdot 0.002\}_{i=0}^4$ and bucket spacing $\Delta \in \{1 + 2i\}_{i=0}^4$, and with  color indicating fee rate $\gamma \in \{0.0005, 0.003, 0.005, 0.01, 0.02\}$.
    \label{fig:fees-pareto}}
\end{figure}

\subsection{The Effect of Contract-Price Volatility}
\label{sec:price-vol-empirical}

In this section, we study the impact of different levels of contract-price volatility.   We recall that four parameters,  $W,k,\lambda$, and $\gamma$, govern  LP belief profiles. $k$ represents the number of non-arbitrage trades in a given round, and $\lambda$ represents the multiplicative change in price caused by a non-arbitrage trade. Lower values of each parameter reduce the overall contract price changes that arise from non-arbitrage trading. $W$ represents the volatility of the GBRW governing market prices and $\gamma$ represents the fee tier of a v3 contract. Larger values of $W$ and lower values of $\gamma$ increase the contract price changes that arise from arbitrage trading. 
%
Until now, we have focused on an intermediate volatility regime, given by $(W,k,\lambda,\gamma) = (5,10,0.00025,0.01)$. In what follows, we explore a low volatility regime given by $(W,k,\lambda,\gamma) = (3,5,0.0002,0.01)$ and a high volatility regime given by $(W,k,\lambda,\gamma) = (7,15,0.0003,0.01)$. 

\subsubsection{Implications for the OPT-GAS Pareto Frontier}
\label{sec:Pareto-results-volatility}

\begin{figure}[h]
    \centering
         \includegraphics[width=0.45\textwidth]{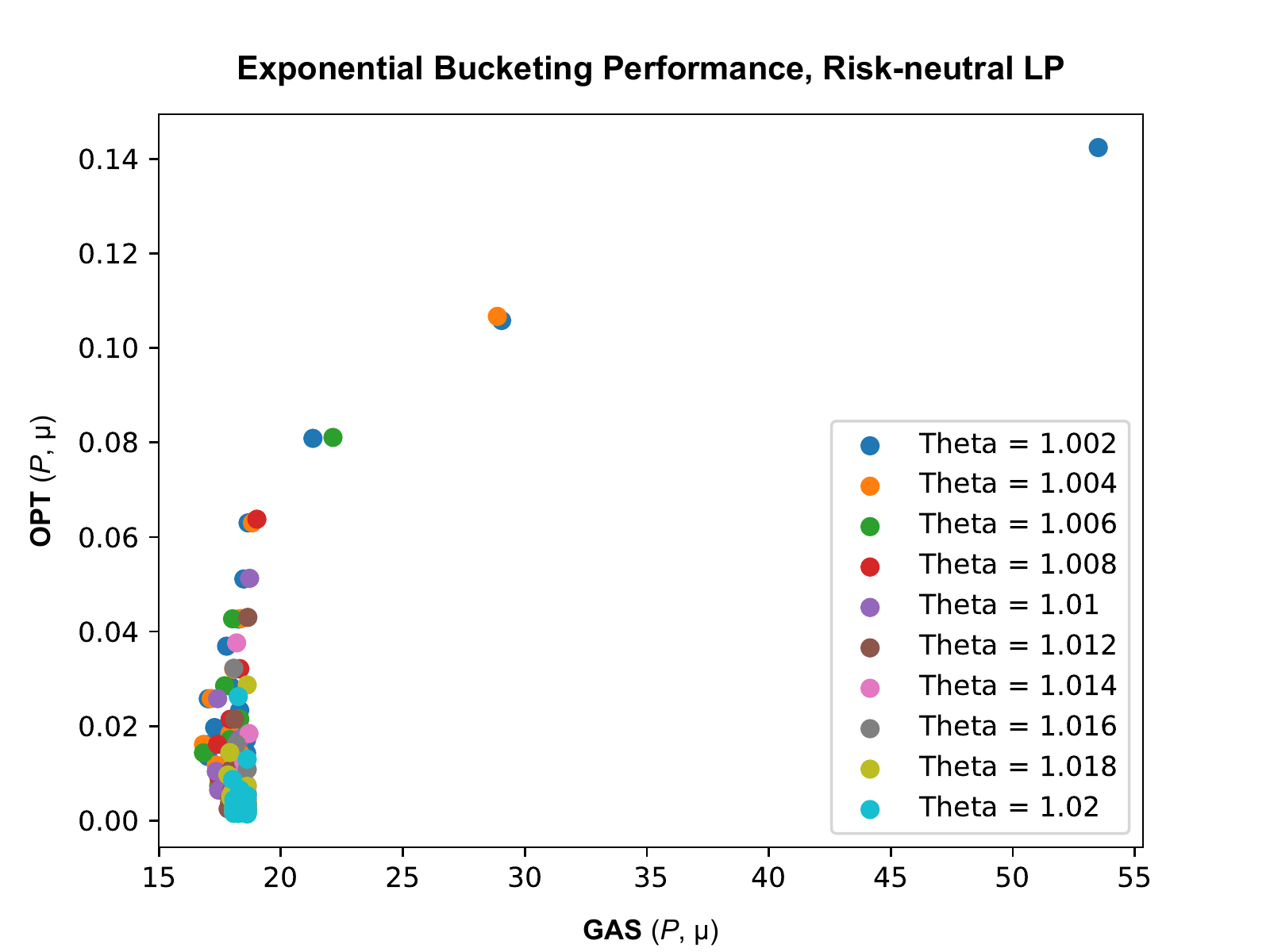}
         \includegraphics[width=0.45\textwidth]{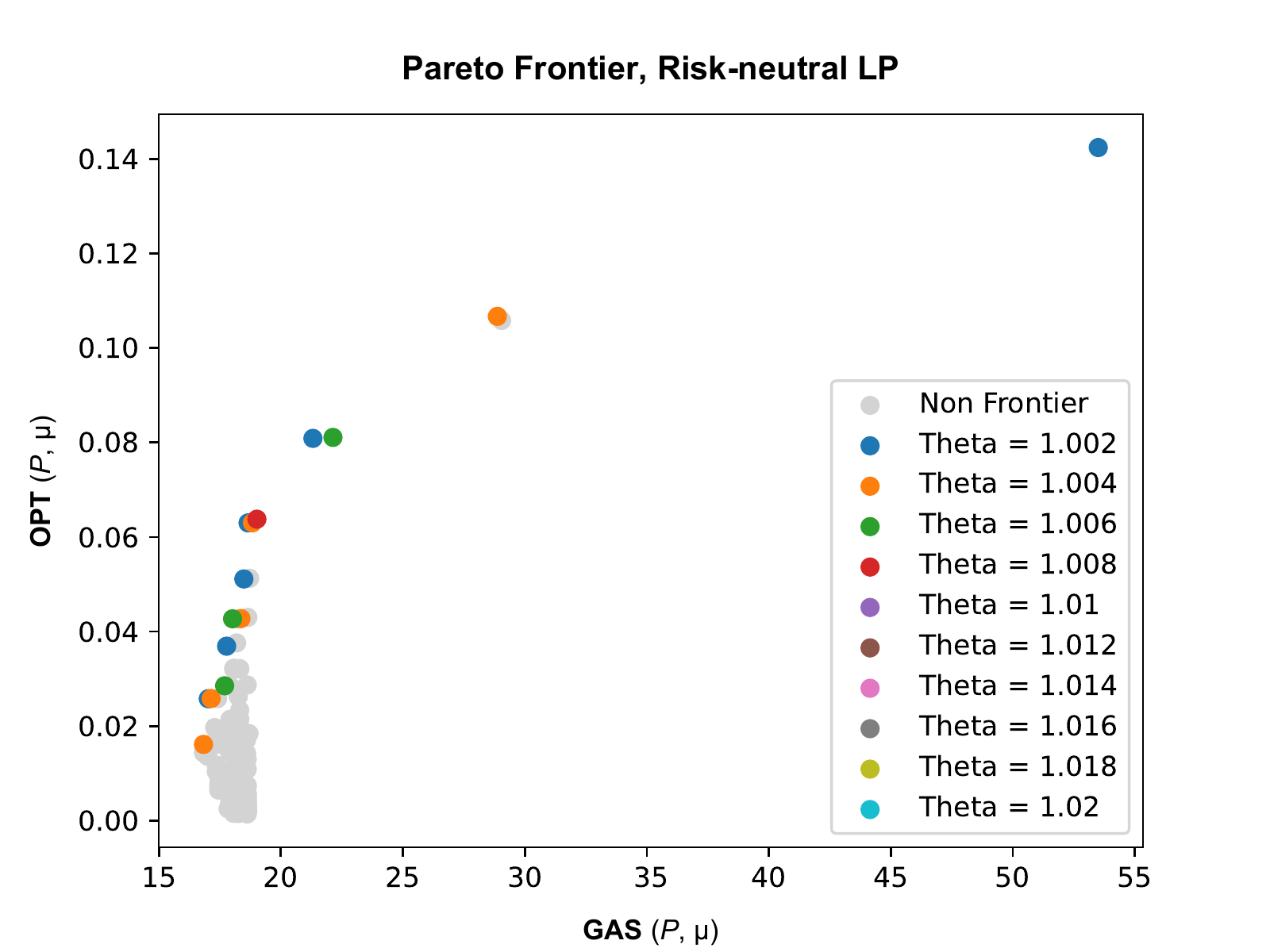}
         \includegraphics[width=0.45\textwidth]{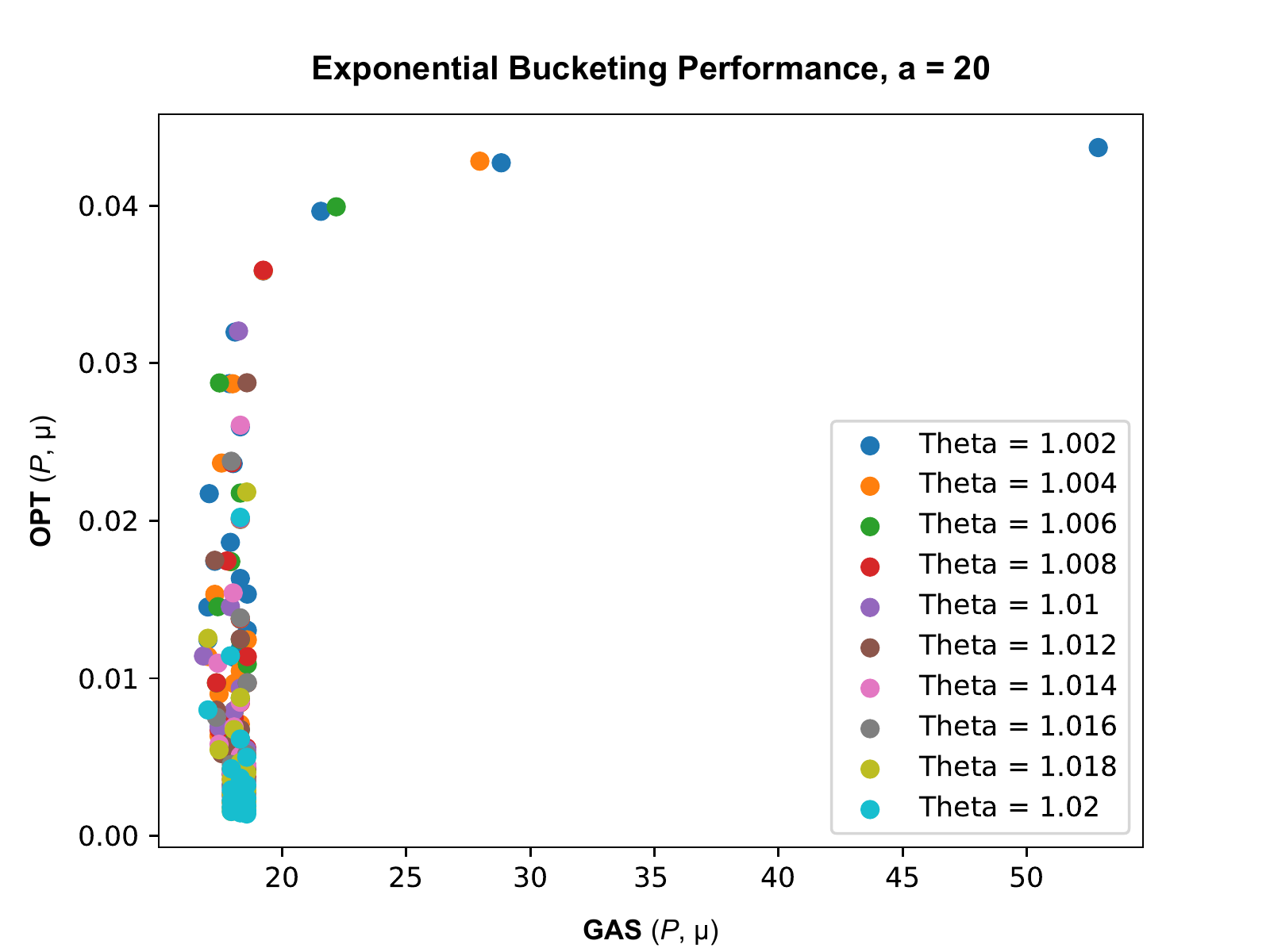}
         \includegraphics[width=0.45\textwidth]{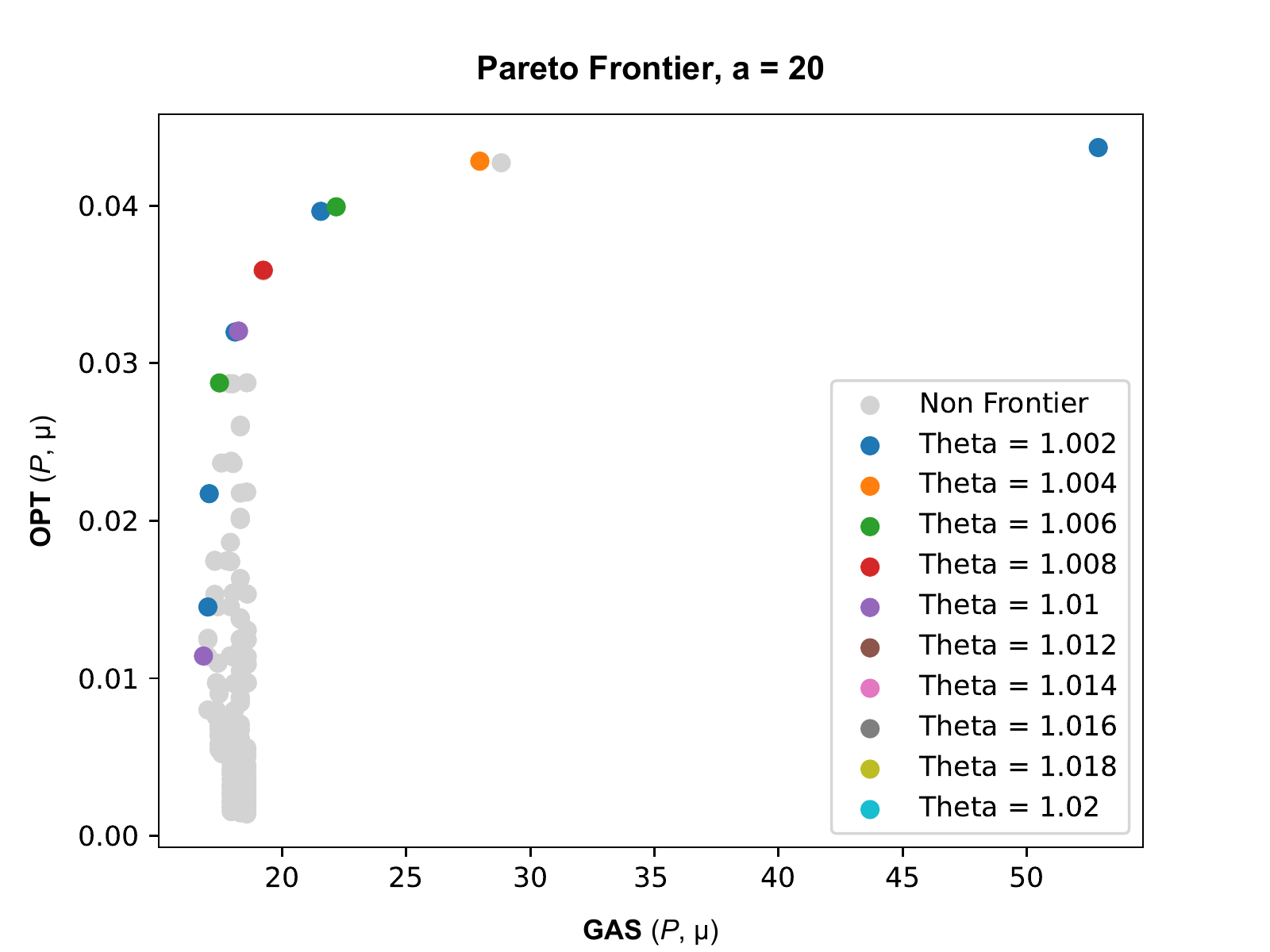}
    \caption{LP PnL vs.~Trader gas fees for low volatility regime $(W,k,\lambda,\gamma) = (3,5,0.002,0.01)$, LPs with risk parameter $a\in \{0,20\}$, and different  exponential bucketing schemes.  Each point corresponds to a specific $(\theta,\Delta)$, varying bucket spacing $\Delta \in \{1,\dots,20\}$, and with color  indicating  multiplicative factor ($\theta$). The right column  highlights only the bucketing schemes on the Pareto Frontier. 
    \label{fig:Pareto-low-vol}}
\end{figure}
\begin{figure}[h]
    \centering
         \includegraphics[width=0.45\textwidth]{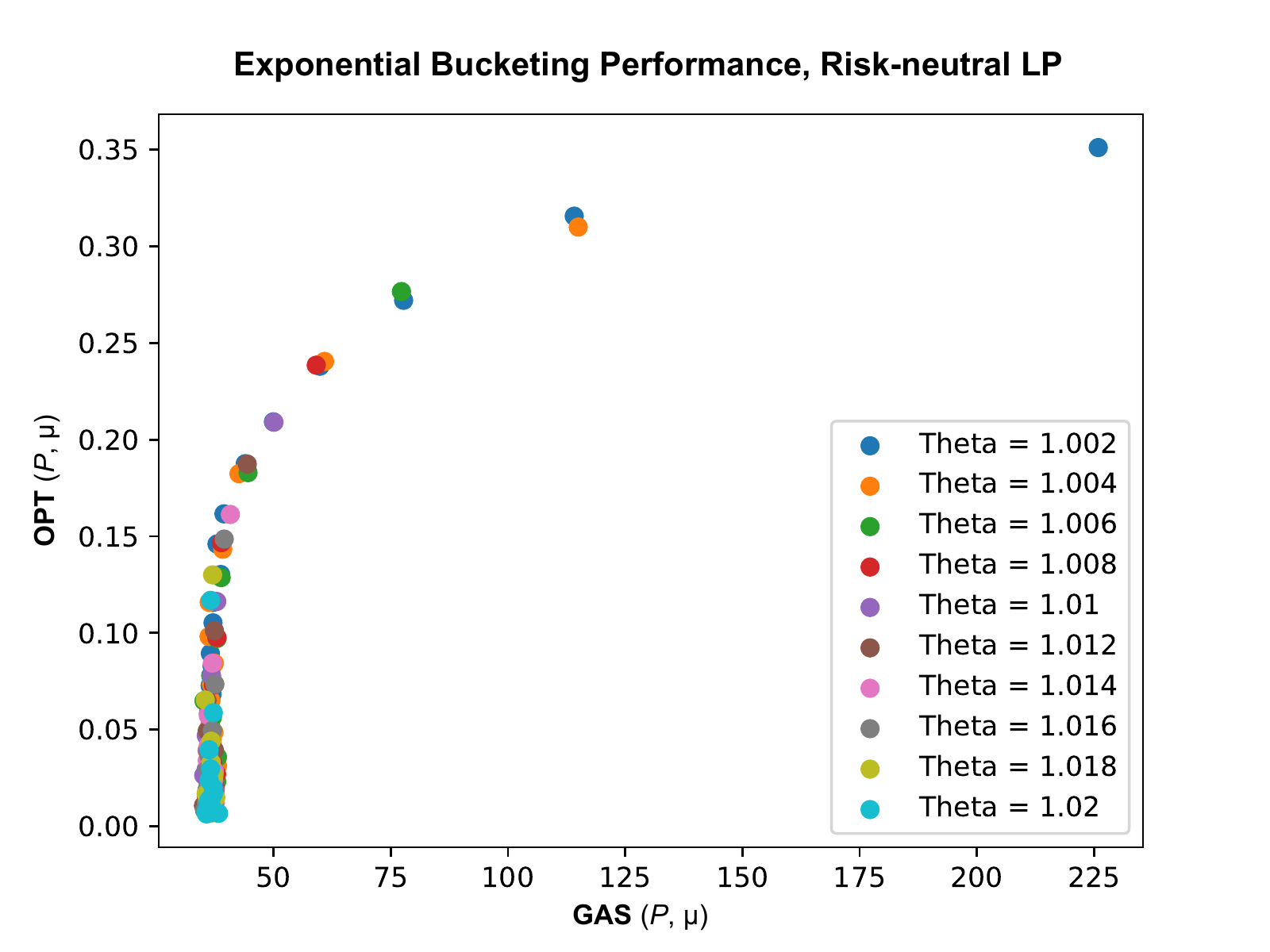}
         \includegraphics[width=0.45\textwidth]{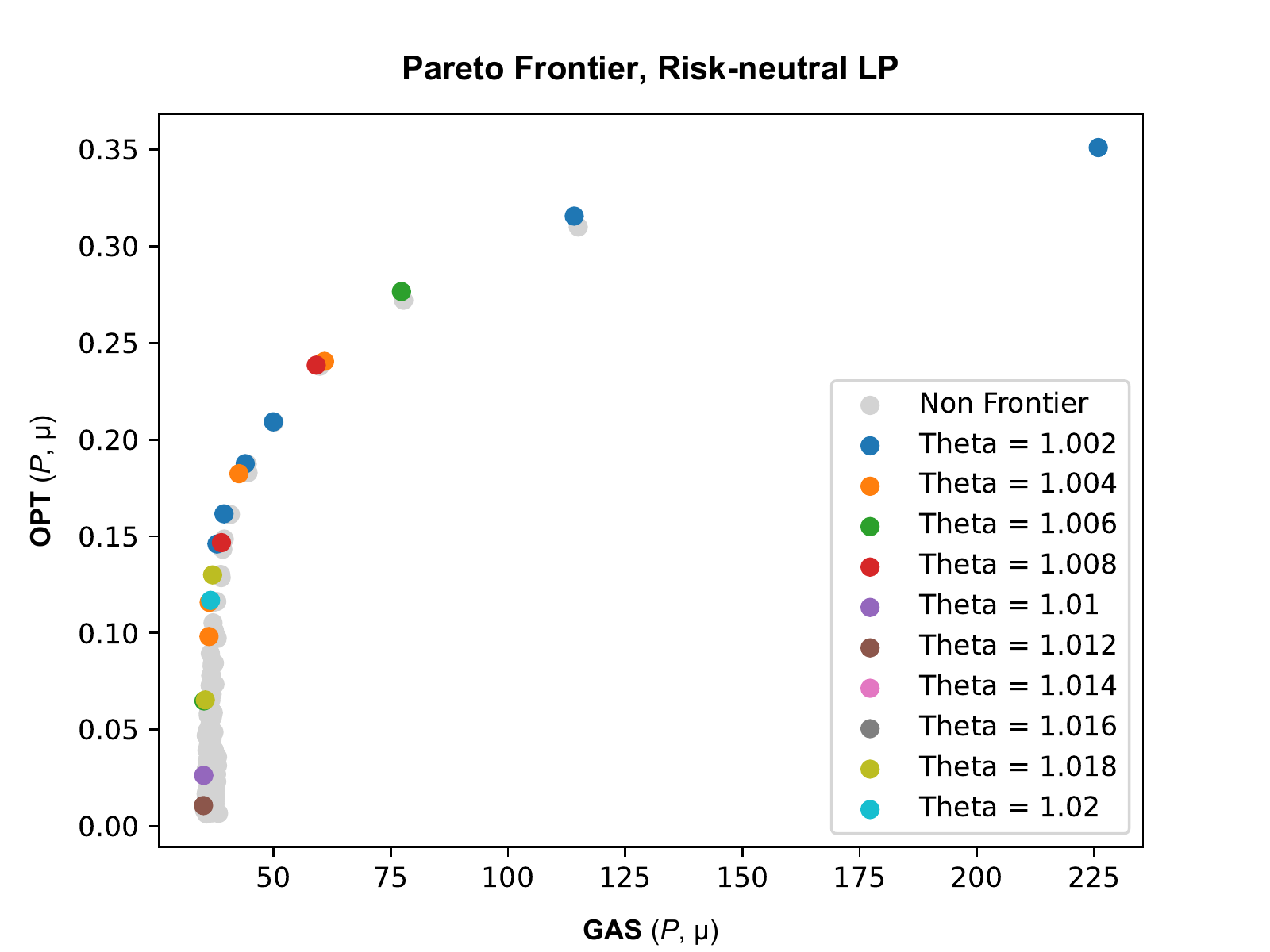}
         \includegraphics[width=0.45\textwidth]{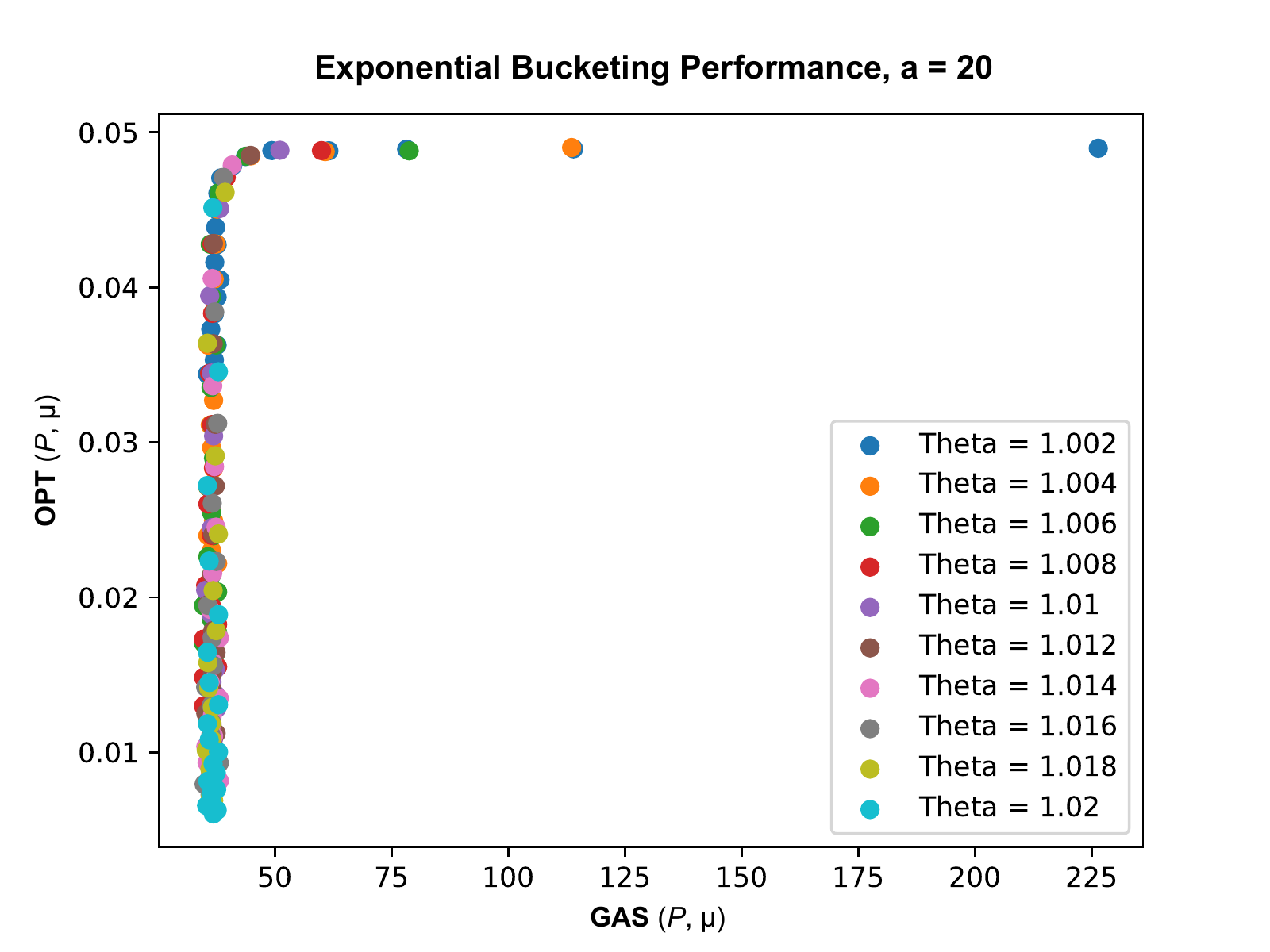}
         \includegraphics[width=0.45\textwidth]{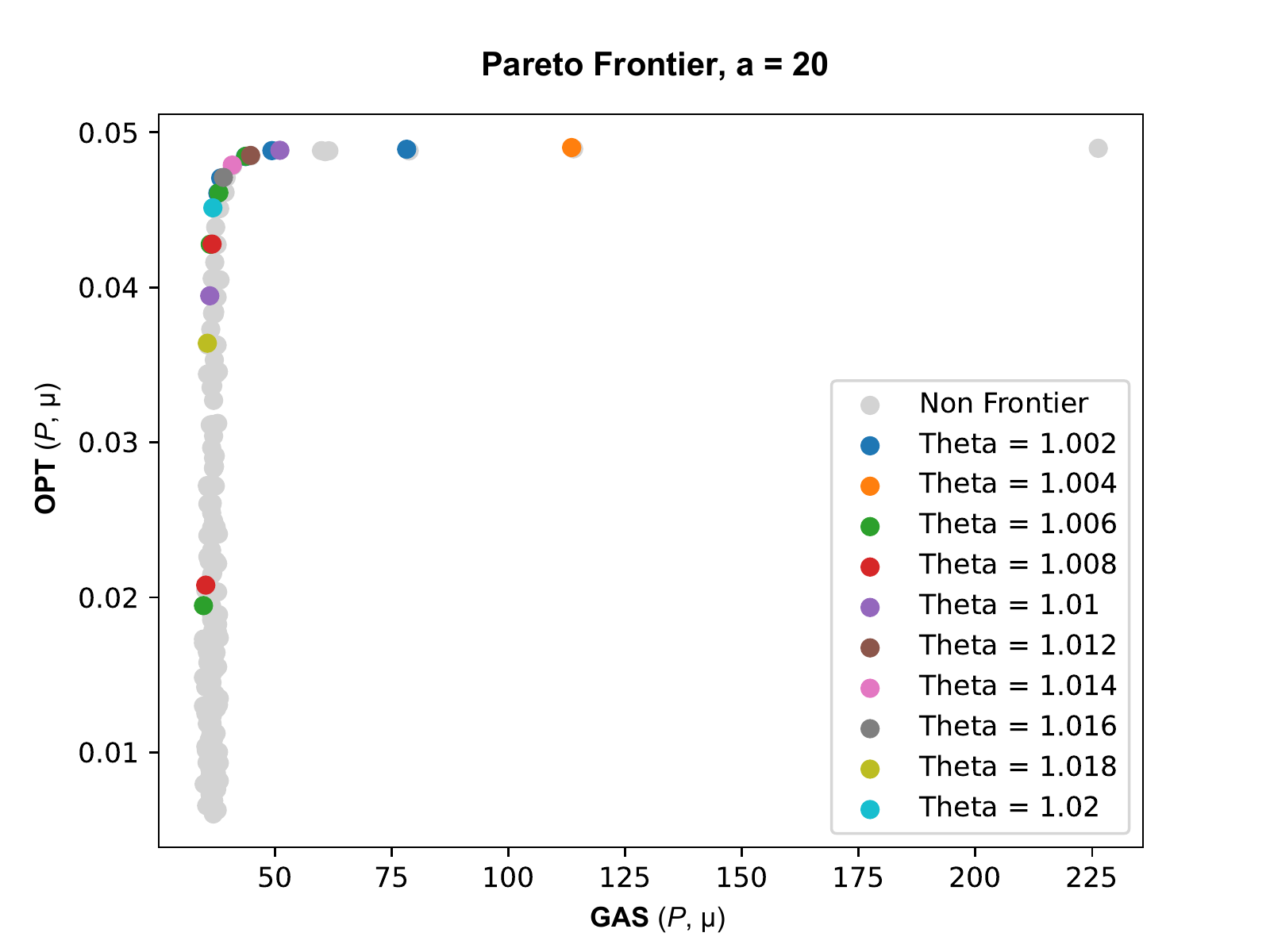}
    \caption{LP PnL vs.~Trader gas fees for high volatility regime $(W,k,\lambda,\gamma) = (7,15,0.003,0.01)$, LPs with risk parameter $a\in \{0,20\}$, and different $(\theta,\Delta)$-exponential bucketing schemes.  Each point corresponds to a specific $(\theta,\Delta)$, varying $\Delta \in \{1,\dots,20\}$ and with color  indicating $\theta$ value. The right column  highlights only the bucketing schemes on the Pareto Frontier.
    \label{fig:Pareto-high-vol}}
\end{figure}

In Figures~\ref{fig:Pareto-low-vol} and~\ref{fig:Pareto-high-vol}, we plot the OPT-GAS Pareto frontier for low volatility and high volatility contract price regimes, respectively. The main results of Sections~\ref{sec:Pareto-results} and~\ref{sec:riskaverse-results} are robust to variations in contract-price volatility. Most importantly, we see that multiple  different  bucketing scheme continue to lie on the Pareto frontier, again  showing the importance of allowing a  diverse set of price bucketing schemes for contract design, this extending to  different volatility regimes. 

Furthermore, we also see that for all levels of volatility, when risk-aversion increases, Pareto curves become more steep. In addition, the Pareto curves are also steeper at higher levels of volatility when risk-aversion is maintained constant. 
This confirms that different exponential bucketing schemes give rise to a wider spread of expected PnL for risk-averse LPs,
with similar gas costs, 
and with this this spread increasing in higher volatility contract-price regimes.

\subsubsection{Implications for Risk-averse LPs}

We also confirm that  both expected PnL and the standard deviation of PnL are decreasing functions in the risk-aversion parameter $a$ of an LP for each of the low-volatility and high-volatility contract-price regimes (Figure~\ref{fig:utility-std-vs-a-volatilities}). The main difference between volatility regimes is that the overall spread of PnL earned (as a function of risk-aversion) is larger for high volatility regimes. 

The fact that expected PnL and standard deviation of PnL are decreasing as functions of $a$ for all volatility regimes is to be expected for the constant absolute-risk aversion function. As for the fact that larger volatility regimes give rise to larger spreads of expected PnL (as a function of $a$), this is  explained in part by the fact that higher contract price volatility results in more fees earned by an LP (especially through increased non-arbitrage trades).  
\begin{figure}[h]
    \centering
         \includegraphics[width=0.48\textwidth]{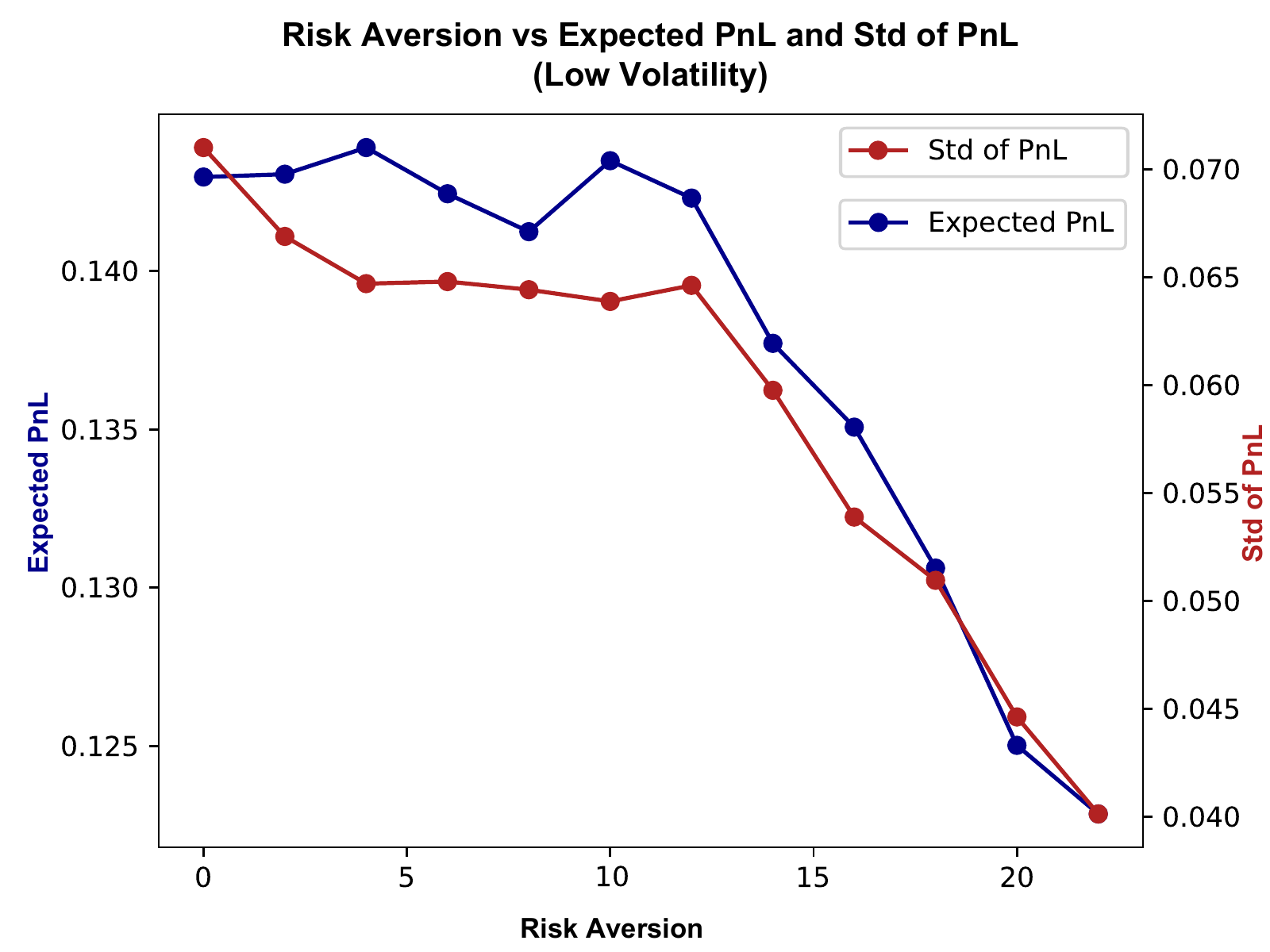}
         \includegraphics[width=0.48\textwidth]{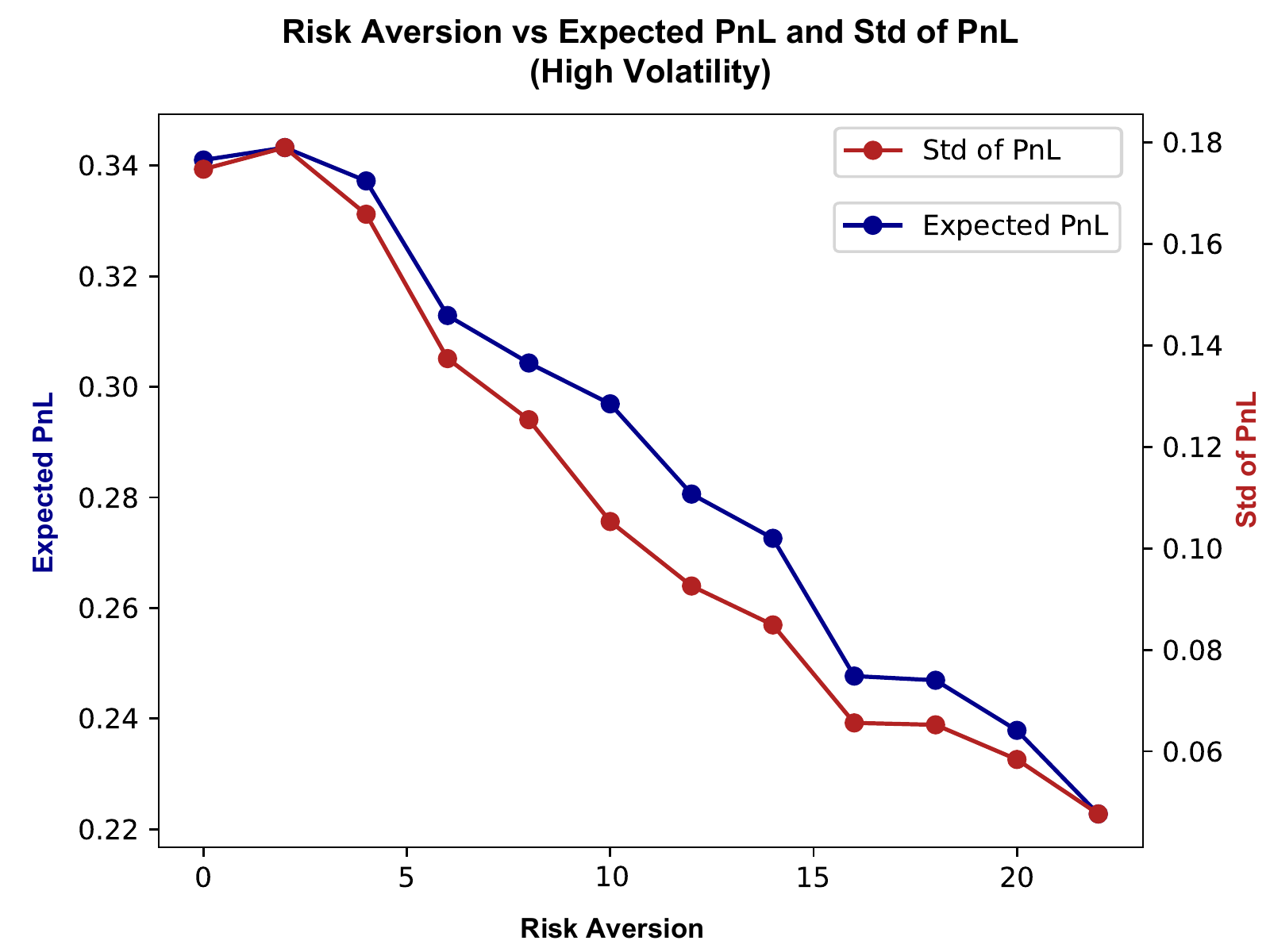}
    \caption{LP PnL and standard deviation of PnL as~risk parameter $a$ varies, for  low volatility regime $(W,k,\lambda,\gamma) = (3,5,0.0002,0.01)$ (left) and high volatility regime  $(W,k,\lambda,\gamma) = (7,15,0.0003,0.01)$ (right), and with exponential-bucketing scheme $\theta = 1.002$, $\Delta = 1$.
    \label{fig:utility-std-vs-a-volatilities}}
\end{figure}

We also plot optimal LP allocations for different risk-aversion values in Figure~\ref{fig:prop_liquidity_allW-volatilities}. The most interesting observation is that the spread of an LP's optimal allocation is more sensitive to risk-aversion in higher volatility contract-price regimes than lower volatility regimes. With the same  risk-aversion, an LP operating in a high-volatility regime has a more dispersed optimal liquidity allocation relative to operating in a low-volatility regime. This phenomenon makes  sense, as narrow liquidity allocations run a larger risk on missing out on transaction fees in a higher volatility contract-price regime. 
\begin{figure}[t]
    \centering
         \includegraphics[width=0.45\textwidth]{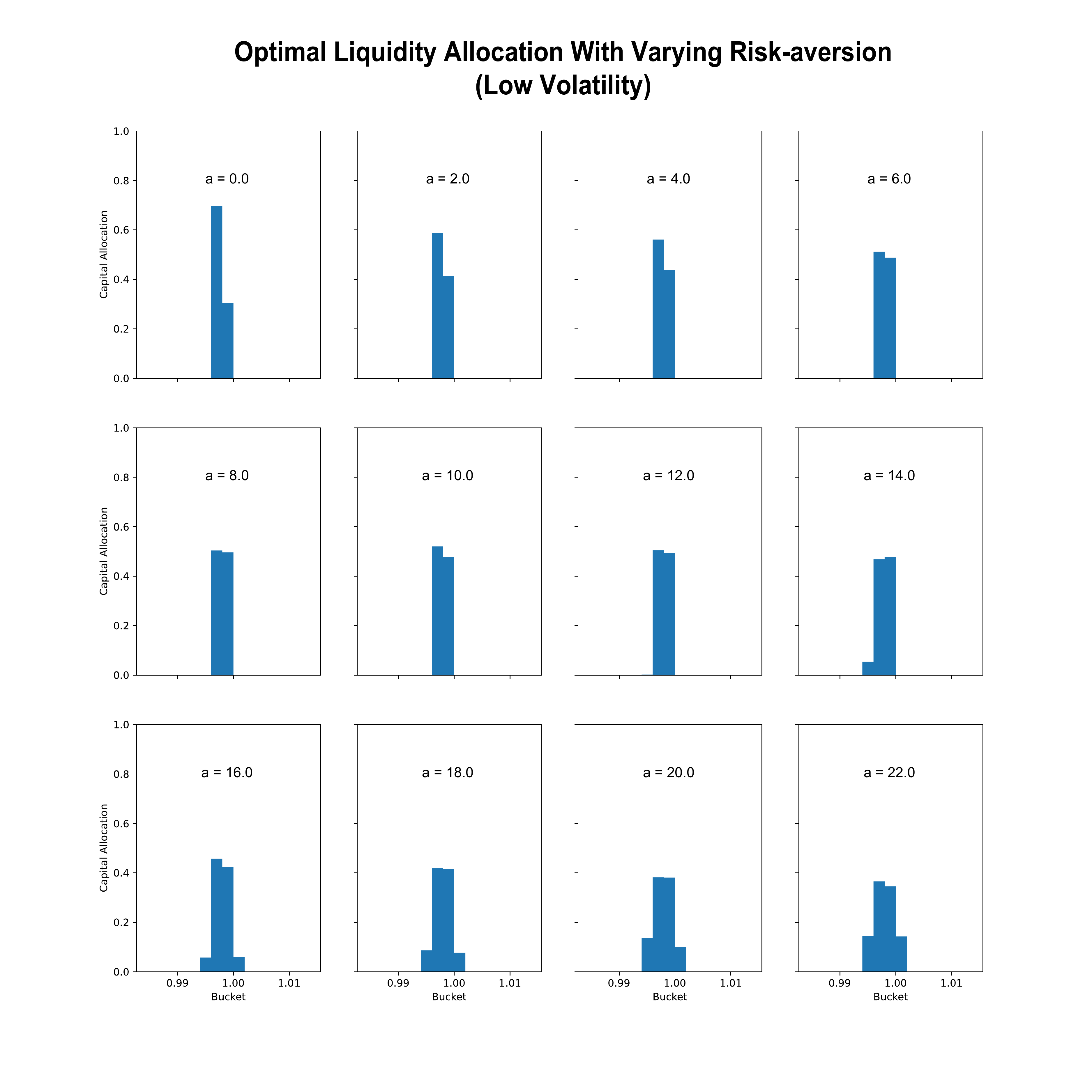}
         \includegraphics[width=0.45\textwidth]{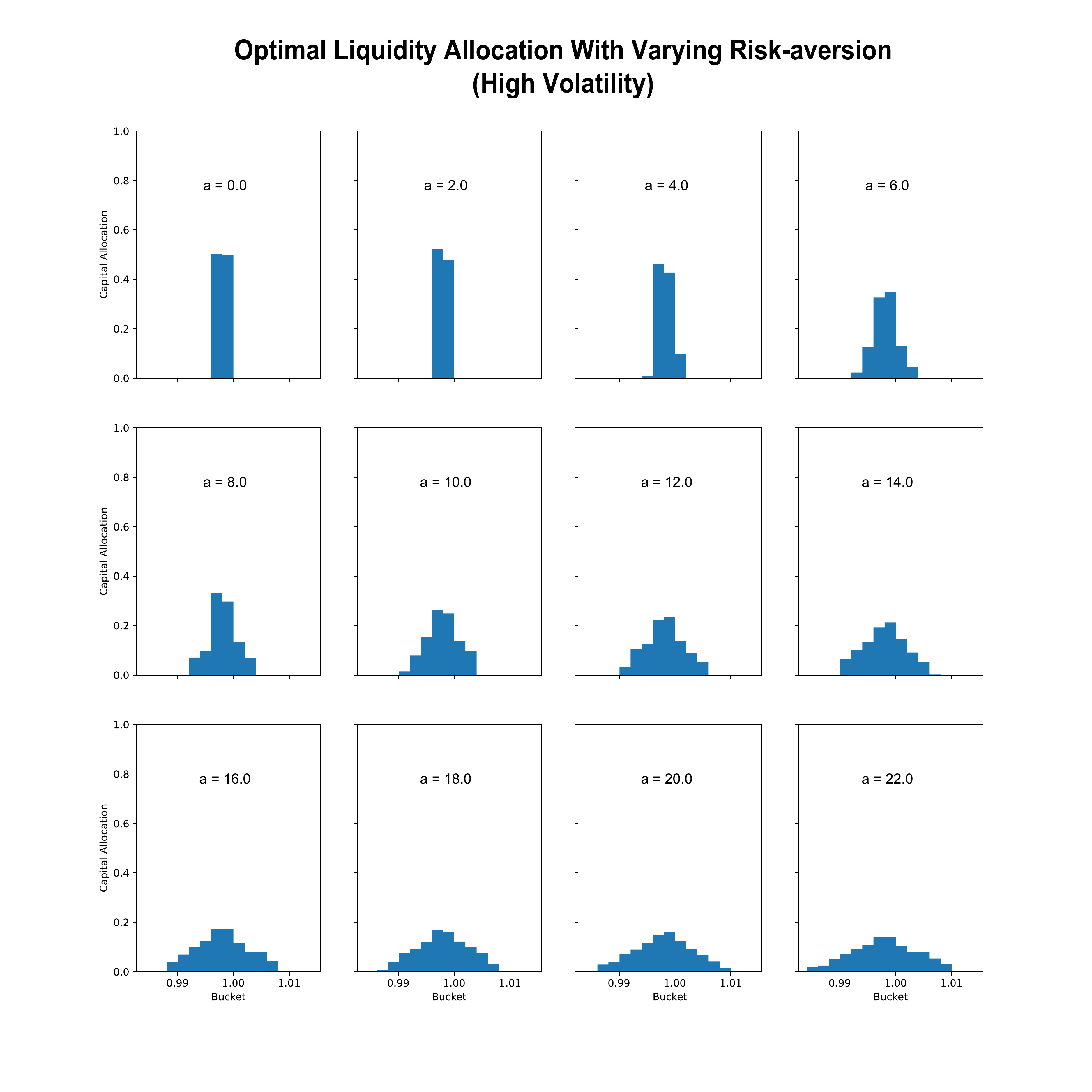}
    \caption{Optimal liquidity allocations for different risk-aversion parameters $a$ and low volatility regime $(W,k,\lambda,\gamma) = (3,5,0.0002,0.01)$ (left) and high-volatility regime $(W,k,\lambda,\gamma) = (7,15,0.0003,0.01)$ (right).  Each bar represents the proportion of an LP's initial capital allocated to a  bucket in the optimal liquidity allocation.  
    \label{fig:prop_liquidity_allW-volatilities}}
\end{figure}

\section{Conclusion}
\setcounter{equation}{0}

In this paper, we have developed a model for the tradeoffs that LPs are faced with in regards to how to optimally allocate liquidity in Uniswap v3 contracts. We  give explicit expressions for LP profit and loss that incorporate profits from fees accrued from traders as well as  impermanent loss from deviations in contract price. 
%
%
We have explored optimal liquidity provision strategies when LPs are endowed with stochastic beliefs over how  prices will evolve as well as differing degrees of risk aversion, developing linear and convex optimization formulations for LP investment and a method for computing expected Gas fees incurred by traders for a given bucketing scheme in a v3 contract. Contract prices in the model are induced by a stochastic model of market prices as well as non-arbitrage and arbitrage trade in the smart contract.

Adopting an empirically-informed belief profiles of  price changes,   we show  that LP PnL is maximized for smaller bucket sizes and a higher volume of non-arbitrage trades (Figures \ref{fig:delta-opt-gas} and \ref{fig:utility-std-vs-k-lambda})
but this comes at a cost of higher gas fees for traders. Viewed as a multi-objective optimization problem where we care about LP PnL as well as trader gas cost, we see the value of providing a more diverse set of price bucketing schemes than those  currently available in Uniswap v3. We have also developed initial insight into the effect of transaction fees on LP PnL, showing that
for Liquidity-independent LP beliefs, higher transaction fees  provide higher LP utility at potentially lower gas costs for traders if adopted together with a suitably optimized  bucketing scheme (Figure \ref{fig:fees-pareto}). As mentioned in Section \ref{sec:modulating-extra-params}, though this may suggest that higher fee rates are desirable from both the perspective of increasing LP PnL and decreasing trader gas fees, we caution that pool desirability for traders is also affected by gas fees, hence more comprehensive models are needed to provide further insight regarding the impact of fees on v3 contracts. Finally, we also study the impact risk-aversion has on LP PnL as well as corresponding optimal liquidity allocations. As seen in Figure \ref{fig:prop_liquidity_allW}, we find that as risk-aversion increases, LPs correspondingly spread their liquidity across larger prices ranges to reduce variance in PnL. Furthermore, this spread is increased for more volatile contract-market price sequences. 

This paper contributes to a growing body of work on differential liquidity provision  in Uniswap v3 contracts. 
We 
leave open many interesting directions to pursue. One would  relax the assumption, critical here in making the analysis  decision-theoretic rather than  game-theoretic, 
 of Liquidity-independent trade flow
(here, each LP can separately optimize for PnL given their belief profile).  One way to relax this assumption is to consider a model in which prices change via  trade dynamics that also incorporate trade volume. For example, if a bucket has a  large amount of liquidity  then moving the price through the bucket  requires a larger trade volume, thereby affecting price movement, at least in the short term. In this scenario, the fact that LP liquidity allocation affects price movements makes their actions interdependent, and thus requires an equilibrium analysis. 
Another means by which LP liquidity allocations can influence price dynamics is by indirectly signaling their own beliefs about how prices may evolve over time. Indeed our work is an initial step in this direction, as we can interpret the optimal liquidity allocation  as a  signal of an LP's belief profile. A fundamental question  is how to interpret the liquidity profile of  LPs with  locked assets in a  v3 contract, and whether this information can signal collective price beliefs of LPs.

In addition, this  work focuses on a family of relatively simple and short-term liquidity allocation strategies, whereby an LP allocates liquidity for fixed amount of time and then removes their liquidity and collects their fees for potential profit and loss. It will be interesting to study more complicated liquidity provision strategies, especially those that involve LPs actively re-allocating liquidity as prices evolve over time. Indeed this thread of work has been explored in Neuder et al.~\cite{NRMP-strategic-LP}, albeit without consideration to impermanent loss, and it will be interesting to link the  gas fees  LPs pay to reallocate liquidity to our notion of the gas cost incurred by traders for different
choices of bucketing regimes.

\section*{Acknowledgments}
We would like to thank anonymous reviewers for helping improve earlier versions of this work, as well as insightful conversations with research teams at Wisdomise, Maverick and Uniswap. This work is supported in part by two generous gifts to the Center for Research on Computation and Society at Harvard University to support research on applied cryptography and society. 
Our work also benefited from Microsoft Azure credits provided by the Harvard Data Science Initiative.

\bibliography{refs}
\bibliographystyle{plain}

\appendix

\section{Omitted Proofs from Section 2}
\label{appendix:sec-2-proofs}
\setcounter{equation}{0}

\Firstprop*
\begin{proof}
We work with the product $(x+x')(y+y') = xy + x'y' + x'y + xy'$.
The first two terms are equal to $L^2$ and $(L')^2$ respectively, since $(x,y) \in \mathcal{V}^{(2)}(L)$ and $(x',y') \in \mathcal{V}^{(2)}(L')$. We have $L = \sqrt{xy}$ and $L' = \sqrt{x'y'}$, and $LL' = \sqrt{xyx'y'}$. However, we can invoke the fact that $P = P'$, to get $x'y = xy'$, and thus $LL' = \sqrt{(x'y)^2} = x'y = xy'$. Putting everything together, we have:

\begin{align}
(x+x')(y+y') = L^2 + 2LL' + (L')^2 = (L+L')^2,
\end{align}

and $(x+x',y+y') \in \cR^{(2)}(L+L')$, as desired. As for  $(y+y')/(x+x') = P$, we have $y = Px$ and $y' = Px'$, and thus $(y+y')/(x+x') = P(x+ x')/(x+x') = P$. The case where we remove $(x',y')$ from the bundle $(x,y)$ is identical.
\end{proof}

\Secondprop*
\begin{proof}
Since we take the limit as $a \rightarrow 0$ and $b\rightarrow \infty$, we can suppose that $a < P < b$, in which case $\mathcal{V}^{(3)}(a,b,P) = (\Delta^x_{b,P}, \Delta^y_{a,P})$. We have $\lim_{b \rightarrow \infty} \Delta^x_{b,P} = \lim_{b\rightarrow \infty}[1/\sqrt{P} - 1/\sqrt{b}] = 1/\sqrt{P}$, and  $\lim_{a \rightarrow 0} \Delta^x_{a,P} = \lim_{a \rightarrow 0}[\sqrt{P} - \sqrt{a}] = \sqrt{P}$. Altogether, we have $\lim_{a \rightarrow 0,  b \rightarrow \infty} \mathcal{V}^{(3)}(a,b,P) = [1/\sqrt{P}, \sqrt{P}] =  \mathcal{V}^{(2)}(P)$.
\end{proof}

\Thirdprop*

\begin{proof}
The correspondence to the v2 revenue curve follows directly from  Definition \ref{def:v3-reserve-curve}. For the rest of the proof, we notice that for $(w,z) \in \mathcal{R}^{(2)}(L)$, we can use the fact that $w z = L^2$, to express the price $P =z/w$ in terms of $w$ or $z$ only: $P = L^2/w^2 = z^2/L^2$. We have  $\phi_{a.b}(x,y) \in \mathcal{R}^{(2)}(L)$, and as a point on the v2 reserve curve, it has an associated price  $P = \phi^B_{a,b}(y)/\phi^A_{a,b}(x)$, which we now know is equivalent to $L^2/\phi^A_{a,b}(x)^2$. Thus, when we treat $P$ as a function of $x$, we see that it is a decreasing, injective function. Furthermore, it is easy to check that at extremal bundles on the v3 reserve curve, i.e., $(\Delta^x_{b,a},0), (0,L\Delta^{y}_{a,b}) \in \mathcal{V}^{(3)}(L,a,b)$, we get $P$ values of $a$ and $b$ respectively. It thus follows that for all $(x,y) \in \mathcal{V}^{(3)}(L,a,b)$, $P = \phi^B_{a,b}(y)/\phi^A_{a,b}(x) \in [a,b]$. This, along with Definition \ref{def:a-b-value}, gives  $\mathcal{V}^{(3)}(L,a,b,P) = (L\Delta^x_{b,P}, L\Delta^y_{a,P})$. Let us focus on the first term in the bundle: $L\Delta^x_{b,P} = L/\sqrt{P} - L/\sqrt{b}$. From before, we know that $P = L^2/\phi^A_{a,b}(x)^2$, hence $\sqrt{P} = L/\phi^A_{a,b}(x)$. It follows that $L\Delta^x_{b,P} = \phi^A_{a,b}(x) - L/\sqrt{b} = x$ as desired. The same argument holds for $L \Delta^y_{a,b}$,  albeit using  substitution $P = \phi^B_{a,b}(y)^2/L^2$.
\end{proof}

\Fourthprop*

\begin{proof}
We recall that $\cV^{(3)}$ is linear in $L$, hence it suffices to prove the statement for $L = 1$. We begin by considering the scenario where 
$a < P < c < b$, in which case $\cV^{(3)}(a,c,P) = (\Delta^x_{c,P},\Delta^y_{a,P})$ and $\cV^{(3)}(c,b,P) = (\Delta^x_{b,c},0)$. 

\begin{align}
    \cV^{(3)}(a,c,P) + \cV^{(3)}(c,b,P) &= (\Delta^x_{c,P},\Delta^y_{a,P}) + (\Delta^x_{b,c},0)   \\ 
    &= (\Delta^x_{b,c} + \Delta^x_{c,P}, \Delta^y_{a,P}) \nonumber \\
    &= (\Delta^x_{b,P}, \Delta^y_{a,P}) \nonumber \\
    &= \cV^{(3)}(a,b,P) \nonumber
\end{align}

The scenario where $a<c<P<b$ is almost identical to that above, hence we continue to the case where $P < a < c <b$:

\begin{align}
    \cV^{(3)}(a,c,P) + \cV^{(3)}(c,b,P) &= (\Delta^x_{c,a},0) + (\Delta^x_{b,c},0)   \\ 
    &= (\Delta^x_{b,c} + \Delta^x_{c,a}, 0) \nonumber \\
    &= (\Delta^x_{c,a}, 0) \nonumber \\
    &= \cV^{(3)}(a,b,P) \nonumber
\end{align}

The final case where $a < c < b < P$ is almost identical.

\end{proof}

\section{Omitted Proofs from Section 3}
\label{appendix:sec-3-proofs}.                      
\setcounter{equation}{0}

\Fifthprop*

\begin{proof}
For the entirety of this proof we focus on the claims for v3 impermanent loss. The results for v2 hold for the same reason that they hold for v3 in the special case when the beginning and end price $P,P'$ lie in the bucket for which liquidity has been obtained. 

Since $IL^{(3)}(\ell,P,P')$ is a linear combination of $IL^{(3)}(\ell_i,B_i,P,P')$ terms, it suffices to show that $IL(\ell_i,B_i,P,P')$ is linear in $\ell_i$. We recall that in Section \ref{sec:v3-description}, we showed that $\mathcal{V}^{(3)}(L,a,b,P) = L \cdot \mathcal{V}^{(3)}(a,b,P)$. Furthermore, we also recall we defined $\mathcal{B}((x,y),P) = Px + y$ as the token $B$ worth of a bundle $(x,y)$ when token $A$ has price $P$. Combining these two facts, it is clear that 

\begin{align}
\mathit{IL}^{(3)}(\ell_i,B_i,P,P') = \ell_i(v^{(3)}_h(1,B_i,P,P') - v^{(3)}_p(1,B_i,P')) = \ell_i \cdot \mathit{IL}(1,B_i,P,P'),
\end{align}
which establishes our desired linearity.

As for the remainder of the proposition, we assume that market price and contract price coincide, hence we slightly abuse notation going forward, so that $P = P_c = P_m$ and $P' = P_c' = P_m'$. To prove non-negativity of impermanent loss, we will show that $IL^{(3)}(\vect{\ell},P,P') = \sum_{i=-m}^n IL^{(3)}(\ell_i,B_i,P,P') \geq 0$ by showing that it is always the case that $IL^{(3)}(\ell_i,B_i,P,P') \geq 0$. We can expand the left hand side of the inequality to see that this boils down to showing that $v_h^{(3)}(\ell_i,B_i,P,P') \geq v_p^{(3)}(\ell_i,B_i,P')$. which given the linearity of each of these terms in $\ell_i$, is further equivalent to showing that $v_h^{(3)}(1,B_i,P,P') \geq v_p^{(3)}(1,B_i,P')$. We prove this statement by showing that the function $v_p^{(3)}(1,B_i,P')$ is  differentiable and concave in $P'$, and that $v_h^{(3)}(1,B_i,P,P')$ is   the tangent to $v_p^{(3)}(1,B_i,P')$ at $P' = P$, hence the inequality holds. We can expand the expression for $v^{(3)}_p(1,B_i,P')$ as follows:

\begin{align}
    v^{(3)}_p(1,B_i,P') &= 
    \begin{cases}
        P' \cdot \Delta^x_{b_i,a_i}  & \quad \mbox{if $P' < a_i$} \\
        \Delta^y_{a_i,b_i}  & \quad  \mbox{if $P' > b_i$} \\
        P' \cdot \Delta^x_{b_i,P'} + \Delta^y_{a_i,P'}  & \quad  \mbox{if $P' \in [a_i,b_i]$}.
    \end{cases}
\end{align}

We begin by considering the third case of the piece-wise definition above. I.e. $P' \in [a_i,b_i]$, which in turn gives $v^{(3)}_p(1,B_i,P') = P' \cdot \Delta^x_{b_i,P'} + \Delta^y_{a_i,P'}$. If we simplify the expression, we obtain: 
\begin{align}
    v_p^{(3)}(B_i,P') &= P' \cdot \Delta^x_{b_i,P'} + \Delta^y_{a_i,P'}  \\
    &= P' \left( \frac{1}{\sqrt{P'}} - \frac{1}{\sqrt{b_i}} \right) + (\sqrt{P'} - \sqrt{a_i})  \nonumber \\
    &= - \frac{1}{\sqrt{b_i}}P' + 2\sqrt{P'} - \sqrt{a_i} \nonumber 
\end{align}
From this, we obtain the first and second derivatives of the expression:
\begin{align}
    \frac{d}{dP'} v_p^{(3)}(B_i,P') = \frac{1}{\sqrt{P'}} - \frac{1}{\sqrt{b_i}} = \Delta^x_{b_i,P'} , \qquad \frac{d^2}{dP'^2} v_p^{(3)}(B_i,P') = - \frac{1}{2P'^{3/2}}.
\end{align}
From this, we have that $v_p^{(3)}(B_i,P')$ is differentiable for all $P'$, as derivatives match over the piece-wise definition of the function, furthermore, the only non-linear component (when $P' \in [a_i,b_i]$) is smooth and concave from the negative second derivative above.

Finally, we expand the expression for $v_h^{(3)}(1,B_i,P,P')$: 
\begin{align}
    v^{(3)}_h(1,B_i,P,P') &= 
    \begin{cases}
        P' \cdot \Delta^x_{b_i,a_i}  & \quad \mbox{if $P < a_i$} \\
        \Delta^y_{a_i,b_i}  & \quad  \mbox{if $P > b_i$} \\
        P' \cdot \Delta^x_{b_i,P} + \Delta^y_{a_i,P}  & \quad  \mbox{if $P \in [a_i,b_i]$}.
    \end{cases}
\end{align}

This is linear in $P'$ and that $v_h^{(3)}(1,B_i,P,P) = v_p^{(3)}(1,B_i,P)$. To show that this is tangent to $v_p^{(3)}(1,B_i,P')$ at $P' = P$, it suffices to consider where the initial price $P$ lies relative to $a_i\geq b_i$. In all cases though, from the above we see that the gradient of $v_h^{(3)}(1,B_i,P,P')$ as a function of $P'$ matches up with the gradient of $v_p^{(3)}(1,B_i,P')$ at $P'=P$, hence $v_h^{(3)}(1,B_i,P,P')$ is indeed tangent at $P' = P$ as desired. 
\end{proof}

\section{Example Uniswap Dynamics}
\label{appendix:example-uniswap-dynamics}
\setcounter{equation}{0}

\subsection{Example v2 Dynamics}
\label{appendix:example-v2-dynamics}
In all that follows, we suppose that a Uniswap v2 contract has been set up to trade between token $A$ and token $B$. In addition, we assume that the trade fee rate is given by $\gamma = 0.5$.

\paragraph{Initial Liquidity}
A first LP, denoted by $LP_1$ provides a token bundle given by $(x,y) = (10,10)$ to initialize the contract's bundle state, which must necessarily also be $(10,10)$. The corresponding liquidity-price state of this contract is thus given by $(10,1)$ and we denote the liquidity units owned by this LP with $L_1 = 10$. 

\paragraph{Moving Price from $P=1$ to $P=\frac{1}{4}$}
A trader sends $\Delta x = 20$ units of token $A$ to the contract. Of this, $\gamma \Delta x = 10$ is skimmed for liquidity providers, and since $LP_1$ is the only provider, they receive the entirety of this amount. The remaining $(1-\gamma)\Delta x = 10$ is used to trade with the contract. The token bundle state of the contract changes to $(20,5)$, with a corresponding liquidity-price state of $(10,\frac{1}{4})$. This implies that the trader receives $-\Delta y = 10 - 5 = 5$ units of $B$ tokens in return for the $\Delta x = 20$ units of $A$ tokens they sent. 

\paragraph{$LP_2$ enters the contract}
A new liquidity provider, $LP_2$ wishes to provide $L_2 = 40$ units of liquidity given the current contract price $P = \frac{1}{4}$. To do so, they must send a token bundle consisting of $\cV^{(2)}(L_2,P) = (80,20)$ to the contract. Upon doing so, the contract's token-bundle state becomes $(100,25)$ and the liquidity-price state becomes $(50,\frac{1}{4}) = (L_1+L_2,P)$. 

\paragraph{Moving Price from $P = \frac{1}{4}$ to $P = 25$}
Suppose a trader sends $\Delta y = 400$ units of $B$ tokens to the contract. Of this $\gamma \Delta y = 200$ are skimmed for liquidity providers. $LP_1$ receives $\frac{L_1}{L_1+L_2}200 = 40$ $B$ tokens and $LP_2$ receives the remaining $160$ $B$ tokens. The remaining $(1-\gamma)\Delta y = 100$ $B$ tokens are used to move the token-bundle state of the contract along the reserve curve $\cR^{(2)}(50)$. The token-bundle state of the contract changes to $(10,250)$, with a corresponding liquidity-price state of $(50,25)$. This implies that the trader receives $-\Delta x = 50 - 10 = 40$ $A$ tokens in return for the $\Delta y = 200$ $B$ tokens they sent the contract.

\paragraph{$LP_1$ exits with their Liquidity}
We recall that $LP_1$ has $L_1 = 10$ units of liquidity in the contract. Suppose they remove this liquidity at  price $P = 25$. This means that they receive a token bundle consisting of $\cV^{(2)}(L_1,P) = (2,50)$. Notice that at the contract price $P = 25$, the token $B$ worth of this bundle is $\cV^{(2)}((2,50),25) = 100$ $B$ tokens. On the other hand, the initial bundle they had provided to create this allocation (when price was $P = 1$) consisted of $(10,10)$. The token $B$ worth of this bundle at the given price $P = 25$ is given by $\cB((10,10),25) = 260$. In this example, $LP_1$ has suffered an impermanent loss of $260-25 = 245$ units of $B$ tokens.

\subsection{Example v3 Dynamics}
\label{appendix:example-v3-dynamics}

Hopefully this example is enough to walk readers through the intricacies of liquidity provision, fees, price dynamics, etc. In what follows we will follow closely the acitivity of 3 liquidity providers as they provide and remove liquidity at different price intervals. We will also see what fees they accrue and how the value of their assets change as traders interact with the system. In all that follows, we suppose that the contract fee rate is given by $\gamma = 0.5$.

\paragraph{Initial Liquidity.} Our first liquidity provider, which we denote $LP_1$, provides $L_1 = 60$ units of liquidity over the price range $I_1 = [a_1,b_1] = [1/16,16]$. Furthermore, we suppose that they have allocated a bundle at price $P = 1$. This means that they have sent the contract the following bundle:
$$
\cV^{(3)}(L_1,a_1,b_1,P) = (45,45).
$$
Since $LP_1$ is the only liquidity provider and $P \in [a_1,b_1]$, it follows that they are the only active LP in the system. Furthermore, the minimal active price interval, $[a^*,b^*]$ is also trivially $[a_1,b_1]$. Similarly, $L = L_1$ is the total amount of active liquidity, and the contract's active bundle is the same as $LP_1$'s active bundle, which as we have seen is $(45,45)$. 

\paragraph{Moving Contract Price from $P=1$ to $P=9$.}
A trader sends $\Delta y = 240$ units of token $B$ to the contract. Given the fee rate, $\gamma$, it follows that $\gamma 240 = 120$ units of token $B$ are proportionally allocated to all active liquidity providers, which in this scenario is only $LP_1$. The remaining 120 units of token $B$ are to be used for trading via the v3 reserve curve given by $\cR^{(3)}(L,a^*,b^*)$. The resulting active bundle of the system changes from $(45,45)$ to $(5,165)$, where the latter can be verified to also lie on $\cR^{(3)}(L,a^*,b^*)$. In fact, the corresponding point on the virtual reserve curve is $\phi_{a^*,b^*}(5,165) = (20,180)$, for which it is simple to verify that the contract price is indeed $P=9$. In concrete terms, this means that the trader received $45 - 5 = 40$ $A$ tokens in exchange for the $240$ $B$ tokens they sent to trade and that $LP_1$ received $120$ $B$ tokens in fees for their liquidity provision.

\paragraph{$LP_2$ enters the contract.}
A new liquidity provider, which we dub $LP_2$, wishes to provide $L_2 = 120$ units of liquidity over the price range $I_2 = [a_2,b_2] = [1/25,4]$. Given the fact that the current contract price is $P = 9$, which is above their price interval, this means that they need to provide a token bundle which only consists of $B$ tokens. More specifically, they send the following bundle:
$$
\cV^{(3)}(L_2,a_2,b_2,P) = (0,216)
$$
which has 216 $B$ tokens to be exact. Notice however that the minimal active price interval is now given by $[a^*,b^*] = [4,16]$, for at a price of $P= 4$, $LP_2$'s liquidity becomes active. At the current price $P = 9$, only $LP_1$ has active liquidity, which in turn implies that the total active liquidity in the contract is given by $L = L_1$. 
At the same time, the contract's active bundle is now different from $LP_1$'s overall assets in the pool for the minimal active price interval has changed. The active bundle of the contract is given by $\cV^{(3)}(L,a^*,b^*,9) = (5,60)$, all of which belongs to $LP_1$.  

\paragraph{Moving Contract Price from $P = 9$ to $P = 1/16$.}
In what follows, let us assume that a trader has sent $\Delta x = 1280$ units of token $A$ to the contract. As we will see shortly, this will indeed move the contract price from $P = 9$ to $P = 1/16$, however such a price movement must necessarily involve a change in active liquidity, for the price movement traverses the price $P = 4$, where the $L_2$ units of $[a_2,b_2]$-liquidity of $LP_2$ become active. For this reason, we break up the $\Delta x$ into two smaller trades: $\Delta x_\alpha = 20$ and $\Delta x_\beta = 1260$, each consisting of $A$ tokens to be sent to the contract by the trader.

For the first trade, we recall that the active price interval is given by $[a^*,b^*]$ and that the total active liquidity in the contract is $L = L_1$. Furthermore, the active bundle of the contract is also given by $(5,60) \in \cR^{(3)}(L,a^*,b^*)$. Of the $\Delta x_\alpha = 20$ $A$ tokens sent for trading, a $\gamma$ portion is accrued as fees. This amounts to $10$ $A$ tokens, which go entirely to $LP_1$. The remaining $(1-\gamma)\Delta x_\alpha = 10$ $A$ tokens are used for trading along $\cR^{(3)}(L,a^*,b^*)$. This changes the active contract bundle to $(15,0)$, which means that the trader receives 60 $B$ tokens for such an exchange. We also notice that $(15,0) \in \cV^{(3)}(L,a^*,b^*)$ corresponds to the bundle $\phi_{a^*,b^*}(15,0) = (30,120) \in \cR^{(2)}(L)$ on the virtual reserve curve. This bundle can easily be verified to exhibit a contract price of $4 = a^*$, which means that to continue trading the remaining $\Delta x_\beta = 1260$ $A$ tokens, the contract must update its set of active $LP$s. 

As we continue trading from price $P=4$, it is straightforward to see that the new active price interval is given by $[a^*,b^*] = [1/16,4]$, and that both $LP_1$ and $LP_2$ are active liquidity providers. This in turn means that $L = L_1 + L_2$ is the new total active liquidity in the contract. In addition, contract's active bundle can be expressed by:
$$
\cV^{(3)}(L,a^*,b^*,P) = (0,315).
$$
This bundle  is the aggregation of the active bundles of $LP_1$ and $LP_2$, which are each given by $\cV^{(3)}(L_1,a^*,b^*,P) = (0,105)$ and $\cV^{(3)}(L_2,a^*,b^*,P) = (0,210)$ respectively. Now we can return to the question of trading $\Delta x_\beta = 1260$ $A$ tokens. As before, $\gamma \Delta x_\beta = 630$ is taken as fees for active $LP$s. Notice however, that in this case we have 2 active $LP$s, hence this quantity must be split proportionally amongst them. $LP_1$ receives $\frac{L_1}{L} 630 = 210$ $A$ tokens and $LP_2$ receives the remaining 420 $A$ tokens as fees. Now let us consider using the remaining 630 $A$ tokens to change the contract bundle along the v3 reserve curve, $\cR^{(3)}(L,a^*,b^*)$. We can quickly verify that indeed $(630,0) \in \cR^{(3)}(L,a^*,b^*)$, and has a corresponding bundle on the virtual reserve curve given by $\phi_{a^*,b^*}(630,0) = (720,45)$ which can quickly be verified to exhibit a contract price of $P = 1/16$. As the contract bundle went from $(0,315)$ to $(630,0)$ in this second trade, we see that the trader receives $315$ $B$ tokens in exchange for the $\Delta x_\beta$ $A$ tokens they sent for trading. In summary, the trader sent $\Delta x = \Delta x_\alpha + \Delta x_\beta = 1280$ units of $A$ tokens and received $60 + 315 = 375$ $B$ tokens in exchange.

\paragraph{$LP_3$ enters the contract.}
A new liquidity provider, $LP_3$ wishes to provide $L_3 = 180$ units of liquidity over the range $I_3 = [a_3,b_3] = [1/9,36]$. We recall the current contract price is $P = 1/16$, which is below the $LP$'s desired price range. For this reason, $LP_3$ will have to deposit a bundle consisting entirely of $A$ tokens to establish this liquidity allocation. More specifically, they will provide the contract with the following bundle:
$$
\cV^{(3)}(L_3,a_3,b_3,P) = (510,0).
$$
Notice that at the current price $P = 1/16$, the liquidity of $LP_3$ is not active. Indeed it does not become active until a price of $P = 1/9$ is reached. For this reason, the contract has a new minimal active price interval given by $[a^*,b^*] = [1/16,1/9]$. It follows that the contract still has a total active liquidity given by $L = L_1 + L_2$. At the same time, the contract's active bundle is also given by 
$$
\cV^{(3)}(L,a^*,b^*) = (180,0).
$$
This bundle is also the aggregation of the active bundles of $LP_1$ and $LP_2$ over $[a^*,b^*]$, which are given by $\cV^{(3)}(L_1,a^*,b^*,P) = (60,05)$ and $\cV^{(3)}(L_2,a^*,b^*,P) = (120,0)$ respectively. 

\paragraph{Moving Contract Price from $P = 1/16$ to $P = 1$.}
In what follows, we will assume that a trader has sent $\Delta y = 510$ units of token $B$ to the contract to exchange for $A$ tokens. As we will see shortly, this will indeed move the contract price from $P - 1/16$ to $P=1$. Such a price movement however, must necessarily involve changes of active liquidity, for at price $P = 1/9$ the liquidity of $LP_3$ becomes active. For this reason, similar to before, we break up $\Delta y$ into tow smaller trades: $\Delta y_\alpha = 30$ and $\Delta y_\beta  = 480$, each consisting of $B$ tokens to be sent to the contract by the trader. 

For the first trade, we recall that the active price interval is given by $[a^*,b^*]$ and that the total active liquidity in the contract is given by $L = L_1 + L_2$. In addition, as we've previously seen, the active bundle of the contract is given by $(180,0) \in \cR^{(3)}(L,a^*,b^*)$. Of the $\Delta y_\alpha = 30$ tokens, $\gamma \Delta y_\alpha = 15$ are skimmed for fees, which are shared between $LP_1$ and $LP_2$. $LP_1$ receives $\frac{L_1}{L}15 = 5$ $B$ tokens and $LP_2$ receives the remaining 10 $B$ tokens. After such fees are leveed, the remaining $(1-\gamma) \Delta y_\alpha = 15$ $B$ tokens are used to move active bundle along the v3 reserve curve, $\cR^{(3)}(L,a^*,b^*)$. The resulting active bundle is thus $(0,15) \in \cR^{(3)}(L,a^*,b^*)$, which corresponds to the bundle $\phi_{a^*,b^*}(0,15) = (540,60) \in cR^{(2)}(L)$ on the virtual reserve curve. This bundle can easily be verified to exhibit a contract price of $P = 1/9 = a^*$, which means that to continue trading the remaining $\Delta y_\beta = 480$ units of token $B$, the contract must update its set of active $LP$s. 

As we continue trading from $P = 1/9$, it is straightforward to see that the new active price interval is given by $[a^*,b^*] = [1/9,4]$. In this price range all $LP$s are active, hence the total active liquidity is given by $L = L_1 + L_2 + L_3 = 360$. As before, we can also compute the contract's new active bundle as:
$$
\cV^{(3)}(L,a^*,b^*) = (900,0).
$$
This bundle is made of the smaller active bundles of each $LP$, which are given by $\cV^{(3)}(L_1,a^*,b^*) = (150,0)$, $\cV^{(3)}(L_2,a^*,b^*) = (300,0)$, and $\cV^{(3)}(L_3,a^*,b^*) = (450,0)$ respectively for $LP_1$, $LP_2$ and $LP_3$. Now let us focus on the $\Delta y_\beta = 480$ units of token $B$ sent to the contract by the trader. As before, we first skim provider fees, which amount to $\gamma \Delta y_\beta = 240$ $B$ tokens to be split proportionally amongst active $LP$s. In this case, $LP_1$ receives $\frac{L_1}{L}240 = 40$ $B$ tokens, $LP_2$ receives $\frac{L_2}{L} 240 = 80$ $B$ tokens, and $LP_3$ receives the remaining $80$ $B$ tokens. Now let us consider using the remaining $240$ $B$ tokens to change the contract bundle along the v3 reserve curve, $\cR^{(3)}(L,a^*,b^*)$. We can quickly verify that indeed $(180,240) \in \cR^{(3)}(L,a^*,b^*)$, which has a corresponding token bundle on the virtual reserve curve given by $\phi_{a^*,b^*}(180,240) = (360,360)$, which exhibits a price of $P=1$ as desired. The contract bundle changed from $(900,0)$ to $(180,240)$, meaning that the trader receives 720 $A$ tokens in exchange for the $\Delta y_{\beta}$ sent to the contract. Putting everything together, the trader receives 900 $A$ tokens in exchange for 510 $B$ tokens. 

\paragraph{$LP_1$ exits with their liquidity.}
In what follows, we assume that $LP_1$ removes their liquidity from the contract. Before doing so, we recall that after the previous trade, the contract price is given by $P =1$. In addition, the active price interval is still $[a^*,b^*] = [1/9,4]$, and all $LP$s contribute to the active liquidity and active price bundle of the contract. We recall that $LP_1$ has $L_1 = 60$ units of $[a_1,b_1]$-liquidity for $a_1 = 1/16$ and $b_1 = 16$. By removing their liquidity from the contract, they receive a bundle given by:
$$
\cV^{(3)}(L_1,a_1,b_1,P) = (45,45).
$$
Notice that this is the same as the bundle they used to open their liquidity allocation. This is due to the fact that they are withdrawing their liquidity when the contract price is given by $P = 1$, which is the same as the contract price when they created their position. In this example however, removing $LP_1$'s liquidity has no effect on the active price range, What does change however is the active liquidity in the contract, and consequently, the active token bundle of the contract. The contract's active liquidity comes from $LP_2$ and $LP_3$, hence the active liquidity is $L = L_2 + L_3 = 300$. In addition, the active bundle is once again given by:
$$
\cV^{(3)}(L,a^*,b^*) = (150,200).
$$
This can be further split into the active bundle of each $LP$, $LP_2$ has an active bundle given by $\cV^{(3)}(L_2,a^*,b^*) = (60,80)$ and $\cV^{(3)}(L_3,a^*,b^*) = (90,120)$ respectively for $LP_2$ and $LP_3$.

\paragraph{Moving Contract Price from $P = 1$ to $P = 1/4$.}
A trader sends $\Delta x = 600$ $A$ tokens to the contract. $\gamma \Delta x = 300$ $A$ tokens are skimmed as fees for $LP_2$ and $LP_3$ who are both active. $LP_1$ earns $\frac{L_2}{L}300 = 120$ $A$ tokens and $LP_2$ gets the remaining $180$ $A$ tokens as fees. The remaining $(1-\gamma) \Delta x = 300$ $A$ tokens are used to trade against the v3 reserve curve. At price $P = 1$, the active token bundle of the contract is given by $(150,200)$, however the trade shifts the active token bundle to $(450,50) \in \cR^{(3)}(L,a^*,b^*)$, which has a corresponding bundle on the virtual reserve curve given by $\phi_{a^*,b^*}(450,50) = (600,150)$, which clearly has a contract price of $P = 1/4$ as desired. Overall, the trader sent $\Delta x = 600$ $A$ tokens to the contract and received $150$ $B$ tokens in exchange.  

\paragraph{$LP_2$ exits with their liquidity}
After the previous price move, the contract still has the same active price interval $[a^*,b^*] = [1/9,4]$, and the contract's active bundle is $(450,50)$. Furthermore, the current contract price is given by $P = 1/4$. We assume that $LP_2$ wishes to remove their $L_2$ units of $[a_2,b_2]$-liquidity, where $a_2 = 1/25$, and $b_2 = 4$. In this case, $LP_2$ receives the following bundle upon extracting their liquidity:
$$
\cV^{(3)}(L_2,a_2,b_2,P) = (180,36).
$$
Consequently, the active price interval changes to $[a^*,b^*] = [1/9,36]$, which is  the interval over which $LP_3$ lent liquidity, for they are the only remaining $LP$. As a final point, note that the bundle received by $LP_2$ is different from what they used to create their liquidity allocation, which was $(0,216)$. At the current price $P = 1/4$, the bundle $(180,36)$ is worth $81$ $B$ tokens, which is much less than the original $216$ $B$ tokens used to create the liquidity allocation. This is the impermanent loss suffered by $LP$s as the contract price fluctuates.

\section{Empirical Findings on Liquidity Provision}
\label{appendix:empirical-liquidity-dist}
\setcounter{equation}{0}

We focus on providing empirical evidence in favor of our {\it bucket coverage assumption} from Section \ref{sec:trader-gas-fees}. To this end, Figure \ref{fig:emp-bucket-coverage} provides a typical snapshot of locked liquidity in a USDC/ETH v3 pool taken on February 6, 2022. The snapshot was reconstructed from transactions on the Ethereum blockchain. Each bar in the image corresponds to a bucket, and the height of each bar represents the amount of liquidity locked in that bucket. As we can see, buckets near the contract price, which is the bar in red, have different amounts of liquidity. Specific liquidity amounts are shown in Table \ref{tab:emp-bucket-coverage}, where the index of the bucket in the first column represents the bucket position relative to the contract price (which is in bucket $B_0$). This typical difference in liquidity values for buckets around the contract price in turn implies that said buckets' endpoints must in turn be active, providing credence to our bucket coverage assumption. 
\begin{figure}[t]
     \centering
     \begin{subfigure}[b]{\textwidth}
         \centering
         \includegraphics[width=\textwidth]{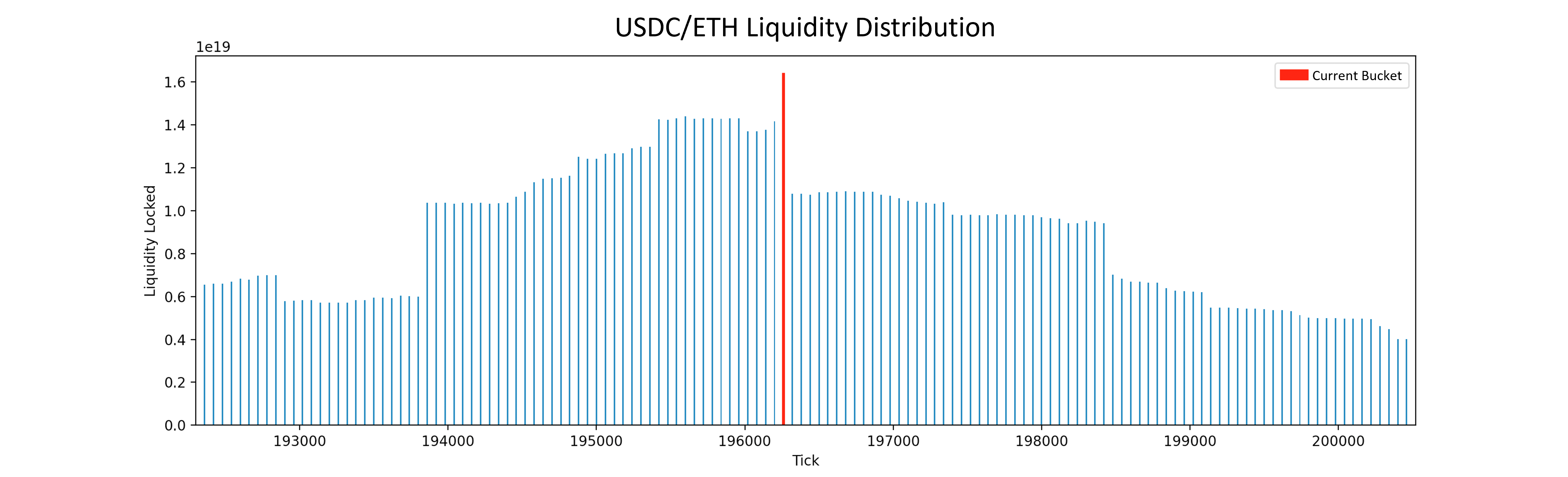}
     \end{subfigure}
     \caption{Uniswap v3 Liquidity Distribution for the top USDC/ETH pool on February 6th, 2022. Each bar represents a bucket, and the height of the bar represents the amount of liquidity locked in the bucket around the current contract price, which is given by the bar in red. As we can see, most buckets near the price have different amounts of locked liquidity, which forcibly means that their endpoints are active ticks, in line with our bucket coverage assumption.}
     \label{fig:emp-bucket-coverage}
\end{figure}

\begin{table}
\begin{center}
\caption {Locked liquidity around the contract price for the top USDC/ETH pool on February 6th, 2022} \label{tab:emp-bucket-coverage} 
\begin{tabular}{||c | c||} 
 \hline
 Bucket & Units of Liquidity Locked \\ [0.5ex] 
 \hline\hline
 $B_{-5}$ & 14308585389682429965 \\ 
 \hline
 $B_{-4}$ & 13699020740164778906 \\ 
 \hline
 $B_{-3}$ & 13704199306520970180 \\
 \hline
 $B_{-2}$ & 13761114721595607988 \\
 \hline
 $B_{-1}$ & 14169372155386979909 \\
 \hline
 $B_{0}$ & 16389750355923557295 \\
 \hline
 $B_{1}$ & 10795735828654461424  \\
 \hline
 $B_{2}$ & 10786362926145864699  \\
 \hline
 $B_{3}$ & 10730695468843515712  \\
 \hline
 $B_{4}$ & 10847573897531124431  \\
 \hline
 $B_{5}$ & 10865194479196829563  \\ [1ex] 
 \hline
\end{tabular}
\end{center}
\end{table}

\section{Maximum Likelihood Estimation (MLE)}
\label{append:MLE}
\setcounter{equation}{0}

We are given a price sequence given by $\vect{p} = (p_1,\dots,p_{n+1})$ where each $p_i \in \mathbb{R}^+$. As per Section \ref{sec:risk-averse-comp-methods}, we model market price evolution as a geometric binomial random walk with two relevant parameters: $\omega$ and $W$. The parameter $\omega$ represents the multiplicative ratio between discrete prices the random walk evolves over, and $W$ is the bounded bandwidth that limits how far prices can deviate in the random walk conditional on a given price. 

It is important to note that empirical price data given by the sequence $\vect{p}$ need not adhere to a single multiplicative ratio, as said prices are only constrained to lie in $\mathbb{R}^+$. However, as in Proposition \ref{prop:binomial-martingle}, we use the approximation that $x_t = \log_\omega (\frac{p_t}{p_{t+1}}) \sim \cN(2Wp-W,2Wp(1-p))$, with $p = \frac{m+2-\sqrt{m^2+4}}{2m}$, where $m = \log(\omega)$. With this normal approximation, we can then perform a maximum likelihood estimate over the sequence $x_1,\dots,x_n$, where each $x_t = \log_\omega (\frac{p_t}{p_{t+1}})$ is assumed to be i.i.d normal with mean $\mu = 2Wp-W$ and variance $\sigma^2 = 2Wp(1-p)$. With this in hand, we can express the likelihood of a choice of $\omega$ and $W$ for the given price sequence $\vect{p}$ as follows:

\begin{align}
L(\omega,W \mid \vect{p}) = \mathbb{P} (x_1,\dots,x_n \mid \omega,W)
\end{align}

As is standard we maximize this expression by taking logarithms and minimizing the negative log-likelihood given by:

\begin{align}
NLL(\omega,W \mid \vect{p}) = 
\frac{n}{2} \mbox{log} (2\pi \sigma^2) + \frac{1}{2\sigma^2} \sum^n_{i=1} (x_i - \mu)^2
\end{align}

\begin{lemma}
\label{lemma:closed-form-MLE}
The unique $W^*$ that minimizes $NLL$ for a fixed $\omega$ is given by:

\begin{align}
W^{*} = \frac{\sum_{i=1}^n x^2_i}{n(p(1-p) + \sqrt{p^2(1-p)^2 + \frac{1}{n}\sum x^2_i (2p-1)^2})}
\end{align}
\end{lemma}

\begin{proof}
Going forward, we re-write the negative log-likelihood as a function of $W$:

\begin{align}
f(W) = NLL(\omega,W \mid \vect{p}) = \frac{n}{2} \mbox{log} (aW) + \frac{b}{W} \sum^n_{i=1} (x_i - cW)^2
\end{align}
where $a = 4\pi p(1-p), b = \frac{1}{4p(1-p)}, c = 2p-1$. After some collecting terms, the derivative of $f$ has the following form:
\begin{equation}
     f'(W) = - \frac{\sum_{i=1}^n x^2_i}{4p(1-p)}\frac{1}{W^2} + \frac{n}{2}\frac{1}{W} + \frac{n(2p-1)^2}{4p(1-p)}
\end{equation}
We can solve for $f'(W^*) = 0$ under the constraint given by $W^*>0$. We get a unique closed form solution given by:

\begin{align}
W^{*} = \frac{\sum_{i=1}^n x^2_i}{n(p(1-p) + \sqrt{p^2(1-p)^2 + \frac{1}{n}\sum_{i=1}^n x^2_i (2p-1)^2})}
\end{align}

as desired. Furthermore, We can compute the second derivative of $f$:
\begin{equation}
    f''(W) = \frac{ \left(n \sum_{i=1}^n x_i^2 \right)- n^2p(1-p)W}{2np(1-p)W^3}
\end{equation}

Clearly the second derivative depends on the data (i.e. the values of $x_i$), but if we plug the given $W^*$ into the expression, we obtain:
\begin{equation}
    f''(W^{*}) = \frac{\sum_{i=1}^n x^2_i}{2p(1-p)(W^{*})^3} \left( 1-\frac{1}{1+\sqrt{1+\frac{(2p-1)^2\sum_{i=1}^n x^2_i }{np^2(1-p)^2}}}
    \right) > 0
\end{equation}
Hence $W^*$ is a minimum as desired. 
\end{proof}

\subsection{Empirical Results}
In this section we outline our methods for obtaining relevant $\omega$ and $W$ values to be used in the LP belief profiles from Section \ref{LP-belief-model-assumptions}. 

\paragraph{Baseline Price Sequences.}
As mentioned in Section \ref{LP-belief-model-assumptions}, $\omega$ and $W$ both govern the volatility of the stochastic process dictating market prices relevant to an LP. In this regard, we looked historical prices between Ethereum (ETH) and Bitcoin (BTC), as well as historical prices between Ethereum (ETH) and and Tether (USDT). The former pair was chosen as a low volatility representative, as most cryptocurrencies have prices correlated to Bitcoin. The latter pair was chosen for its relatively higher volatility, for USDT is pegged to the dollar. For each of these price pairs, we obtained a baseline price sequence from \href{https://www.binance.com/en/landing/data}{\url{https://www.binance.com/en/landing/data}} of per-minute prices for the month of February 2020.

\paragraph{Sub-sampling from Price Sequences.}
Let us denote the baseline price sequence for the ETH-BTC pair by $X_1',...,X_T'$, and that of the ETH-USDT pair by $Y_1,...,Y_T'$. Our first step of price-processing involves sampling a collection of prices from each sequence at a desired frequency to then apply MLE estimates from the previous section to obtain $\omega$ and $W$ values. To this end, let $t_0 \in \{1,\dots,T\}$, and $g,k \geq 1$ be integers. Let us focus on the ETH-USDT sequence given by $X_i'$. A choice of $(t_0,g,k)$ implies we take a subsequence from $\{X_i'\}_{i=1}^T$ given by $\{X_{t_j}'\}_{j=1}^k$, where $t_j = t_0 + g\cdot j$. Clearly this implies that we impose the constraint $t_0 + g \cdot k \leq T$. Ultimately, we are interested in optimal liquidity allocations from LPs as per the methodology of Section \ref{sec:riskneutral}. As we focus on LP strategies that consist of allocating liquidity over $T$ time-steps of a stochastic process and subsequently removing said liquidity (and extracting fees), we are interested in sub-sampling over smaller time horizons at smaller values of $g$, as it is in this regime where such LP strategies are reasonable. Though this does imply that we are interested in smaller values of $k$, at the same time, too small of values of $k$ would give us poor MLE estimates as per the previous section, as we would  not have many samples to work with.  

\paragraph{MLE Computation.}
Once we choose values of $(t_0,g,k)$, we obtain a price subsequence as per the exposition above, which for ETH-BTC is given by $\{X_{t_j}'\}_{j=1}^k$, where $t_j = t_0 + g\cdot j$. To proceed, assume a fixed value of $\omega > 1$, and use this to generate the sequence $X_1,\dots,X_{k-1}$, where 
$X_j = \log_\omega \left( \frac{X_{t_{j+1}}'}{X_{t_j}'} \right)$. For such a sequence, we can compute $W^*$ as per the closed form of Lemma \ref{lemma:closed-form-MLE}, and from this obtain $NLL(\omega,W^* \mid X_1,\dots,X_{k-1})$. We can repeat this computation for multiple choices of $\omega$ to find an $\omega^*$ such that the negative log-likelihood is minimized via grid search. In summary, once we fix sub-sampling parameters $(t_0,g,k)$, we can obtain MLE estimates $\omega^*,W^*$ as above. 

\paragraph{Results.}
For both price pairs, we used sub-sampling parameters given by $(t_0,g,k) = (0,1,256)$. The reason for this is that we want to computationally explore the impact of price volatility on LP PnL and trader gas fees given an LP belief profile, hence it is important to maintain these sub-sampling parameters fixed over both price pairs of different real-world volatility. Our resulting MLE estimates were $(\omega^*,W^*) = (1.005,3.0607)$ for ETH-USDT and $(\omega^*,W^*) = (1.005,6.7695)$ for ETH-BTC. For reference, Figure \ref{fig:emp-MLE} contains a visualization of the fit of the MLE normal distribution over empirical log price ratios. Ultimately, LP beliefs are constructed via an approximate Geometric Binomial Random Walk, hence we require integral values of $W$. For this reason, our computational results in Section~\ref{sec:computational-results} use $W$ values from the set $\{3,5,7\}$, where the lower and upper bounds on $W$ are precisely informed by our empirical results of this section.

\begin{figure}
    \centering
    \includegraphics[width=0.45\textwidth]{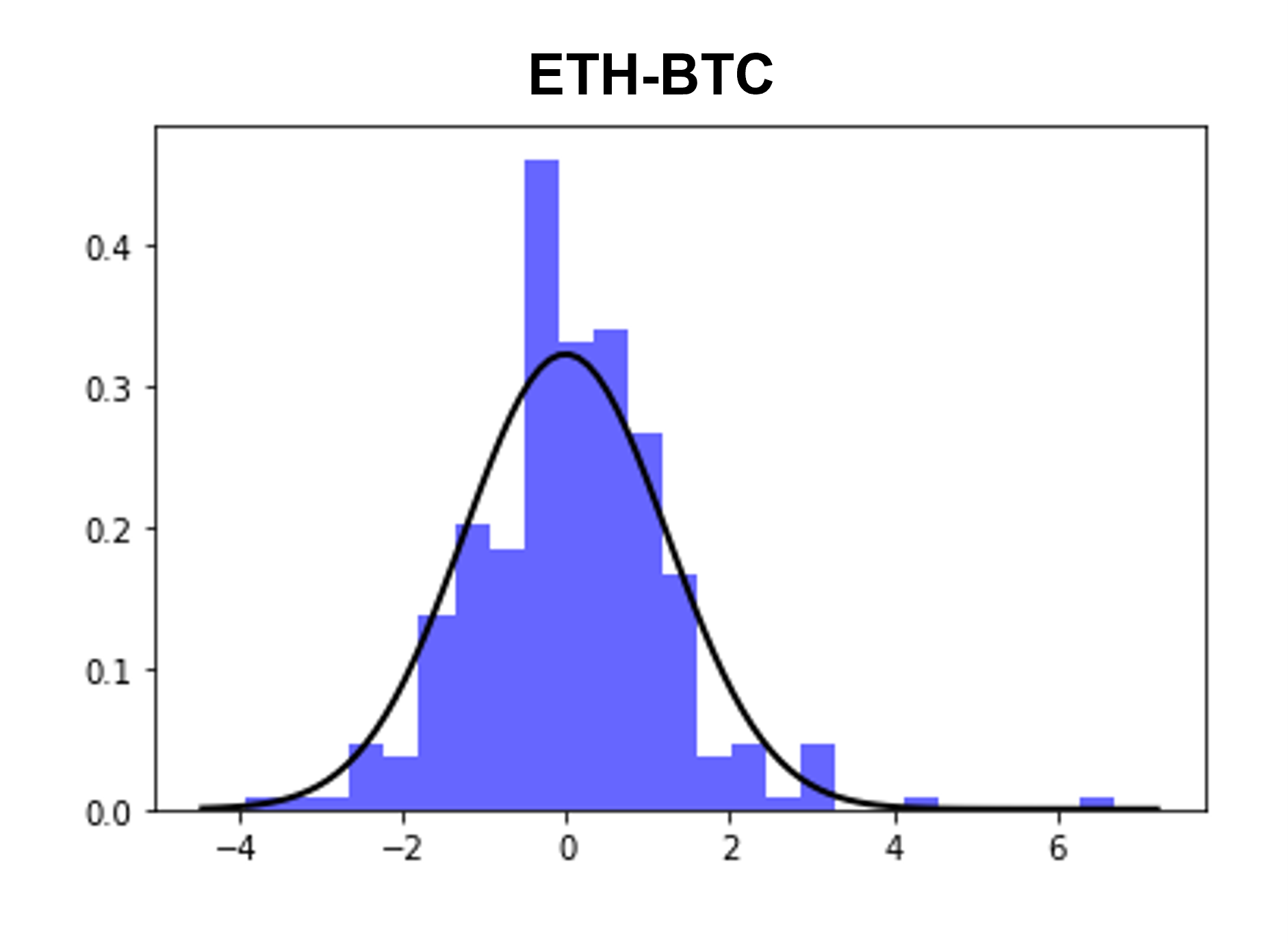}
    \includegraphics[width=0.45\textwidth]{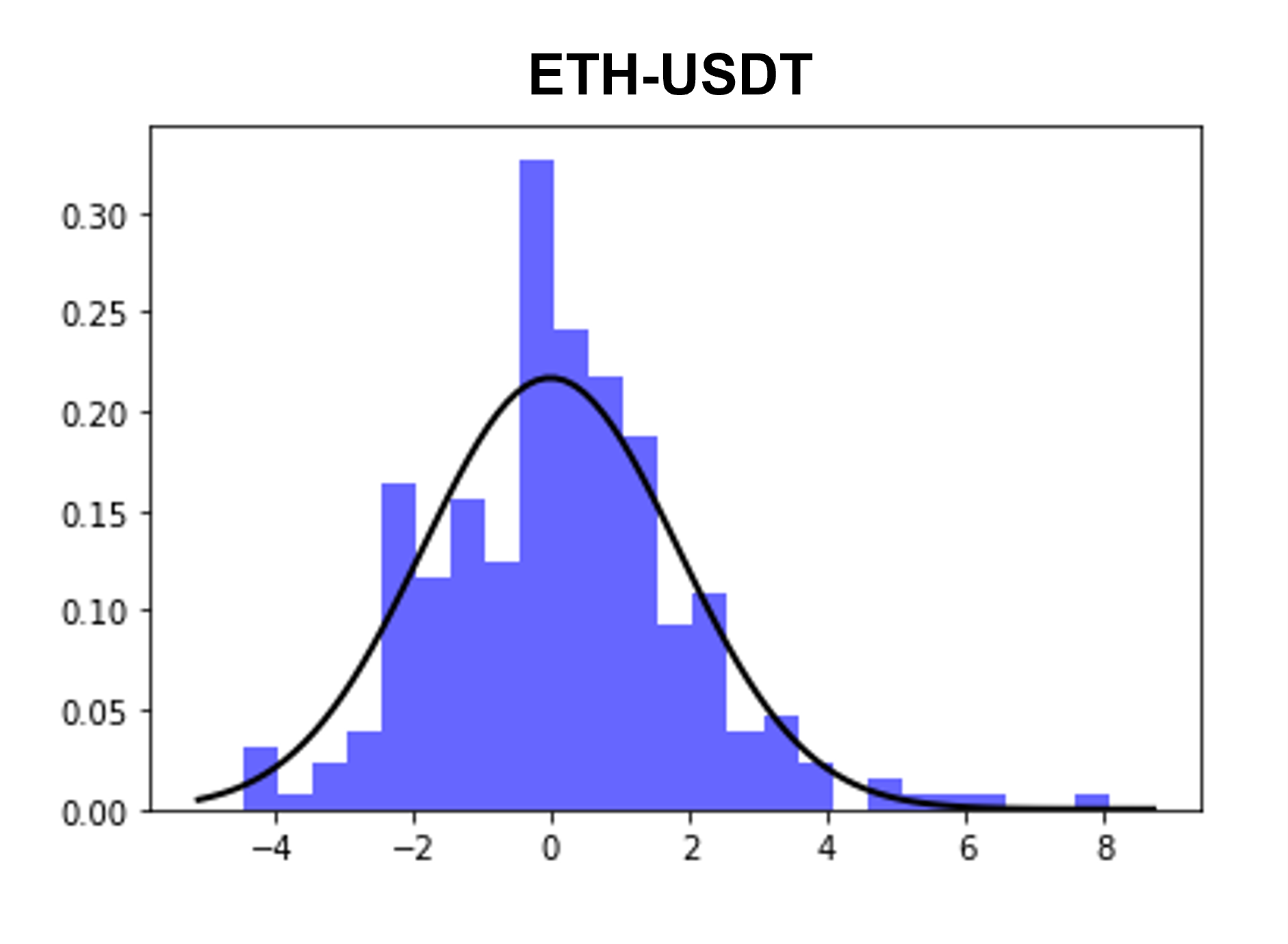}
    \caption{The MLE fit over sub-sampled data from ETH-BTC and ETH-USDT.
    \label{fig:emp-MLE}}
\end{figure}

\end{document}